\documentclass[aps,pra,reprint,twocolumn,superscriptaddress,floatfix,nofootinbib,longbibliography]{revtex4-2}
\usepackage{graphicx,amsmath,amsfonts,amssymb,amsthm,xr}
\usepackage{epsfig,amsmath,amssymb,color,dsfont,upgreek,physics}
\usepackage{mathrsfs}
\usepackage{mathtools}
\usepackage{bbold}
\usepackage{float}

\usepackage[bookmarks=true,colorlinks,linkcolor=OrangeRed,urlcolor=NavyBlue,citecolor=RoyalBlue]{hyperref}
\usepackage[dvipsnames]{xcolor}
\usepackage{orcidlink}
\usepackage{pgfplots}

\usepackage{graphicx}
\usepackage{dcolumn}
\usepackage{bm}
\usepackage{color}

\usepackage{graphicx}

\usepackage[capitalise]{cleveref}
\newtheorem{theorem}{Theorem}
\newtheorem{lemma}[theorem]{Lemma}

\theoremstyle{plain}

\newcommand{\ie}{\emph{i.e.}}

\newcommand{\npf}{N_{\mathrm{PF}}}

\crefname{section}{Sec.}{Secs.}
\crefname{figure}{Fig.}{Figs.}
\crefname{appendix}{App.}{Apps.}
\crefname{equation}{Eq.}{Eqs.}
\crefname{definition}{Definition}{Defs.}
\crefname{theorem}{Theorem}{Thms.}
\crefname{corollary}{Corollary}{Cors.}
\crefname{remark}{Remark}{Remarks}
\crefname{proposition}{Prop.}{Props.}

\begin{document}
\title{Obtaining continuum physics from dynamical simulations of Hamiltonian lattice gauge theories
}

\author{Christopher F. Kane$^{\orcidlink{0000-0002-2020-8971}}$}
\email{cfkane24@umd.edu}
\affiliation{Maryland Center for Fundamental Physics and Department of Physics, University of Maryland, College Park, MD 20742, USA}
\affiliation{Joint Center for Quantum Information and Computer Science, University of Maryland, College Park,
Maryland 20742, USA.}

\author{Siddharth Hariprakash$^{\orcidlink{0009-0000-3524-4724}}$}
\email{siddharth_hari@berkeley.edu}
\affiliation{Physics Division, Lawrence Berkeley National Laboratory, Berkeley, California 94720, USA} 
\affiliation{Leinweber Institute for Theoretical Physics and Department of Physics, University of California, Berkeley, California 94720, USA}

\author{Christian W.~Bauer$^{\orcidlink{0000-0001-9820-5810}}$}
\email{cwbauer@lbl.gov}
\affiliation{Physics Division, Lawrence Berkeley National Laboratory, Berkeley, California 94720, USA}
\affiliation{Leinweber Institute for Theoretical Physics and Department of Physics, University of California, Berkeley, California 94720, USA}

\date{\today}

\begin{abstract}

Taking the continuum limit is essential for extracting physical observables from quantum simulations of lattice gauge theories. 
Achieving the correct continuum limit requires careful control of all systematic uncertainties, including those arising from approximate implementations of the time evolution operator. 
In this work, we review existing approaches based on renormalization techniques, and point out their limitations.
To overcome these limitations, we introduce a new general framework—the Statistically-Bounded Time Evolution (SBTE) protocol—for rigorously controlling the impact of approximate time evolution on the continuum limit. 
The central insight is that, since exact time evolution introduces no UV divergences, errors from approximate evolution can be treated as a source of systematic uncertainty that can be neglected if reduced below the working statistical uncertainty. 
We show that, using the SBTE protocol, which prescribes driving the approximate time evolution error below the working statistical uncertainty, leads to a simplified renormalization procedure. 
Furthermore, we show that, due to the existence of rigorous error bounds, one can guarantee \emph{a priori} that such errors are negligible and do not affect the continuum limit.
Ultimately, our protocol lays the foundation for performing systematic and fair comparisons between different simulation algorithms for lattice gauge theory simulations.

\end{abstract}

\maketitle

\section{Introduction}

The simulation of lattice gauge theories (LGTs) is a central challenge toward understanding strongly correlated quantum systems, with profound implications for high-energy physics, nuclear physics, and condensed matter systems.
Quantum computers hold the promise of enabling simulations of LGTs that are currently beyond the reach of classical simulations methods.
This advantage was first shown for scalar field theories in the pioneering work of Jordan, Lee and Preskill~\cite{Jordan:2012xnu}.
Since then, there has been remarkable progress towards simulating LGTs and lattice field theories more generally~\cite{Bauer:2022hpo,Bauer:2023qgm,DiMeglio:2023nsa}, including the study of U(1)~\cite{Banerjee:2012pg,Hauke:2013jga,Zohar:2013zla,Kuhn:2014rha,Kasper:2015cca,Zohar:2015hwa,Martinez:2016yna,Yang:2016hjn,Kokail:2018eiw,Klco:2018kyo,Lu:2018pjk,Kaplan:2018vnj,Mil:2019pbt,Davoudi:2019bhy,Surace:2019dtp,Haase:2020kaj,Luo:2019vmi,Shaw:2020udc,Yang:2020yer,Ott:2020ycj,Paulson:2020zjd,Nguyen:2021hyk,Zhou:2021kdl,Riechert:2021ink,Bauer:2021gek,Kane:2022ejm,Grabowska:2022uos,zhang2023observation,Farrell:2023fgd,Nagano:2023uaq, Gustafson:2024bww, Crane:2024tlj}, ${\rm SU}(2)$~\cite{Zohar:2012xf,Stannigel:2013zka,Mezzacapo:2015bra,Mathur:2015wba,Raychowdhury:2018osk,Raychowdhury:2019iki,Klco:2019evd,Dasgupta:2020itb,Davoudi:2020yln,Atas:2021ext,ARahman:2021ktn,Osborne:2022jxq,halimeh2022gauge,ARahman:2022tkr,zache2023quantum,Alexandru:2023qzd,DAndrea:2023qnr,Turro:2024pxu,Grabowska:2024emw,Burbano:2024uvn} 
and ${\rm SU}(3)$~\cite{Anishetty:2009nh,Alexandru:2019nsa,Ciavarella:2021nmj,Farrell:2022wyt,Farrell:2022vyh,Atas:2022dqm,Ciavarella:2021lel,hayata2023qdeformedformulationhamiltonian,Farrell:2024fit,Ciavarella:2024fzw, Gustafson:2024kym, Balaji:2025afl, Ciavarella:2025bsg} gauge theories, as well as work towards simulating light-front formulations~\cite{Kreshchuk:2020kcz,kreshchuk2022quantum, kreshchuk2023simulatingscatteringcompositeparticles}, quantum parton showers~\cite{Bauer:2019qxa,Bauer:2023ujy,Chigusa:2022act}, improved lattice Hamiltonians~\cite{Carena:2022kpg, Gustafson:2023aai, Ciavarella:2023mfc, Illa:2025dou, Illa:2025njz}, lattice formulations of nuclear effective field theories~\cite{Roggero:2019myu, Watson:2023oov}, understanding finite volume errors~\cite{Burbano:2025pef},
and much more~\cite{Tagliacozzo:2012df,Bazavov:2015kka,Jordan:2017lea,Gonzalez-Cuadra:2017lvz,Gorg:2018xyc,Lamm:2019bik,Zohar:2019ygc,Buser:2020cvn,Barata:2020jtq,Stryker:2021asy,Davoudi:2024osg, Mueller:2024mmk, Davoudi:2025tbi,Surace:2024bht, De:2024smi,Davoudi:2024wyv,Bennewitz:2024ixi, Mueller:2022xbg, Kurkcuoglu:2024cfv, Carena:2024dzu, Ciavarella:2024lsp, Ciavarella:2025zqf, Ingoldby:2025bdb, Gupta:2025xti,Bauer:2025nzf,Alvi:2022fkk}.
There has also been progress in algorithmic studies~\cite{Kan:2021xfc,Alexandru:2022son,tong2022provably,Davoudi:2022xmb,Kane:2023jdo,Hariprakash:2023tla,Rhodes:2024zbr,Du:2024ixj,Li:2024lrl,Kane:2024odt, Murairi:2024xpc, Assi:2024pdn,  Lamm:2024jnl, Gomes:2024tup, Hardy:2024ric, Simon:2025pbo, Decker:2025jts, Davoudi:2025rdv} and tensor network studies~\cite{Pichler:2015yqa,Banuls:2017ena,Banuls:2018jag,Banuls:2022vxp, Mathew:2025fim, Belyansky:2023rgh}.
Despite this progress, there is still progress to be made to fully understand what combination of Hamiltonian formulations, quantum algorithms, quantum computing platforms, etc., will result in the fastest path to quantum advantage.

One crucial ingredient  in any LGT calculation is taking the continuum limit, \ie, taking the lattice spacing $a \to 0$.
Doing so requires several steps, including tuning the bare parameters in the Hamiltonian as a function of the lattice spacing $a$, renormalizing the operators being studied, and determining the lattice spacing through a scale setting procedure.
Each of these steps in principle requires calculating a physical observable, which implies that the cost of each step could be comparable to the cost of state preparation and time evolution.
While strategies for reducing the quantum cost of renormalization through the use of classical Euclidean LGT simulations~\cite{Clemente:2022cka, Crippa:2024cqr, Carena:2021ltu, Funcke:2022opx, Gross:2025qae} and perturbation theory~\cite{Carena:2022hpz} have been proposed, the full cost of calculating continuum physics using quantum simulations of LGTs is not fully understood. 

One important question to ask is how error from approximate time evolution impacts the continuum limit.
This issue was studied in Ref.~\cite{Carena:2021ltu} in the context of simulations using product formulas (PFs).
Using the Baker–Campbell–Hausdorff formula, it was shown that the effective Hamiltonian one actually time evolves according to with PF-based time evolution is closely related to the transfer matrix in Euclidean LGTs.
Through this connection, approximate time evolution can be interpreted as exact evolution on a temporal lattice with spacing $\delta_t$, which is proportional to the Euclidean temporal spacing $a_0$. 
The associated error from finite $a_t$ can then be removed by extrapolating $\delta_t \to 0$ (or fixing the anisotropy $\xi \equiv a/\delta_t$ and taking $a \to 0$), following a procedure analogous to the continuum limit $a \to 0$.

While this is a valid approach, it is currently applicable only to PF-based simulations and does not extend to the broader class of post-Trotter simulation algorithms~\cite{Childs:2012gwh, Low:2016sck, Low:2016znh, PhysRevLett.123.070503, Li:2016vmf, Haah:2018ekc, Yuan:2018jdl, Motlagh:2023oqc, Zeng:2022pim, Watson:2024yqs, Watson:2024dvw, Chakraborty:2025sry}.
This implies that it is currently not possible to perform direct comparisons of the full cost of taking the continuum limit when using different time evolution algorithms.
This is particularly important given the recent efforts for determining which formulations and algorithms which result in the smallest gate counts~\cite{Rhodes:2024zbr, Kan:2021xfc, Ciavarella:2025bsg, Balaji:2025afl, Davoudi:2022xmb}.
In order to broaden the class of simulation algorithms at our disposal, including those not yet developed, an alternative approach is necessary. 

In this work, we develop an alternative approach to controlling the impact of approximate time evolution on the continuum limit.
The key insight is that, while approximate time evolution may be interpreted through the lens of renormalization, no additional UV divergences are introduced in the limit of exact time evolution.
This suggests that, rather than invoking the full machinery of renormalization, as in Ref.~\cite{Carena:2021ltu}, the error from approximate time evolution can simply be treated as a systematic uncertainty.
As long as this error remains negligible compared to the statistical precision of the simulation, its effect on the continuum limit is negligible and can be ignored.
In particular, we will argue that this applies to any unphysical behavior introduced from approximate time evolution, and that no problem analogous to the notorious fermion doubling problem is possible in the limit of exact time evolution.
We now state the central question we address in this work:

\vspace{1em} 
\emph{For a general time evolution algorithm, what is the additional computational cost required to ensure that the time evolution error has a negligible affect on the continuum limit relative to the working statistical precision?}
\vspace{1em} 

In this work, we provide a general framework that can be used to answer this question for any time evolution algorithm \emph{a priori}, that is, before performing any simulations.
We show that, as long as the systematic error incurred for each step in the renormalization procedure is below the working statistical precision, the need to treat the approximate time evolution error using renormalization is no longer necessary, which results in significant simplifications in taking the continuum limit.
We coin this proposal Statistically-Bounded Time Evolution (SBTE) Protocol.

To demonstrate our procedure, we apply our formalism to simulations to two common simulation algorithms, namely product formulas and those based on Quantum Signal Processing (QSP). 
We find that, although the additional cost required to reduce the time evolution error below the working statistical precision can be significant when using Trotter-based methods, it becomes negligible with QSP due to its provably optimal scaling with the approximation error. 
While this result is noteworthy, a comprehensive cost analysis of the entire renormalization procedure is still necessary to determine which approach yields the lowest overall quantum cost. 
Looking ahead, our framework can be applied to any algorithm with known rigorous upper bounds on the error, and thus serves as a foundation for comparing the total cost of calculating continuum physics from quantum simulations of LGTs using different time evolution algorithms.

In addition to this general formalism, we propose an alternative approach to renormalization, differing from that used in Ref.~\cite{Carena:2021ltu}, when employing product-formulas.
Our approach relies on the observation that simulations using PFs are equivalent to simulations with extra irrelevant operators in the Hamiltonian that disappear in the continuum limit.
Using this fact, we argue that, rather than tuning the bare parameters in the Hamiltonian as a function of both the lattice spacing $a$ and Trotter step-size $\delta t$, one can simply tune the bare parameters assuming $\delta t=0$ and only introduce systematic errors that scale as $\mathcal{O}(a^p)$, where $p$ is the order of the PF used.
This approach offers additional flexibility with the methods used to tune the bare parameters.

The rest of this work is organized as follows.
In Sec.~\ref{sec:high_level_cont_lim}, we provide a high-level review of the continuum limit and explain the main concept of our strategy.
Then, in Sec.~\ref{sec:renorm_review}, we provide a more detailed review of the steps necessary to renormalize the theory and properly take the continuum limit.
We start with a review aimed at non-experts of the main steps for LGT simulations assuming exact time evolution.
From there, we discuss the procedure in Ref.~\cite{Carena:2021ltu} for taking the simultaneous $a, \delta_t \to 0$ limit when using PFs for simulation, and describe our alternative renormalization approach.
We then show that, if one drives the time evolution error below the working statistical precision, it can be ignored and renormalization simplifies to the exact time evolution procedure.
In Sec.~\ref{sec:complexity}, we show how to bound the additional cost of driving the time evolution error below the shot noise, and then apply our methods to simulation via product formulas and Quantum Signal Processing (QSP).
We will show that, while this procedure requires non-negligible overhead when using PFs, the additional cost is entirely negligible when using QSP.
We present our conclusions in Sec.~\ref{sec:conclusion}.

\section{Obtaining the continuum limit \label{sec:high_level_cont_lim}}

Calculations in lattice field theories provide a numerical technique to calculate observables defined in a continuum quantum field theory. 
The lattice field theory can be viewed as an effective theory of the underlying continuum field theory one aims to simulate, where the lattice spacing and lattice size provide an ultraviolet (UV) and infrared cutoff (IR), respectively. 
Numerical results obtained from a lattice simulation therefore differ from the continuum results by systematic effects arising from the non-zero lattice spacing and finite lattice size.
In addition, one needs to control systematic uncertainties from any other approximation that one might make in the lattice calculation, as well as statistical uncertainties made in the numerical lattice calculation.

The most straightforward way to deal with the systematic uncertainties is to make them subdominant to the statistical uncertainties.
In particular, if a particular systematic effect can be bounded to be significantly below the statistical uncertainties, its effect can be neglected in the lattice calculation.
For example, since the effects from finite lattice spacing are known to vanish polynomially with the size of the lattice spacing, discretization errors in a calculation performed at sufficiently small lattice spacing $a$ are guaranteed to be below the statistical uncertainties. 

In reality, many systematic effects cannot be treated this way.
First, while one often knows the scaling of the systematic effects with certain parameters, the actual size of them can not be bounded reliably. 
Second, the computational cost for performing calculations for which certain systematic effects are negligible are often prohibitively expensive. 
For this reason, many systematic effects are dealt with by peforming extrapolations. 
For example, one typically performs calculations at various values of the lattice spacing $a$, all of which have non-neglible systematics. 
However,  one can extrapolate these results to $a = 0$ using the known polynomial scaling in $a$. 

An important property of quantum field theories is that they are UV divergent. 
This implies that physics from arbitrarily short distance effects contribute to a given observable, such that the integration over all possible energies contributing is divergent. 
This does not cause any problem with the predictability of such quantum field theories, since these divergences can be absorbed into the parameters of the theory, such as masses and coupling constants, rendering physical observables finite.
As will be discussed in more detail in~\cref{sec:renorm_review}, the lattice spacing regulates the divergences at short distances, causing the parameters of the effective lattice theory to depend on the lattice spacing used.
In other words, coupling constants and masses used in the lattice simulation depend on the lattice cutoff used
\begin{align}
g \to g(a)\,, \quad m \to m(a)\,, \quad  \ldots
\,,
\end{align}
and the value of these bare parameters need to be adjusted to yield physically meaningful results.

A systematic effect that arises in Hamiltonian lattice gauge theories is due to approximations that typically need to be made when obtaining the time-evolution operator from the Hamiltonian $H$ of the lattice gauge theory. 
This is due to the fact that it is rarely possible to exponentiate the Hamiltonian exactly to obtain the unitary time evolution operator $\exp(-i H t)$. 
The most common approximation scheme used is based on product formulas, with the simplest approximation to the time evolution operator of a Hamiltonian
\begin{align}
H = \sum_i H_i
\,,
\end{align}
given by
\begin{align}
\label{eq:product_formula_def}
e^{-i H t} = \left[e^{-i H \delta t} \right]^{\npf} = \left[\prod_i e^{-i H_i \delta t} \right]^{\npf} + {\cal O}\left(\delta t\right)^2
\,,
\end{align}
with
\begin{align}
\label{eq:delta_t_def}
\delta t \equiv t / \npf, \hspace{1mm} \mathrm{where} \hspace{1mm} \npf \in \mathbb{N}
\,.
\end{align}
The approximation above is a first-order ($p = 1$) product formula, and higher order product formulas with larger values of $p$ are possible, with the deviation from the exact time-evolution operator scaling as $(\delta t)^{p}$ for fixed $t$. 

Choosing a non-zero value of $\delta t$ can affect the UV behavior of the theory meaning that in general the parameters of the lattice theory depend on this parameter as well
\begin{align}
g \to g(a, \delta t)\,, \quad m \to m(a, \delta t)\,, \quad \ldots
\,.
\end{align}
This implies that additional calculations might need to be performed to determine the correct parameter values, adding to the computational cost. 
In the limit $\delta t \to 0$ the product formula approximation of the time evolution operator becomes exact, such that one has
\begin{align}
g(a, 0) = g(a)\,, \quad m(a, 0) = m(a)\,, \quad \ldots
\,,
\end{align}
In other words, taking the limit of exact time evolution does not introduce any UV divergences.
We will discuss the implications of this in Sec.~\ref{sec:syst_unc}.

One way to deal with this systematic uncertainty is therefore to choose a small enough value of $\delta t$ (large enough value of $N_{\rm PF}$) to ensure this error is below the systematic uncertainties and therefore negligible. 
One issue with this approach is that, while one knows the scaling of the error with $\delta t$, the absolute size depends on the details of the Hamiltonian being simulated and is difficult to tightly bound \emph{a priori}. 
Furthermore, the computational cost required to achieve an approximation to the exact time evolution operator with error $\epsilon$ grows as $\epsilon^{-1/p}\exp(p)$, where $\exp(p)$ comes from the fact that the implementation of a $p$'th order Suzuki-Trotter product formula requires exponential resources in $p$. This makes it difficult to reach a sufficiently small value of $\epsilon$.

Given these complications, an alternative procedure to deal with the systematic effect from finite values of $\delta t$ was proposed in Ref.~\cite{Carena:2021ltu}.
In their proposal one keeps the ratio of $\delta t$ and the lattice spacing fixed and takes the simultaneous limit $a \to 0$ and $\delta t \to 0$. 
This is similar to a Euclidean lattice calculation, where both space and Euclidean time are discretized in a space-time lattice, with a fixed ratio of the Euclidean time-like and space-like lattice spacing $a_0$ and $a$. 
While the additional dependence on $\delta t$ can change the value of the bare parameters that have to be chosen, the authors of Ref.~\cite{Carena:2021ltu} went on to show that one can determine the values of these parameters from classical calculations on Euclidean lattices through analytic continuation. 
This however introduces additional systematic uncertainties that need to be quantified.
We discuss this approach in more detail in~\cref{subsec:Euclidean}, and also introduce several variants to the approach proposed in Ref.~\cite{Carena:2021ltu}.

It is important to realize, however, that approximations to the time evolution operator beyond product formulas exist.
A notable example are methods that use the general idea of Quantum Signal Processing (QSP), which allows one to construct a wide class of polynomials of operators. 
These techniques give rise to an approximation to the time-evolution operator, with the computational complexity scaling only as $\log \frac{1}{\epsilon}$\cite{Low:2016sck, Low:2016znh,Motlagh:2023oqc,Gilyen:2018khw}. 

A worry of the community has been\footnote{We thank Henry Lamm for many useful discussions on this point.} whether or not deviations of QSP techniques from the exact time evolution operator can be captured by an effective Hamiltonian picture, as is the case when using PFs.
One possible issue is that a polynomial approximation to the time-evolution operator is generally not unitary (this issue is handled in QSP approaches through ancilla qubits and post-selection protocols), and it is therefore not obvious how to construct a Hermitian effective Hamiltonian describing this non-unitary evolution.
In the absence of an effective theory description, it might be difficult to understand the dependence of the bare parameters of the lattice theory on the approximation used. 
One might therefore worry that a lattice calculation obtained using QSP techniques can not easily be related to the continuum physics one aims to extract.

However, as we argue in this work, if the error from the approximation to time evolution is significantly below all other systematic and statistical uncertainties in the problem, it can not affect the results of the simulation in a meaningful way.
To the accuracy of the calculation, the bare parameters will therefore only have a dependence on the spatial lattice spacing $a$. 
This implies that there is no in-principle impediment to using QSP techniques, or any other time evolution algorithms, for approximating the time evolution of a quantum field theory.
Whether such an approach is favorable compared to the more standard techniques based on product formulas depends on details of the approaches chosen, and need to be investigated carefully.

In the following sections we will explain the two basic topics raised in this section.
In~\cref{sec:renorm_review} we discuss the renormalization of the parameters of the theory and how their numerical values depend on the lattice spacing and the approximation to the time-evolution operator. 
In~\cref{sec:complexity} we study the scaling of the computational cost with the deviation from the exact time-evolution operator for both product formula and QSP based approaches.

\section{Renormalization in Hamiltonian Lattice Gauge Theories \label{sec:renorm_review}}

In this section, we start with a high-level overview of renormalization in quantum field theories aimed at non-experts.
We then more formally discuss how to renormalize LGT calculations and take the continuum limit where time evolution is performed exactly.
From there, we discuss how the process changes when time evolution is performed approximately.

The Kogut-Susskind (KS) Hamiltonian of a pure gauge theory is given by
\begin{align}
    H_{\rm KS} &= \frac{1}{a} \left[ \frac{g_t^2}{2} \sum_{\ell} \left(\hat E_\ell\right)^2 
- \frac{1}{g_s^2} \sum_{p} \mathrm{Re} \, \mathrm{Tr}\, \hat P_p\right]\nonumber\\
 & \equiv H_E + H_B\,,
\end{align}
were hatted operators denote dimensionless operators, rescaled by appropriate powers of the only length-scale in the problem, namely the lattice spacing $a$.
Due to the breaking of Lorentz invariance the dimensionless coupling constants of the electric ($g_t$) and magnetic $(g_s)$ term are different, and the bare speed of light is defined as\footnote{Note that, while the incorrect definition of the speed of light $g_s/g_t$ was used in Ref.~\cite{Carena:2021ltu}, the methods presented are still valid.}
\begin{align}
    c \equiv \frac{g_t}{g_s}
    \,.
\end{align}
One can therefore write
\begin{align}
\label{eq:KS_Hamiltonian}
    H_{\rm KS} &= \frac{c}{a} \left[ \frac{g^2}{2} \sum_{\ell} \left(\hat E_\ell\right)^2 
- \frac{1}{g^2} \sum_{p} \mathrm{Re} \, \mathrm{Tr}\, \hat P_p\right]\nonumber\\
 & \equiv \frac{c}{a}\left[\hat H_E + \hat H_B \right]\nonumber\\
 & \equiv \frac{c}{a}\hat H_{\rm KS}
\,,
\end{align}
where we have defined
\begin{align}
    g = \sqrt{g_s g_t}
    \,.
\end{align}
Using this parameterization introduces a second dimensionful parameter into the theory, namely the time-scale 
\begin{align}
    a_t \equiv \frac{a}{c}
    \,.
\end{align}
Interactions with fermions give rise to additional terms that involve the masses $m_i$ of the fermions as well as a so-called hopping term describing the interactions between gauge bosons and fermions with coefficient $\kappa$.\footnote{To the best of our knowledge, the need for the additional bare parameter $\kappa$ multiplying the fermionic kinetic term in the Hamiltonian formulation has not been discussed in the literature. The need for this parameter is discussed in Appendix~\ref{app:scale_set_and_HKS}.}

It is important to remember that the KS Hamiltonian is only equal to the full continuum Hamiltonian up to corrections that vanish as the lattice spacing approaches zero.
The pure gauge Hamiltonian as written in~\cref{eq:KS_Hamiltonian} is correct up to corrections scaling as $a^2$, such that the continuum Yang-Mills (YM) Hamltonian can be written as
\begin{align}
\label{eq:YM_to_KS}
\hat H_{\rm YM} = \hat H_{\rm KS} + \sum_i c_i \hat O^{(3)}_i + \dots
\,,
\end{align}
where $\hat O^{(3)}_i$ denotes higher dimensional operators whose contributions to long-distance physics at distances $L$ scale as $a^3 / L^3$. 
Since the rescaled Hamiltonian $\hat{H}_{\rm KS}$ vanishes as $a / L$ as just discussed, any calculation performed with the KS Hamiltonian will only reproduce the desired continuum physics up to relative corrections suppressed by $a^2 / L^2$. 

As discussed in more detail in the next section, any dimenionful parameter in the Hamiltonian will also be replaced by the corresponding dimensionless parameter by multiplying by the appropriate factors of the lattice spacing and the speed of light.
The ruler with which temporal- and spatial-type quantities are measured with respect to are $a/c$ and $a$, respectively\footnote{The fact that spatial-type quantities are measured with respect to the lattice spacing is reflected in the fact that the lattice momentum operator would have an overall prefactor of $1/a$ instead of $c/a$ in the case of the Hamiltonian.}.
Examples are fermion masses in the Hamiltonian, the time one evolves the system by, momenta, or the physical distance between correlators.
Their dimenionless equivalents are given by
\begin{align}
\label{eq:hatted_def}
    \hat m = a_t \, m \,, \qquad \hat t = \frac{1}{a_t} t\,, \qquad \hat{p}=ap, \qquad \hat x = \frac{1}{a} x
    \,.
\end{align}

In the limit of vanishing lattice spacing $a \to 0$, the KS Hamiltonian has to reproduce the spectrum of the continuum theory. This implies that the eigenvalues of the Hamiltonian $H_{\rm KS}$ have to be finite in that limit, which further implies that the eigenvalues of the rescaled Hamiltonian $\hat{H}_{\rm KS}$ have to vanish as $a \to 0$.

\subsection{Renormalization for exact time evolution}
\label{subsec:renorm_exact_te}

Calculating physical observables using quantum field theory requires integrating over all possible energy and momentum values of intermediate states that contribute to a process.
These integrals, however, generally diverge as the integration range is taken to infinity.
One intuitive explanation for this is that such calculations assume that the theory is valid at arbitrarily high energy scales; because current theories are incomplete (in the sense that they are not the final ``theory of everything''), this is simply not true.
This, however, does not give rise to an ill-defined theory or a loss of predictive power. 
The physics at short distances is captured using a process known as \emph{renormalization}, in which the parameters of the effective theory (coupling strengths of interactions, masses of particles, etc.) are adjusted to reproduce the physics at long distances one aims to compute.

In general, renormalization requires two steps.
The first is to \emph{regulate} the theory by manually imposing some cutoff $\Lambda$ on the energies or momenta included in the theory (in a lattice theory $\Lambda \sim 1/a$). 
While this renders calculations finite, it naively makes physical observables depend on the cutoff $\Lambda$.
The key insight of renormalization is that this unphysical behavior can be removed by choosing the bare parameters in the theory to depend on the cutoff in precisely the correct way such that some physical observable is reproduced correctly and is independent of the cutoff.
In particular, the bare coupling, speed of light and masses become $\Lambda$ dependent $g \to g(\Lambda)$, $c \to c(\Lambda)$, $m_i \to m_i(\Lambda)$ and $\kappa \to \kappa(\Lambda)$.
While this leads to the bare parameters diverging (or going to zero) as $\Lambda \to \infty$, this is not a problem as they are not physical observables in the theory.

While the dependence of the parameters on $\Lambda$ can in some cases be computed perturbatively, in general this is not possible, and one has to determine this dependence by demanding that the effective theory reproduces a limited set of physical observables. 
After adjusting the bare parameters to absorb the divergences as $\Lambda \to \infty$ and to reproduce a small set of observables, the theory has been renormalized and all other predictions will be finite and equal to the physical values\footnote{If using perturbation theory, the predictions will be equal up to higher order corrections.} in the $\Lambda \to \infty$ limit.

Since computers (both classical and quantum) can only deal with dimensionless numbers, one numerically simulates a rescaled Hamiltonian $\hat H_{\rm KS}$ introduced in~\cref{eq:KS_Hamiltonian}, and from now on, all parameters of the lattice Hamiltonian are assumed to be dimensionless, and the parameter $a_t=a/c$ only appears as an overall constant.
This has two consequences that are intimately connected and which we will discuss more later. 
First, lattice calculations can only compute dimensionless quantities, and in order to obtain dimensionful quantities observed in nature, such as masses, one needs to multiply by the appropriate powers of $a$ in the end. 
Second, because one only chooses values of the rescaled parameters, the lattice spacing is not an explicit input in a simulation.
Instead, its value is only obtained after the fact by comparing calculations to dimensionful quantities observed in nature.
See Appendix~\ref{app:scale_set_and_HKS} for a more detailed discussion of this property of lattice field theories.

We distinguish two types of parameters in our theory. 
The first set are dimensionless, possibly rescaled parameters that determine the relative size of terms in the Hamiltonian. 
For the Kogut-Susskind Hamiltonian the only such parameter is the gauge coupling $g$, but when including fermions,  rescaled masses are present as well. 
We denote these parameters collectively as 
\begin{align}
    p_i \equiv \{g, \hat m_i, \hat\kappa \}
    \,.
\end{align}
where $\hat \kappa = \kappa / c$ as shown in Appendix~\ref{app:scale_set_and_HKS}.
The second set of parameters are those that set the scales in the problem, one scale for time $a_t$ and one for length $a$, or equivalently $a$ and the speed of light $c$.
A common choice is the lattice spacing $a$ and the speed of light $c$. 

As already mentioned, the presence of a lattice spacing introduces a cutoff $\Lambda_a \sim 1/a$, since wavelength shorter than the lattice spacing $a$, and time scales shorter than $a / c$ can not be resolved.
From the previous discussion it therefore follows that renormalization demands that the bare parameters in the Hamiltonian depend on the lattice spacing.
This gives a lattice gauge theory Hamiltonian $H(a)$ with all parameters being functions of the lattice spacing 
\begin{align}
p_i \to p_i(a)
\,.
\end{align}
As we will discuss, the speed of light will also depend on the lattice spacing $a$.

At first sight, one might be worried that this makes it impossible to use a lattice field theory to predict physical observables. 
After all, it was just discussed that the value of the lattice spacing $a$ can only be extracted at the end of a lattice simulation, and that the functional dependence $p_i(a)$ is non-perturbative and therefore not known analytically.
Given this, how does one know what values to choose for the bare parameters $p_i(a)$ used in the calculation?
The answer is that one needs to tune these  parameters to reproduce experimentally accessible quantities, such as masses in the spectrum of the full theory.
One typically views this as a two step process, as we will now discuss. 

In a first step, one determines the dimensionless parameters of the Hamiltonian. 
This is accomplished by finding parameters for which the numerical calculation reproduce experimentally known dimensionless quantities. 
A common choice for such dimensionless ratios are ratios of masses
\begin{align}
    \frac{\frac{a}{c}M_1(p_i)}{\frac{a}{c}M_2(p_i)} = \frac{M_1^{(\rm phys)}}{M_2^{(\rm phys)}}\,.
\end{align}
Note that for a pure gauge Hamiltonian, this step is not necessary as all choices of the gauge coupling are valid, with different choices corresponding to different lattice spacings as determined in the next step.
When fermions are included, however, the bare fermion masses have to be tuned by the above procedure. 
The parameter values are typically determined by iterating over a range of parameter values in the numerical calculation and using some fitting procedure for the final tuning.

Having reproduced dimensionless ratios, one determines the scales $a_t$ and $a$ by comparing the results obtained in the previous step to observed dimensionful quantities. 
For example, the scale $a_t$ is can be obtained by taking the ratio of a dimensionless mass computed numerically to its observed value
\begin{align}
    a_t = \frac{\hat M(p_i)}{M^{(\rm phys)}}
    \,.
\end{align}
The scale $a$, or equivalently the speed of light, is obtained by demanding that observables depending on spatial distances are reproduced correctly. 
For example, one can demand that the energy of a massive particle with given lattice momentum reproduces the energy predicted from the relativistic dispersion relation.

Thus, for each parameter set $p_i$ one finds corresponding values of $a$ and $a_t$, or equivalently $a$ and $c$.
One can then interpret these parameter sets as labeled by the lattice spacing found and write
\begin{align}
    \{p_i, c, a\} \to \{p_i(a), c(a)\}
    \,.
\end{align}
One therefore obtains a line in parameter space, where each point corresponds to a set of parameter values and lattice spacing.
This line is often called the renormalization trajectory.
Given that this trajectory corresponds to exact time evolution, we denote it by
\begin{equation}
\label{eq:traj0}
    {\rm traj}_{0}\, : \, \{p_i(a), c(a)\}\,.
\end{equation}
Due to the divergences in the underlying continuum theory, the parameter values corresponding to vanishing lattice spacing are typically either infinite or vanishing; because Lorentz invariance is restored in the continuum limit, the bare speed of light and the parameter $\kappa(a)$ obeys $\lim_{a \to 0} c(a) =\lim_{a \to 0} \kappa(a) = 1$.
Note that different renormalization trajectories associated with different observables used in scale setting are possible, with each giving the same result in the $a \to 0$ limit.
The dependence of the dimensionless coupling $g$ on the lattice spacing is logarithmic. 

Given this renormalization trajectory, one can now compute new observables of interest. 
After renormalization, a lattice calculation of low energy quantities differs from the continuum result by powers of the lattice spacing; in modern calculations the leading errors are typically $\mathcal{O}(a^2)$.
One can not perform a calculation directly at $a = 0$, since as discussed above the parameter values at that point are either infinite or vanishing, and would also require an infinite number of lattice sites to maintain a constant physical volume.
One therefore calculates the desired physical observables for various parameter choices along the renormalization trajectory, resulting in a value that depends on the corresponding lattice spacing.
To obtain the continuum result, one performs a final extrapolation to $a = 0$. 

Given that a major advantage of Hamiltonian lattice gauge theory simulated on quantum computers over traditional lattice gauge theory is the ability to calculate time-dependent observables, we now consider the calculation of expectation values of a time-dependent operator $\langle O(t) \rangle$\footnote{Note that our arguments also apply to the scenario where one is computing more general matrix elements of an operator between two different states $\bra{\psi_f} O(t) \ket{\psi_i}$.} through the calculation of the expectation value of an operator defined on the lattice $\langle O(t, a) \rangle$. 
The dependence on $a$ captures the fact that calculation is performed with the parameter values $\{p_i(a)\}$ of the renormalization trajectory corresponding to the value $a$. 
As before, the lattice only allows one to compute dimensionless quantities, such that for an operator of mass dimension $\gamma$ one computes the expectation value $\langle \hat O(t,a) \rangle =  \langle a^\gamma O(t,a)\rangle$. 

Unless protected by a symmetry in the theory, the operator $\hat O(t, a)$ also needs to be renormalized. 
This is done by introducing renormalization factors $Z(a)$.
For simplicity we assume multiplicative renormalization, but the arguments presented in this manuscript generalize to additive renormalization factors and operator mixing.
We determine  $Z(a)$ by demanding some other chosen matrix element of $Z(a)\hat O(t, a)$ is equal the associated physical value, \emph{i.e.},
\begin{equation}
    Z(a) \equiv \frac{\bra{\phi_1} O(t, a)\ket{\phi_2}}{[\bra{\phi_1}O(t)\ket{\phi_2}]^{({\rm phys})}} \,,
\end{equation}
for some states $\ket{\phi_1}$ and $\ket{\phi_2}$.
Once $Z(a)$ has been determined, the physical observable is given in the limit
\begin{equation}
    \langle O(t) \rangle = \lim_{a \to 0} Z(a) a^{-\gamma} \langle \hat O(t, a) \rangle\,.
\end{equation}
Note that $Z(a)$ is time-independent in this case, as the renormalized operator is given by $Z(a) O(t,a) = e^{i H(a) t} \left(Z(a) O(0,a) \right) e^{-i H(a) t}$. 

We now summarize the discussion given so far, by providing the steps required to take the continuum limit of a time-dependent observable. 
They are
\begin{enumerate}
    \item Determine the renormalization trajectory and therefore $g(a)$, $\hat m_i(a)$, $\hat{\kappa}(a)$, and $c(a)$
    \begin{enumerate}
        \item Calculate a set of known dimensionless observables for a given parameter set $p_i$
        \item Adjust the parameters to until calculations match the known dimensionless values
        \item Determine $a$ and $a_t$ for this parameter set by comparing against a dimensionful observables (recall that converting between $a$ and $a_t$ is done by a trivial rescaling with $c(a)$)
        \item Repeat steps (1a)-(1c) until the renormalization trajectory is obtained
    \end{enumerate}
    \item Compute the desired physical observable
    \begin{enumerate}
    \item Pick a value of $a$ by choosing a set of parameters $p_i(a)$ and $c(a)$
    \item If calculating a time-dependent observable, determine the appropriate value of $\hat t = c(a) t / a$
    \item Determine $Z(a)$ using the previously determined parameter values
        \item Calculate $\langle \hat O(\hat t,a) \rangle$ and rescale by powers of $a$ to convert to physical units
        \item Repeat steps (2a)-(2d) for several values of $a$ and extrapolate $Z(a) \langle O(\hat t,a) \rangle$ to zero lattice spacing
    \end{enumerate}
\end{enumerate}
From the discussion in this section, it should be clear that while all parameters and observables in the above steps have explicit dependence on $a$, this dependence on $a$ only labels the parameter values used in the given calculation. 

We conclude this section by noting that, for a given lattice geometry and Hamiltonian, one only needs to tune the bare parameters and scale set once; once the renormalization trajectory has been determined, the relative cost of these steps can be reduced by reusing it for several observables.

\subsection{Renormalization for  approximate time evolution via product formulas}
Product formulas are based on the Baker-Campbell Haussdorff (BCH) formula
\begin{align}
    e^{-i H_1 \delta_t}e^{-i H_2 \delta_t} = e^{-i (H_1 + H_2) \delta_t - \frac{\delta_t^2}{2}[H_1,H_2] + \ldots}
    \,,
\end{align}
and approximate the time evolution operator by products of individual exponentials.
For example, writing
\begin{align}
    H = H_1 + H_2
\end{align}
the first and second order product formulas will approximate the time evolution operator $U = \exp[-i H t]$ as
\begin{align}
\label{eq:BCHFormula}
    S_1(t) &= \left[S_1(\delta_t)\right]^{N_{\mathrm{PF}}} = \left[e^{-i  \delta_t H_1}e^{-i  \delta_tH_2}\right]^{N_{\mathrm{PF}}}\\
    S_2(t) &= \left[S_2(\delta_t)\right]^{N_{\mathrm{PF}}} = \left[e^{-i  \delta_t H_1/2}e^{-i \delta_t H_2}e^{-i  \delta_t H_1/2}\right]^{N_{\mathrm{PF}}}\nonumber
    \,,
\end{align}
where $\delta_t$ is given by \cref{eq:delta_t_def}.
Note that time evolution using product formulas requires two independent parameters ${N_{\mathrm{PF}}}$ and $\delta_t$, which need to combine to produce the desired physical time of interest
\begin{align}
    {N_{\mathrm{PF}}} \delta_t = t
    \,.
\end{align}

One way to use product formulas is to choose a value of ${N_{\mathrm{PF}}}$ that is large enough such that the systematic uncertainties to time evolution becomes negligible. 
This will be discussed in detail and generalized to other approximations to time evolution in~\cref{sec:syst_unc}, and requires to find a value of ${N_{\mathrm{PF}}}$ for which the the resulting systematic uncertainties are guaranteed to be below the statistical uncertainties made in the numerical calculation. 
In the rest of this section, we discuss how to handle the systematic uncertainties through renormalization, such that they vanish as the continuum limit is taken. 
Much of our discussion follows Ref.~\cite{Carena:2021ltu}, and we point out when we find results that differ from that work.

Using the BCH formula,~\cref{eq:BCHFormula}, time evolution via product formulas correspond to time evolution for a time step $\delta_t$ according to the effective Hamiltonians
\begin{align}
    S_k(\delta_t) = e^{-i \delta_t H_{\rm eff}^{(k)}}
    \,,
\end{align}
with 
\begin{align}
   H^{(1)}_{\rm eff}(\delta_t) &=  H_1 + H_2 - i \frac{\delta_t}{2} \left[H_1, H_2\right] + \ldots\\
   H^{(2)}_{\rm eff}(\delta_t) &=  H_1 + H_2 + \frac{\delta_t^2}{24} \Big(\left[H_1, [H_1,H_2]\right]\\
   &\hspace{1cm} +2 \left[H_2, [H_1,H_2]\right]\Big) + \ldots\nonumber
   \,.
\end{align}
We call the terms involving commutators of $H_i$ BCH terms, since they arise from the higher order terms in the BCH formula.
Rescaling all parameters by the appropriate factors of the lattice spacing $a$, including the appropriate factors of the speed of light
\begin{align}
\label{eq:hatdef}
    \hat H_i = \frac{a}{c} H_i \,, \qquad \hat{\delta}_t = \frac{\hat t}{{N_{\mathrm{PF}}}}
    \,,
\end{align}
where $\hat t$ is defined in~\cref{eq:hatted_def}. 
One therefore finds
\begin{align}
\label{eq:H1eff}
   \hat H^{(1)}_{\rm eff}(\delta_t) &=  \hat H_1 + \hat H_2 - i \frac{\hat \delta_t}{2} \left[\hat H_1, \hat H_2\right]+ \ldots\\
\label{eq:H2eff}
   \hat H^{(2)}_{\rm eff}(\delta_t) &=  \hat H_1 + \hat H_2 + \frac{\hat\delta_t^2}{24}\Big(\left[\hat H_1, [\hat H_1,\hat H_2]\right]\\
   &\hspace{1cm} +2 \left[\hat H_2, [\hat H_1,\hat H_2]\right]\Big) + \ldots\nonumber
   \,.
\end{align}
For $k$'th order product formulas the leading BCH terms contain $k$ powers of $\hat \delta_t$ and commutators involving $k+1$ factors of $\hat H_i$.

When taking the continuum limit, one needs to decide how to scale the number of Trotter steps as $a$ is taken to zero.
One obvious choice is to take the limit in such a way that $\hat \delta_t$ stays constant. 
Since the in the limit $a \to 0$ one also has $a_t \to 0$ (since the speed of light goes to a constant), this requires the number of Trotter steps to increase as $1/a$ and ensures that the resolution of the Trotter time step scales with the resolution introduced by $a_t$.
For this choice of $\hat \delta_t$ scaling 
one finds, by analyzing~\cref{eq:H1eff,eq:H2eff}, that the BCH terms are suppressed by $k$ powers of $a$ relative to the leading order term when using a $k^{\mathrm{th}}$ order product formula. 
This implies that, if one uses a lattice Hamiltonian with discretization errors scaling as $\mathcal{O}(a^\gamma)$ for some positive integer $\gamma$, one must use a PF with $k\geq \gamma$ to avoid introducing larger lattice spacing errors.
For the rest of this paper we will use second order Trotter approximations to time evolution.

As in the previous section, renormalization of the system corresponding to the effective Hamiltonian is required to calculate any physical quantities.
We begin by discussing two ways to perform this renormalization, and will show that both of these are equivalent to the order one is working.
In the first approach already discussed, one chooses a fixed value for the parameter $\hat \delta_t$ and then performs the same steps as before to determine the parameters $p_i$, as well as $a$ and $c$. 
This will result in a renormalization trajectory that depends on the choice of parameter $\hat \delta_t$, which we will denote as
\begin{align}
\label{eq:traj1}
    {\rm traj}_1:\, \{p_i(a, \hat \delta_t), c(a, \hat \delta_t)\}
    \,.
\end{align}
The second approach would be to treat $\hat \delta_t$ as an additional parameter in the Hamiltonian to be fixed by renormalization. 
This would require knowledge of one additional experimental observable in the continuum, and would result in the trajectory
\begin{align}
\label{eq:traj2}
    {\rm traj}_2:\, \{p_i(a), c(a), \hat \delta_t(a)\}
    \,.
\end{align}
Note that the values $p_i$ and $c$ corresponding to a given lattice spacing will differ from those obtained using exact time evolution to ${\cal O}(a^2)$; this is true for either trajectory.

Both of these approaches give rise to differences in observed quantities that are ${\cal O}(a^2)$. 
This is because, as already discussed, the parameter $\hat \delta_t$ modifies higher-dimensional operators that contribute at relative order $a^2$ to any observable at long distances. 
These operators therefore contribute at the same order as other higher dimensional operators that are neglected in the Kogut-Susskind Hamiltonian, \emph{i.e.}, the $\hat{O}_i^{(3)}$ operators in \cref{eq:YM_to_KS}.
This implies that including or not including the BCH contributions only changes numerical simulations at a given lattice spacing only up to $\mathcal{O}(a^2)$ corrections, at which level other effects not included contribute as well.

In fact, these considerations allow for additional options, which are simpler than the first two, at least conceptually. 
In a first option, one simply uses the renormalization trajectory corresponding to $\hat \delta_t = 0$ in Eq.~\eqref{eq:traj0}, which is the renormalization trajectory corresponding to exact evolution already discussed in Sec.~\ref{subsec:renorm_exact_te}.
While the additional $\hat \delta_t$ dependence implies that the renormalization conditions are not satisfied exactly, they are correct within the ${\cal O}(a^2)$ systematic uncertainties already present in the numerical calculations.
A second alternative, which is even simpler than before, would be to neglect the renormalization of the speed of light and $\kappa$ altogether, since their differences from their continuum value $c = \kappa = 1$ are themselves of ${\cal O}(a^2)$. 

Which of these approaches are should be taken depends on the  relative size of the prefactor of the neglected $\mathcal{O}(a^2)$ terms. 
The speed of light is tuned in most lattice simulations with anisotropic lattices because it that it is doing so is straightforward and reduces discretization errors from higher order terms.
We therefore expect that using the trajectory in Eq.~\eqref{eq:traj2} will give rise to the smallest overall corrections.

\subsection{Using Euclidean simulations on anisotropic lattices to perform renormalization}
\label{subsec:Euclidean}

A different approach to renormalization is to use Euclidean simulations performed on classical computers to obtain the renormalization trajectory.
We will discuss this approach, which was first proposed in Ref.~\cite{Carena:2021ltu}, in the rest of this section.
The starting point is the fact that the Euclidean action used in classical lattice Monte-Carlo calculations and the Hamiltonian in Hamiltonian lattice simulations can be related to one another using the transfer matrix formalism. 

This formalism, first proposed in Ref.~\cite{Creutz:1983njd}, is usually used to show that one can obtain the Kogut-Susskind Hamiltonian from the Wilson action by taking the limit of vanishing Euclidean time-like lattice spacing.
The starting point is to define a transfer matrix operator $T$, such that the partition function $Z$ can be written as 
\begin{align}
    Z \equiv \int {\rm D}U \, e^{-S_E} = \Tr T^N 
    \,,
\end{align}
where the Euclidean Wilson action is given by
\begin{align}
\label{eq:Wilson}
    S_E = -\beta_s \sum_s \Re \Tr P_s -\beta_s \sum_t \Re \Tr P_t
    \,,
\end{align}
where $P_s$ ($P_t$) denotes the plaquette operators defined on space-like (time-like) plaquettes. 
The parameters $\beta_i$ are defined as
\begin{align}
    \beta_s = \frac{a_0}{g_s^2 a}\,, \qquad \beta_t = \frac{a}{g_t^2 a_0}
    \,,
\end{align}
where $a_0$ is the lattice spacing in the Euclidean time direction.

The transfer matrix is defined as the matrix element of the transfer operator $T$
\begin{align}
    T_{UU'} \equiv \langle U | T | U' \rangle\,.
\end{align}
Defining the gauge configurations $U_i$ as those that correspond to the fixed time-slice $i$ on the periodic Euclidean lattice (such that $U_0 = U_N$), one can therefore write
\begin{align}
    T = \int {\rm d}U_0\ldots {\rm d}U_{N-1} \, T_{U_0 U_1} \, T_{U_1 U_2}\ldots T_{U_{N-1} U_0}
    \,.
\end{align}
Demanding that the matrices $T_{UU'}$ are Hermitian one finds that the operator
\begin{align}
\label{eq:Tdef}
    T = 
    \exp\left[-a_0 \frac{H_B}{2}\right]
    \exp\left[-a_0 H_E\right]
    \exp\left[-a_0 \frac{H_B}{2}\right]
    \,,
\end{align}
reproduces the partition function corresponding to the Wilson action in~\cref{eq:Wilson} in the limit $a_0 \to 0$.

Since the transfer operator provides the relation between different time slices on the Euclidean lattice, it is related to the Hamiltonian of the system which governs the differential time evolution of the system. 
In particular, by taking the $a_0 \to 0$ limit one obtains 
\begin{align}
    T = \exp[-a_0 H_{\rm KS}]
    \,,
\end{align}
up to higher order corrections in $a_0$, producing the result that~\cref{eq:Tdef} reproduces the Kogut Susskind Hamiltonian in this limit\footnote{Since this transfer operator only reproduces the Wilson action in the $a_0 \to 0$ limit, this transfer operator is not unique, implying one has freedom in defining which lattice Hamiltonian to simulate. For example, one could choose a transfer operator that uses more products of the various exponentials, such that it resembles a higher order representation of the Baker Campbell Hausdorff formula. 
One could also use an operator that do not contain the operators $\hat H_B$ and $\hat H_E$ in the exponentials, but rather operators that include higher order corrections in $a_0$.  We are grateful to Wanqiang Liu for discussions on this topic.}.

It was pointed out in Ref.~\cite{Hoshina:2020gdy,Kanwar:2021tkd,Carena:2021ltu} that the operator $T$ in~\cref{eq:Tdef} (associated with the heat-kernel action~\cite{Menotti:1981ry}) is equal exactly to the approximation of the time evolution operator of the pure gauge Kogut-Susskind Hamiltonian when using a second order product formula with the same breakup into $H_E$ and $H_E$ as in~\cref{eq:Tdef}.
In other words, one finds
\begin{align}
    T = S_2\left(-i a_0 \frac{c}{a}\right)
    \,,
\end{align}
without any higher order corrections in $a_0$ or $\hat \delta_t$ so long as $H_1 = H_B$ and $H_2 = H_E$ in \cref{eq:BCHFormula}.
This implies that calculations performed in a Hamiltonian lattice gauge theory with Kogut-Susskind Hamiltonian and this second order Trotter evolution is exactly equal to the analytical continuation of a Euclidean path integral calculation performed with the heat kernel action.
It is important to keep in mind that this works only for this particular second order Trotter formula. 
Using other Trotter expansions, such as first order, second order with a different breakup of the Hamiltonian or a different order of terms, of working higher-order product-formulas, however, is a difficult task, as one must start from the desired higher-order transfer matrix and work backwards to determine the associated Euclidean lattice action.

Another issue with this approach is that the exact relationship between the classical Euclidean action and a Trotterized time evolution operator in the Hamiltonian formulation is spoiled when including fermions, and to our knowledge no classical Euclidean action is known that reproduces a Trotterized version of a corresponding Hamiltonian without introducing ${\cal O}(a_0)$ corrections.
This implies that the approach of using classical calculations to obtain the renormalization trajectory required for Hamiltonian simulations no longer works to the desired accuracy in the presence of fermions.

To summarize, the discussion above shows that one can use classical simulation methods to determine the renormalization trajectory~\cite{Carena:2021ltu, Carena:2022hpz} for pure Yang-Mills simulations.
In this approach, avoiding a systematic mismatch at $\mathcal{O}(a^2)$ requires performing classical calculations with the heat kernel action, which is considerably more difficult to use than the Wilson action, and no Euclidan action is currently known that allows to extend this approach to include fermions, as is required to simulate the strong interaction.

In the next section we discuss an alternative approach, which does not suffer from the shortcomings of the renormalization approach, and works generically for any lattice Hamiltonian. 
This approach takes the ordered limit $\lim_{a\to 0} \lim_{\delta_t\to 0}$ in a way that incurs negligible systematic uncertainties from approximate time evolution. 
The key insight is that, as all numerical computations performed, whether for tuning of parameters or to compute the final observables, have associated statistical uncertainties, and as long as the systematic uncertainties from the approximate time evolution can be guaranteed to be below these statistical uncertainties, they can be neglected. 

\section{Statistically-Bounded Time Evolution Protocol}
\label{sec:syst_unc}

In the previous section we have outlined a general procedure to
numerically compute observables in Hamiltonian lattice gauge theories, including the requirements to renormalize the bare parameters of the Hamiltonian. 
In particular, we discussed how to perform the two necessary steps when an approximation to the time evolution operator is being used, as is typically always required. 
In this section we will discuss how to deal with the systematic uncertainties that arise from this approximation to the time evolution operator, in particular how to ensure that these systematic uncertainties have no effect on the continuum limit. 
We refer to this general protocol as the Statistically-Bounded Time Evolution (SBTE) protocol.
An important aspect of our discussion will be the computational cost of the SBTE protocol.

From the previous discussion on renormalization when approximating time evolution using product formulas, it should be clear that all required numerical calculations are performed with a fixed set of parameters in the Hamiltonian. 
In the first step, these calculations are used to determine the renormalization trajectory, while in the second step calculations are performed for parameters that lie on this trajectory.
It is important to realize that any numerical calculation will come with an associated calculational uncertainty, unrelated to any systematic uncertainty due to underlying theoretical assumptions. 
Such calculational uncertainty can arise from noisy gates on quantum hardware, but even in the absence of such noise, calculations have statistical uncertainties, either from using Monte-Carlo integrations on classical computers or from measuring expectation values of operators from repeated circuits on quantum computers. 
This implies that calculating the expectation value of any observable is only possible to a given uncertainty
\begin{align}
    \langle \hat O \rangle_{\rm calc} = \langle \hat O \rangle \pm \sigma_{\hat O} \,.
\end{align}

Now consider calculating this expectation value using some approximation to the time evolution operator, where the approximation is controlled by some parameter $\delta$, with exact time evolution recovered in the limit $\delta \to 0$. 
For approximations using product formulas the parameter $\delta = \delta_t$ is the bare Trotter step size used, but other approximations are possible as well. 
We will later discuss the case of quantum signal processing (QSP), for which the analogous parameter characterizes the accuracy of the Jacobi-Anger expansion of the time evolution operator in terms of 
Chebyshev polynomials up to some finite order. 
For both of these approaches, bounds exist that guarantee that the accuracy of the approximation 
\begin{equation}
    \label{eq:approx_obs}
    \left |\langle \hat O(\delta)\rangle - \langle \hat O\rangle\right| = \epsilon_\delta
    \,.
\end{equation}
is below a certain value. 
For product formulas this is achieved by keeping the number of Trotter steps $N_{\rm PF}$ to be above a certain value, while for QSP it requires the number of Chebyshev polynomials $k$ used in the polynomial approximation to exceed a certain value.

As already mentioned in~\cref{sec:high_level_cont_lim} no divergences appear in the limit $\delta \to 0$, even though bare parameters can depend on the value of $\delta$. 
This implies that choosing a sufficiently large value of $\delta$ guarantees that $\epsilon_\delta \ll \sigma_{\hat{O}}$, ensuring that the resulting systematic uncertainties are negligible compared to $\sigma_{\hat{O}}$.
Therefore all steps of the previous procedure go through as if one had used exact time evolution.

As will be discussed in~\cref{subsec:pf}, ensuring a small enough error with product formulas can lead to prohibitively expensive gate costs for two reasons.
The first is that, while rigorous error bounds exist when using PFs, they are generally loose, see, \emph{e.g.}, Ref.~\cite{PhysRevLett.123.050503}.
The second is the polynomial scaling $\mathcal{O}(\epsilon_\delta^{-1/p})$ of the gate cost.
These facts explain the desire to deal with this uncertainty using the renormalization procedure discussed, which could result in smaller gate costs at the expense of a more complicated renormalization procedure.

For QSP, the scaling of computational resources with their accuracy is however much more favorable in comparison to product formulas, being proportional to $\log (1/\epsilon_\delta)$, and is in fact provably optimal in both total simulation time and error.

\subsection{Algorithmic Complexity}
\label{sec:complexity}

In this section we compute the algorithmic cost required to control the approximation accuracy $\epsilon_{\delta}$ with respect to the statistical uncertainty $\sigma_{\hat O}$ in order to ensure the renormalization steps outlined in \cref{subsec:renorm_exact_te} can be carried out as if one were performing exact time evolution. In particular, for a specific choice of time evolution algorithm, we ask how to choose the algorithmic parameter $\delta$ that controls the approximation error $\epsilon_{\delta}$ incurred by the algorithm to ensure that $\epsilon_{\delta} \leq \sigma_{\hat O}/\beta$ where $\beta \in \mathbb{R}^+$ such that $\beta \geq 1$. In this work, we do so for two common (and conceptually simple) time evolution algorithm i.e. Product Formula (PF) as well as for Quantum Signal Processing (QSP) based algorithms. 

Let $H$ be the Hamiltonian governing the dynamics of the theory we are interested in simulating. Let $\exp(-iHt)$ and $U_{\mathrm{sim}}(t)$ be exact and approximate time evolution operators, where the latter is the operator implemented by the algorithm of choice. We define the operator $\Delta_{\mathrm{sim}}$ as follows:
\begin{align}
    \label{eq:delta_sim}
    \Delta_{\mathrm{sim}} := \mathrm{e}^{-iHt} - U_{\mathrm{sim}}(t)
    \,,
\end{align}
such that the spectral norm $\|\Delta_{\mathrm{sim}}\|$ is the exponentiation error incurred by the time evolution algorithm. The following lemma provides a necessary condition on the quantity $\|\Delta_{\mathrm{sim}}\|$ required while computing operator expectation values of the form $\langle {\hat O}(t, a) \rangle$:
\begin{lemma}
    \label{lemma:alg_err}
    For any choice of time evolution algorithm and any $\beta \in \mathbb{R}^+$ such that $\beta \geq 1$, choosing algorithmic parameters that ensure
    \begin{align}
        \label{eq:delta_sim_ub}
        \|\Delta_{\mathrm{sim}}\| \leq \frac{\sigma_{\hat O}}{2\beta\|{\hat O}(0,a)\|}
        \,,
    \end{align}
    will result in $\epsilon_{\delta} \leq \sigma_{\hat O}/\beta$.
\end{lemma}
\begin{proof}
From \cref{eq:approx_obs}, we have that
    \begin{align}
    \epsilon_{\delta} 
        &= \| \mathrm{e}^{iHt}\hat{O}(0,a)\mathrm{e}^{-iHt} 
            - U_{\mathrm{sim}}(t)\hat{O}(0,a)U_{\mathrm{sim}}^{-1}(t) \| \\
        &= \| \mathrm{e}^{iHt}\hat{O}(0,a)\mathrm{e}^{-iHt} \notag \\
        &\quad - \left(\mathrm{e}^{iHt} - \Delta_{\mathrm{sim}}^{\dagger}\right)
                \hat{O}(0,a)
                \left(\mathrm{e}^{-iHt} - \Delta_{\mathrm{sim}}\right) \| \\
        &= \| \Delta_{\mathrm{sim}}^{\dagger} \hat{O}(0,a)) \mathrm{e}^{-iHt}
            + \mathrm{e}^{iHt} \hat{O}(0,a) \Delta_{\mathrm{sim}} \notag \\
        &\quad - \Delta_{\mathrm{sim}}^{\dagger} \hat{O}(0,a) \Delta_{\mathrm{sim}} \|
        \,.
    \end{align}
    We can then apply the triangle inequality to write:
    \begin{align}
        \label{eq:epsilon_sim_ub}
        \epsilon_{\mathrm{\delta}} &\leq 2\|\Delta_{\mathrm{sim}}\| \cdot \|\hat{O}(0,a)\| + \mathcal{O}(\|\Delta_{\rm sim}\|)^2
        \,,
    \end{align}
    where we ignore the quadratic correction term and upper bound the RHS of \cref{eq:epsilon_sim_ub} with $\sigma_{\hat O}/\beta$ to arrive at \cref{eq:delta_sim_ub}.
\end{proof}

Equipped with this Lemma, we can now consider specific instances of time evolution algorithms and analyze the algorithmic gate complexity required to appropriately bound the correction term $\Delta_{\mathrm{sim}}$.
While we consider the cost to do so using rigorous error bounds in the next section, another approach is to perform the simulation at several values of decreasing $\delta$ and showing that the result does not change within the working statistical precision.
Both methods are included in the SBTE protocol, and we discuss the second approach in more detail in Sec.~\ref{ssec:comparison}.

\subsubsection{SBTE protocol with Product Formulas}
\label{subsec:pf}

We begin this section by briefly reviewing the theory of PFs. For a more careful treatment of this theory, see for example Refs.~\cite{doi:10.1126/sciadv.aau8342, Childs:2019hts,Hariprakash:2023tla, Suchsland:2023cmb}.

For any Hamiltonian $H$, the following decomposition into mutually non-commuting operators exists:
\begin{align}
    H = \sum_{\ell = 1}^{L}H_{\ell}
    \,,
\end{align}
where we assume that there exists an efficient implementation for each operator of the form $\mathrm{e}^{-iP_{\ell}t}$, where $t \in \mathbb{R}$. The most common example of such a decomposition is the Pauli decomposition, \emph{i.e.}, each $P_{\ell}$ is the product of a  Pauli operator and an associated coefficient. A PF approximates the full exponential $\mathrm{e}^{-iHt}$ as a product of exponentials, each involving only a single $P_{\ell}$. In particular, a $p^{\mathrm{th}}$ order PF $S_{p}(t)$ (where $p \in \mathbb{N}$) is one that satisfies the following:
\begin{align}
    \label{eq:pth_order}
    S_p(t) = \mathrm{e}^{-iHt} + \mathcal{O}(t^{p+1})
    \,.
\end{align}
The most common class of PFs arise from the Lie / Suzuki-Trotter construction defined recursively as follows:
\begin{align}
    S_1(t) &= \prod_{\ell = 1}^{L} \mathrm{e}^{-iP_{\ell}t} \\
    S_2(t) &= \prod_{\ell = 1}^{L}\mathrm{e}^{-iP_{\ell}t/2}\prod_{\ell = 1}^{L}\mathrm{e}^{-iP_{L-\ell+1}t/2} \\
    S_{2v}(t) &= S_{2v-2}(u_vt)^2S_{2v-2}((1-4u_v)t)S_{2v-2}(u_vt)^2\,,
\end{align}
where $u_v := 1/(4-4^{1/(2v-1)})$. As given for example in Ref.~\cite{Childs:2019hts}, the number of products over $\ell$ (and hence the total number of exponentials appearing in a PF formula) grows exponentially with the order $p$. Due to this, for practical applications it is common to choose ``low'' values of $v$ and hence we will assume the most common choice of $p = 2$ for the remainder of this work. 

The correction terms appearing in \cref{eq:pth_order} scale polynomially with $t$ and hence PFs provide a good approximation to the full time evolution operator for small $t$. For large $t$, we can split the evolution into smaller steps and apply a PF approximation to each individual step as follows:
\begin{align}
    S_2(t) = \left(S_2(t/\npf)\right)^{\npf}\,,
\end{align}
where $\npf\in\mathbb{N}$ is known as the \emph{Trotter number}. 
In general, for a fixed value of the order $p$, increasing $\npf$ polynomially decreases the exponentiation error (see Ref.~\cite{Childs:2019hts}) given by:
\begin{align}
    \|\Delta_{\mathrm{sim}}\| = \|\left(S_p(t/\npf)\right)^{\npf} - \mathrm{e}^{-iHt}\|
    \,.
\end{align}

We thus have that the total number of Pauli rotations, and hence elementary gate complexity which we denote by $\chi$, required to implement a $2^{\mathrm{nd}}$-order PF scales linearly with the total number of terms present in the PF formula:
\begin{align}
    \label{eq:chi}
    \chi = \mathcal{O}\left(\npf L\right)
    \,.
\end{align}
For the rest of this work we will assume that $L = \mathcal{O}(1)$ and hence \cref{eq:chi} reduces to 
\begin{align}
    \label{eq:pf_complexity}
    \chi = \mathcal{O}(\npf)
    \,.
\end{align}
Thus, under our assumptions, the parameter $\npf$ fully determines the asymptotic scaling of both the exponentiation error $\|\Delta_{\mathrm{sim}}\|$ and the algorithmic gate complexity $\chi$. During an actual simulation procedure the question arises as to how to actually choose an \emph{explicit} value of $\npf$ such that $\epsilon_{\delta} \leq \sigma_{\hat O}/\beta$. The following theorem provides a valid lower bound on the required Trotter number $\npf$:

\begin{theorem}
    \label{thm:pf_E}
     For any $\beta \in \mathbb{R}^+$, to ensure $\epsilon_{\delta} \leq \sigma_{\hat O} / \beta$ when using a $2^{\mathrm{nd}}$ order PF, it suffices to choose the Trotter number $\npf$ as follows:
    \begin{align}
        \label{eq:r_lb}
        \npf = \sqrt{\frac{2\beta\|\hat{O}(0,a)\|ft^3}{\sigma_{\mathrm{O}}}}
        \,,
    \end{align}
    where $f \in \mathbb{C}$ defined as follows:
    \begin{align}
        f := \frac{1}{12}\sum_{\ell_1 = 1}^{L} \left|\left| \left[ \sum_{\ell_3 = \ell_1}^{L} H_{\ell_3}, \left[ \sum_{\ell_2 = \ell_1+1}^{L} H_{\ell_2}, H_{\ell_1} \right] \right] \right|\right| \\
        + \frac{1}{24} \sum_{\ell_1 = 1}^{L} \left|\left| \left[ H_{\ell_1}, \sum_{\ell_2 = \ell_1+1}^{L} H_{\ell_2} \right] \right|\right|
        \,.
    \end{align}
 \end{theorem}
\begin{proof}
    We use Proposition 16 from Ref.~\cite{Childs:2019hts} with $t \rightarrow t/\npf$ and use the triangle inequality $\npf$-many times to arrive at the following upper bound for $\|\Delta_{\mathrm{sim}}\|$:
    \begin{align}
        \label{eq:trotter_error_bound}
        \|\Delta_{\mathrm{sim}}\| \leq \frac{ft^3}{\npf^2}
        \,.
    \end{align}
    We can then bound the RHS of this inequality with the upper bound from \cref{lemma:alg_err} to arrive at the lower bound in \cref{eq:r_lb}.
\end{proof}

Thus, \cref{thm:pf_E} provides a sufficient choice for the Trotter number $\npf$ required to drive the error stemming from approximating the true time evolution operator with a PF formula below the statistical error $\sigma_{\hat O}$ by a factor of $\beta$. 
This allows one to use all the machinery described in \cref{subsec:renorm_exact_te} and proceed as if the time evolution is exact. 

\subsubsection{SBTE protocol with Quantum Signal Processing}
\label{subsec:qsp}
As we did for product formulas, we start this section with a brief review of the theory of QSP. 

QSP provides an alternative algorithmic framework to performing time evolution with the benefit of providing optimal dependence on the simulation time $t$ and exponentiation error $\|\Delta_{\mathrm{sim}}\|$. 
We begin this section by reviewing how to performing Hamiltonian simulation within this framework, and then provide an analytical bound analogous to \cref{thm:pf_E} that provides a sufficient choice for the key algorithmic parameter that controls both $\|\Delta_{\mathrm{sim}}\|$ and the elementary gate complexity $\chi$. 
We note that while there are multiple variants on the core algorithmic routine, we follow the Generalized QSP framework presented in Ref.~\cite{Motlagh:2023oqc}.

The theory of QSP allows us to implement bounded finite degree polynomial transformations of unitary operators. 
In particular, let $U$ be a unitary operation and let $P(U)$ be a degree-$d$ polynomial such that
\begin{align}
    P(x)^2 \leq 1 \hspace{2mm} \forall x \in \{y\in\mathbb{C} : |y| = 1\}.
\end{align}
From Corollary 5 of Ref.~\cite{Motlagh:2023oqc}, we are then guaranteed the existence of a set of $(d+1)$-many pairs of phases $\mathcal{A} = \{(\theta_1,\phi_1),(\theta_2,\phi_2),\dots,(\theta_d,\phi_d),(\theta_{d+1},\phi_{d+1})\}$ (where $(\theta_i,\phi_i) \in \mathbb{R}^2$) and a scalar $\lambda\in\mathbb{R}$ such that the following is true:
\begin{align}
\label{eq:be_P}
 U_{\mathcal{A}}(C_{U,0}) &:= \left(\prod_{i = 2}^{d+1}R(\theta_i,\phi_i,0)C_{U,0}\right)R(\theta_1,\phi_1,\lambda) \\ &= \begin{pmatrix}
     P(U) & * \\
     * & *
 \end{pmatrix},
\end{align}
where each $*$ represents a sub-block with entries chosen to make $U_{\mathcal{A}}(C_{U,0})$ unitary, $C_{U,0}$ is the zero-controlled application of $U$, i.e.
\begin{align}
    C_{U,0} := \begin{pmatrix}
        U & 0 \\
        0 & I
    \end{pmatrix},
\end{align}
and $R(\theta_i,\phi_i,\lambda)$ represents the following $SU(2)$ rotation:
\begin{align}
    R(\theta_i, \phi_i, \lambda) = 
\begin{pmatrix}
e^{i(\lambda+\phi_i)} \cos(\theta_i) & e^{i\phi_i} \sin(\theta_i) \\
e^{i\lambda} \sin(\theta_i) & -\cos(\theta_i)
\end{pmatrix}
\otimes I.
\end{align}
In other words, the unitary $U_{\mathcal{A}}(C_{U,0})$ encodes the polynomial $P(U)$ in its top left block. 
Thus, the goal of QSP is to find the appropriate set of phases $\mathcal{A}$ and scalar $\lambda$ such that the resulting QSP unitary encodes a desired polynomial. 
As shown in Ref.~\cite{Motlagh:2023oqc}, the set $\mathcal{A}$ can be determined efficiently and accurately using an optimization algorithm for $d$ up to $2^{24}$.

Our approach to apply QSP to the task of time evolution involves the following steps:
\begin{enumerate}
    \item Construct a unitary that provides query access to the eigenvalues of the Hamiltonian $H$.
    \item Determine the set $\mathcal{A}$ necessary to implement a polynomial transformation of these eigenvalues such that the resulting polynomial sufficiently approximates the true time evolution operator
\end{enumerate}

We accomplish the first of these tasks as follows. Given a Hamiltonian $H$, the \emph{block-encoding} of $H$ (denoted by $U_H$) is a unitary that encodes $H$ in (typically) its top left block:
 \begin{align}
     \label{eq:be}
     U_H = \begin{pmatrix}
         H / \alpha_H & * \\
         * & *
     \end{pmatrix},
 \end{align}
where $\alpha_H \in \mathbb{R}^+$ is defined such that $\|H/\alpha_H\| \leq 1$ and $*$ refers to a block of entries that make $U_H$ unitary. 
The size of the block $*$ corresponds to the number of ancilla qubits used to construct $U_H$. 
Note that \cref{eq:be_P} provides another example of a block-encoding, specifically of the polynomial $P(U)$. 
The most general construction of $U_H$ can be performed by the linear combination of unitaries (LCU) method~\cite{Childs:2012gwh}. 
However, this can prove to be computationally challenging. 
In many instances, including Hamiltonian LGTs, the structure of $H$ can lend itself to more efficient block-encoding constructions \cite{Kane:2024odt,Hariprakash:2023tla,Rhodes:2024zbr,Simon:2025pbo}.

Given $U_H$, we then use it as input to the \emph{qubitization} procedure presented in Ref.~\cite{Low:2016znh} to construct a unitary $W_H$ (known as the walk operator or iterate) designed such that the eigenphases of $W_H$ encode the eigenvalues of the rescaled Hamiltonian $H/\alpha_H$ and thus provide the desired query access to the eigenvalues of $H$. 
Thus, by setting $U \rightarrow W_H$ in \cref{eq:be_P} we can construct polynomial transformations of $W_H$ and hence also of $H/\alpha_H$. 
In particular, we seek a polynomial transformation that satisfies:
\begin{align}
     P(W_H) = \exp(-iHt) = \exp\left(-i\left(\frac{H}{\alpha_H}\right)\alpha_H t\right).
\end{align}
If we let $x\in\mathrm{spec}(H/\alpha_H)$, one can use the Jacobi-Anger expansion \cite{Low:2016sck,Gilyen:2018khw} to obtain the following polynomial decomposition
\begin{align}
    \mathrm{e}^{-i \alpha_H x t} &= J_0(\alpha_H\, t) + 2\sum_{k \hspace{1mm} \mathrm{even}, > 0}^{\infty} \left(-i\right)^{k/2}J_k(\alpha_H\,t)T_k(x) \hspace{1mm} \nonumber \\ &\qquad + 2i\sum_{k \hspace{1mm} \mathrm{odd}, > 0}^{\infty}\left(-i\right)^{(k-1)/2}J_k(\alpha_H\,t)T_k(x)
    \,,
\end{align}
where $J_k$ is the $k^{\mathrm{th}}$ order Bessel function of the first kind and $T_k$ is the $k^\mathrm{th}$ Chebyshev polynomial of the first kind. 
If we construct a finite degree approximation to this polynomial decomposition, we can then encode an approximation to the true time evolution operator in the top left block of a unitary of the form \cref{eq:be_P}. 
To accomplish this, we consider the following truncated decomposition:
\begin{align}
    \label{eq:qsp_kmax}
    \mathrm{e}^{-i \alpha_H x t} &= J_0(\alpha_H\, t) + 2\sum_{k \hspace{1mm} \mathrm{even}, > 0}^{N_{\mathrm{QSP}}} \left(-i\right)^{k/2}J_k(\alpha_H\, t)T_k(x) \hspace{1mm} \nonumber \\ &\qquad+2i\sum_{k \hspace{1mm} \mathrm{odd}, > 0}^{N_{\mathrm{QSP}}}\left(-i\right)^{(k-1)/2}J_k(\alpha_H\,t)T_k(x)
    \,.
\end{align}

Thus, $N_{\mathrm{QSP}}$ determines the degree of the polynomial used to approximate the true time evolution operator, and from \cref{eq:be_P} we see that the elementary gate complexity $\chi$ scales linearly with $N_{\mathrm{QSP}}$:
\begin{align}
    \label{eq:qsp_complexity}
    \chi = \mathcal{O}(N_{\mathrm{QSP}}).
\end{align}
Thus, similar to the parameter $N_{\mathrm{PF}}$ from \cref{subsec:pf}, $N_{\mathrm{QSP}}$ is the key algorithmic parameter for QSP that determines both the exponentiation error $\|\Delta_{\mathrm{sim}}\|$ (as shown for instance in Ref.~\cite{Low:2016sck}, $\|\Delta_{\rm sim}\|$ decreases \emph{exponentially} as we increase $N_{\mathrm{QSP}}$) as well as $\chi$. In the language of \cref{sec:syst_unc}, we identify $\delta = 1/N_{\mathrm{QSP}}$. The following theorem provides an explicit choice for value of $N_{\mathrm{QSP}}$ required to satisfy the condition $\epsilon_{\delta} \leq \sigma_{\hat O}/\beta$ and hence perform an actual simulation: 
\begin{theorem}
    \label{thm:qsp_E}
    For any $\beta \in \mathbb{R}^+$, to ensure $\epsilon_{\delta} \leq \sigma_{\hat O}/\beta$ when using the Generalized QSP algorithmic routine, it suffices to choose the truncation parameter $N_{\mathrm{QSP}}$ shown in \cref{eq:qsp_kmax} as follows:
    \begin{align}
        N_{\mathrm{QSP}} \geq \frac{e\alpha_Ht}{2} + \log\left(2\beta\|\hat{O}(0,a)\|\right) + \log\left(\frac{1}{\sigma_{\mathrm{O}}}\right).
    \end{align}
\end{theorem}
\begin{proof}
    From Refs.~\cite{Berry:2015hst,Low:2016znh,Gilyen:2018khw}, we have that
    \begin{align}
        \label{eq:qsp_error_bound}
        N_{\mathrm{QSP}} \geq \frac{e\alpha_Ht}{2} + \log\left(\frac{1}{\|\Delta_{\mathrm{sim}}\|}\right).
    \end{align}
    for which the proof is based on the super-exponentially decaying upper bound on the norms of Bessel functions with respect to the order. We can then use \cref{lemma:alg_err} to arrive at the desired result.
\end{proof}

\subsection{Comparisons}
\label{ssec:comparison}
The most striking differences between the two algorithmic kernels we consider for Hamiltonian simulation are associated with how the rigorous error bounds scale with key parameters $N_{\mathrm{PF}}$ and $N_{\mathrm{QSP}}$ functionally depend on $t$, $\sigma_{\hat O}$, $\beta$, and $\|\hat{O}(0,a)\|$. 
From \cref{thm:pf_E,thm:qsp_E} we see that $N_{\mathrm{QSP}}$ asymptotically scales more favorably than $N_{\mathrm{PF}}$ with respect to all these quantities. 
In particular from \cref{thm:pf_E} we see that $N_{\mathrm{PF}}$ scales polynomially with $1/\sigma_{\hat O}$ whereas from \cref{thm:qsp_E} we see that $N_{\mathrm{QSP}}$ scales \emph{logarithmically} with $1/\sigma_{\hat O}$. 
Moreover, while $N_{\mathrm{PF}}$ scales polynomially with the product $\beta\|\hat{O}(0,a)\|$, $N_{\mathrm{QSP}}$ again scales \emph{logarithmically} with the product $\beta\|\hat{O}(0,a)\|$. 
Furthermore, the additional logarithmic overhead when using QSP is \emph{additive} relative to the cost of simulating time evolution to error $\sigma_{\hat O}$. 
Because the scale factor must satisfy $\alpha_H \geq \|H\|$, it generally grows with volume, which implies that the overhead is entirely negligible for any physically interesting operator.

To better understand the additional overhead for both methods, we consider two physically relevant classes of operators $\hat{O}(0,a)$.
In particular, consider the overhead for local operators, \emph{e.g.}, the local quark condensate, and operators that involve summing over the entire lattice, \emph{e.g.}, the hadronic tensor. 
In the case of local operators, the operator norm $\|\hat O(0,a)\|$ is independent of the volume and may introduce only a small overhead when using PFs.
Simulating non-local operators whose norm $\|\hat O(0,a)\|$ grows with the volume, however, will result in a significant cost increase for simulations via PFs.
While it is possible that this cost scaling could be improved with more detailed error analysis, if one instead uses QSP, the overhead can be guaranteed to be entirely negligible \emph{a priori}.

Additionally, as we show for the case of a $\phi^4$ Hamiltonian in \cref{app:comparing_bounds}, the QSP error bound used to prove \cref{thm:qsp_E} is generally known to be tighter than the analogous PF error bound. 
Thus, employing the SBTE protocol we propose in this work, the QSP error bound would then enable us to choose smaller values of $N_{\mathrm{QSP}}$ than the values of $N_{\mathrm{PF}}$ implied by the PF bound relative to the ideal values of these quantities (i.e. the smallest values of $N_{\mathrm{QSP}},N_{\mathrm{PF}}$ for which the desired error bound is satisfied). 

While employing rigorous error bounds to determine the value of $\delta$ that guarantees $\epsilon_\delta \leq \sigma_{\hat O}/\beta$ has the advantage of only needing to perform a single simulation, an alternative approach is to simply perform the calculation at smaller values of $\delta$ until the value $\langle \hat{O}(\delta)\rangle$ no longer changes within error bars.
Such an approach could result in savings if the error bounds are known to be loose, as is the case with PFs. 

From this discussion it should be clear that, for a given problem, it is not immediately obvious which technique would be better suited for performing the SBTE protocol proposed in this work and result in an overall lower gate complexity. 
The answer to this question requires knowledge of the prefactors appearing in \cref{eq:pf_complexity,eq:qsp_complexity} as well as the quantities $f$ (\cref{thm:pf_E}) and $\alpha_H$ (\cref{eq:be}). 
The major contribution to the QSP prefactor is associated with the task of constructing the block-encoding for a given Hamiltonian (which also determines $\alpha_H$). 
This task is highly problem dependent and hence we expect the more efficient algorithmic kernel to also change depending on the specific Hamiltonian as well as the specific values of $\|\Delta_{\mathrm{sim}}\|,t,$ and the size of the lattice. 
The work in Refs.~\cite{Hariprakash:2023tla,Kane:2024odt} presents, for example, values of $\|\Delta_{\mathrm{sim}}\|$ and $t$ for which the \emph{combined} task of block-encoding and time evolution via QSP proves to be more efficient than time evolution by a PF for the case of the $\phi^4$ Hamiltonian at specific lattice sizes. 
Thus, for a specific problem/Hamiltonian one can take advantage of any structure inherent to the problem to construct efficient block-encodings and either analytically or numerically (by extrapolation) determine the parameter ranges corresponding to the either QSP or PF being the more efficient algorithmic choice.

\section{Discussion and Conclusion  \label{sec:conclusion}}

In this work, we have presented an approach, called the SBTE protocol, to taking the continuum limit in dynamical simulations of Hamiltonian lattice field theories. 
Our framework is conceptually simpler than previous methods based on a renormalization procedure, which were shown to allow classical calculations using the heat-kernel action to obtain the renormalization trajectory for time evolution in pure Yang-Mills gauge theory. 
One advantage of this new protocol is that it works for any approximation to the time evolution operator, but more important is that it works for any Hamiltonian, in particular for the strong interaction including fermions, which is one of the most important use cases for Hamiltonian simulations. 
The SBTE protocol is based on the idea that as long as the systematic error from approximate time evolution is negligible compared to the statistical uncertainty of the simulation, its effect on the continuum limit can be safely ignored.
By exploiting rigorous upper bounds on the error incurred from approximate time evolution, one can \emph{a priori} determine the quantum gate cost required to achieve this.

To demonstrate our approach, we considered simulation via product formulas and methods based on Quantum Signal Processing.
Using existing error bounds, we showed that, while driving the time evolution error below the statistical threshold can result in significant overhead when using PFs due to the sub-optimal scaling $\mathcal{O}(t^{1+1/p}\epsilon^{-1/p})$, this additional cost is negligible when using QSP based methods due to its provably optimal scaling $\mathcal{O}(\alpha_Ht+\log \frac{1}{\epsilon})$, in particular due to the \emph{additive} $\log \frac{1}{\epsilon}$ scaling with error; for algorithms with multiplicative error scaling of the form $\sim t \log \frac{1}{\epsilon}$, the overhead will be smaller than for PFs, but still non-negligible.

Rather than using error bounds in the SBTE protocol, we discussed an alternative approach where one simply performs the calculation with smaller values of $\delta$ until the value $\langle \hat O(\delta) \rangle$ no longer changes within error bars.
This approach is most likely to be beneficial when the error bounds are known to be loose, as is the case with product-formulas.
In general, the computational cost in the SBTE protocol is simply given by the computational cost of ensuring the systematic uncertainties introduced by an approximation to the time-evolution operator are below the statistical ones.

In addition to developing this general alternative approach, we provided a detailed discussion of the proposal for using classical resources~\cite{Carena:2021ltu} to help reduce the quantum cost of determining the renormalization trajectory.
In this discussion we highlighted that for pure gauge theories these methods  require classical simulations with complicated actions.
Once fermions are included, no classical actions are known for this approach to work.
It is also important to keep in mind that, because performing time evolution requires quantum computers with similar capabilities as would be needed for tuning bare parameters and scale setting, such methods do not allow one to get away with using smaller or noisier quantum devices.
These approaches can, however, trade total run-time on quantum computers for run-time on classical computers.

Our findings suggest several directions for future research. 
One promising avenue is to develop an explicit effective Hamiltonian description for QSP-based time evolution, which could provide deeper insight into the renormalization process and potentially inspire new strategies similar to the PF-based approach. 
Additionally, a comprehensive end-to-end renormalization procedure—including parameter tuning, scale setting, and operator renormalization—remains to be established for general simulation algorithms. 
There are many different renormalization schemes one can imagine using, and detailed comparisons of the computational cost are necessary to understand which approach is best suited for a given problem.

Ultimately, our method lays the groundwork for systematic and fair comparisons between different quantum simulation algorithms for lattice field theory simulations, enabling one to assess the total quantum resources required to obtain continuum physics. 
This, in turn, will be crucial for guiding the development of quantum algorithms and hardware toward achieving quantum advantage in this challenging domain.

\section*{Acknowledgements}

We would like to thank Clement Charles, Anthony Ciavarella, Dorota Grabowska and Chung-Chun Hsieh, Neel Modi, Henry Lamm and Wanqiang Liu for many useful discussions. 
In particular Henry Lamm and Wanqiang Liu provided many insights into Ref.~\cite{Carena:2021ltu}.
The work of C.F.K. was supported by the U.S. Department of Energy, Office of Science, Accelerated Research in Quantum Computing, Fundamental Algorithmic Research toward Quantum Utility (FAR-Qu). C.F.K. was further supported in part by the Department of Physics, Maryland Center for Fundamental Physics, and the College of Computer, Mathematical, and Natural Sciences at the University of Maryland, College Park, as well as the U.S. Department of Energy, Office of Science, Office of Advanced Scientific Computing Research, Department of Energy Computational Science Graduate Fellowship under Award Number DE-SC0020347.
C.W.B and S.H. are supported by the U.S.~Department of Energy, Office of Science under contract DE-AC02-05CH11231, primarily through Quantum Information Science Enabled Discovery (QuantISED) for High Energy Physics (KA2401032). C.W.B also acknoledges support from the U.S. Department of Energy, Office of Science, National Quantum Information Science Research Centers, Quantum Systems Accelerator. 

\bibliographystyle{apsrev4-1}
\bibliography{references}

\begin{thebibliography}{165}%
\makeatletter
\providecommand \@ifxundefined [1]{%
 \@ifx{#1\undefined}
}%
\providecommand \@ifnum [1]{%
 \ifnum #1\expandafter \@firstoftwo
 \else \expandafter \@secondoftwo
 \fi
}%
\providecommand \@ifx [1]{%
 \ifx #1\expandafter \@firstoftwo
 \else \expandafter \@secondoftwo
 \fi
}%
\providecommand \natexlab [1]{#1}%
\providecommand \enquote  [1]{``#1''}%
\providecommand \bibnamefont  [1]{#1}%
\providecommand \bibfnamefont [1]{#1}%
\providecommand \citenamefont [1]{#1}%
\providecommand \href@noop [0]{\@secondoftwo}%
\providecommand \href [0]{\begingroup \@sanitize@url \@href}%
\providecommand \@href[1]{\@@startlink{#1}\@@href}%
\providecommand \@@href[1]{\endgroup#1\@@endlink}%
\providecommand \@sanitize@url [0]{\catcode `\\12\catcode `\$12\catcode `\&12\catcode `\#12\catcode `\^12\catcode `\_12\catcode `\%12\relax}%
\providecommand \@@startlink[1]{}%
\providecommand \@@endlink[0]{}%
\providecommand \url  [0]{\begingroup\@sanitize@url \@url }%
\providecommand \@url [1]{\endgroup\@href {#1}{\urlprefix }}%
\providecommand \urlprefix  [0]{URL }%
\providecommand \Eprint [0]{\href }%
\providecommand \doibase [0]{http://dx.doi.org/}%
\providecommand \selectlanguage [0]{\@gobble}%
\providecommand \bibinfo  [0]{\@secondoftwo}%
\providecommand \bibfield  [0]{\@secondoftwo}%
\providecommand \translation [1]{[#1]}%
\providecommand \BibitemOpen [0]{}%
\providecommand \bibitemStop [0]{}%
\providecommand \bibitemNoStop [0]{.\EOS\space}%
\providecommand \EOS [0]{\spacefactor3000\relax}%
\providecommand \BibitemShut  [1]{\csname bibitem#1\endcsname}%
\let\auto@bib@innerbib\@empty
\bibitem [{\citenamefont {Jordan}\ \emph {et~al.}(2012)\citenamefont {Jordan}, \citenamefont {Lee},\ and\ \citenamefont {Preskill}}]{Jordan:2012xnu}%
  \BibitemOpen
  \bibfield  {author} {\bibinfo {author} {\bibfnamefont {S.~P.}\ \bibnamefont {Jordan}}, \bibinfo {author} {\bibfnamefont {K.~S.~M.}\ \bibnamefont {Lee}}, \ and\ \bibinfo {author} {\bibfnamefont {J.}~\bibnamefont {Preskill}},\ }\href {\doibase 10.1126/science.1217069} {\bibfield  {journal} {\bibinfo  {journal} {Science}\ }\textbf {\bibinfo {volume} {336}},\ \bibinfo {pages} {1130} (\bibinfo {year} {2012})},\ \Eprint {http://arxiv.org/abs/1111.3633} {arXiv:1111.3633 [quant-ph]} \BibitemShut {NoStop}%
\bibitem [{\citenamefont {Bauer}\ \emph {et~al.}(2023{\natexlab{a}})\citenamefont {Bauer} \emph {et~al.}}]{Bauer:2022hpo}%
  \BibitemOpen
  \bibfield  {author} {\bibinfo {author} {\bibfnamefont {C.~W.}\ \bibnamefont {Bauer}} \emph {et~al.},\ }\href {\doibase 10.1103/PRXQuantum.4.027001} {\bibfield  {journal} {\bibinfo  {journal} {PRX Quantum}\ }\textbf {\bibinfo {volume} {4}},\ \bibinfo {pages} {027001} (\bibinfo {year} {2023}{\natexlab{a}})},\ \Eprint {http://arxiv.org/abs/2204.03381} {arXiv:2204.03381 [quant-ph]} \BibitemShut {NoStop}%
\bibitem [{\citenamefont {Bauer}\ \emph {et~al.}(2023{\natexlab{b}})\citenamefont {Bauer}, \citenamefont {Davoudi}, \citenamefont {Klco},\ and\ \citenamefont {Savage}}]{Bauer:2023qgm}%
  \BibitemOpen
  \bibfield  {author} {\bibinfo {author} {\bibfnamefont {C.~W.}\ \bibnamefont {Bauer}}, \bibinfo {author} {\bibfnamefont {Z.}~\bibnamefont {Davoudi}}, \bibinfo {author} {\bibfnamefont {N.}~\bibnamefont {Klco}}, \ and\ \bibinfo {author} {\bibfnamefont {M.~J.}\ \bibnamefont {Savage}},\ }\href {\doibase 10.1038/s42254-023-00599-8} {\bibfield  {journal} {\bibinfo  {journal} {Nature Rev. Phys.}\ }\textbf {\bibinfo {volume} {5}},\ \bibinfo {pages} {420} (\bibinfo {year} {2023}{\natexlab{b}})},\ \Eprint {http://arxiv.org/abs/2404.06298} {arXiv:2404.06298 [hep-ph]} \BibitemShut {NoStop}%
\bibitem [{\citenamefont {Di~Meglio}\ \emph {et~al.}(2024)\citenamefont {Di~Meglio} \emph {et~al.}}]{DiMeglio:2023nsa}%
  \BibitemOpen
  \bibfield  {author} {\bibinfo {author} {\bibfnamefont {A.}~\bibnamefont {Di~Meglio}} \emph {et~al.},\ }\href {\doibase 10.1103/PRXQuantum.5.037001} {\bibfield  {journal} {\bibinfo  {journal} {PRX Quantum}\ }\textbf {\bibinfo {volume} {5}},\ \bibinfo {pages} {037001} (\bibinfo {year} {2024})},\ \Eprint {http://arxiv.org/abs/2307.03236} {arXiv:2307.03236 [quant-ph]} \BibitemShut {NoStop}%
\bibitem [{\citenamefont {Banerjee}\ \emph {et~al.}(2012)\citenamefont {Banerjee}, \citenamefont {Dalmonte}, \citenamefont {Muller}, \citenamefont {Rico}, \citenamefont {Stebler}, \citenamefont {Wiese},\ and\ \citenamefont {Zoller}}]{Banerjee:2012pg}%
  \BibitemOpen
  \bibfield  {author} {\bibinfo {author} {\bibfnamefont {D.}~\bibnamefont {Banerjee}}, \bibinfo {author} {\bibfnamefont {M.}~\bibnamefont {Dalmonte}}, \bibinfo {author} {\bibfnamefont {M.}~\bibnamefont {Muller}}, \bibinfo {author} {\bibfnamefont {E.}~\bibnamefont {Rico}}, \bibinfo {author} {\bibfnamefont {P.}~\bibnamefont {Stebler}}, \bibinfo {author} {\bibfnamefont {U.~J.}\ \bibnamefont {Wiese}}, \ and\ \bibinfo {author} {\bibfnamefont {P.}~\bibnamefont {Zoller}},\ }\href {\doibase 10.1103/PhysRevLett.109.175302} {\bibfield  {journal} {\bibinfo  {journal} {Phys. Rev. Lett.}\ }\textbf {\bibinfo {volume} {109}},\ \bibinfo {pages} {175302} (\bibinfo {year} {2012})},\ \Eprint {http://arxiv.org/abs/1205.6366} {arXiv:1205.6366 [cond-mat.quant-gas]} \BibitemShut {NoStop}%
\bibitem [{\citenamefont {Hauke}\ \emph {et~al.}(2013)\citenamefont {Hauke}, \citenamefont {Marcos}, \citenamefont {Dalmonte},\ and\ \citenamefont {Zoller}}]{Hauke:2013jga}%
  \BibitemOpen
  \bibfield  {author} {\bibinfo {author} {\bibfnamefont {P.}~\bibnamefont {Hauke}}, \bibinfo {author} {\bibfnamefont {D.}~\bibnamefont {Marcos}}, \bibinfo {author} {\bibfnamefont {M.}~\bibnamefont {Dalmonte}}, \ and\ \bibinfo {author} {\bibfnamefont {P.}~\bibnamefont {Zoller}},\ }\href {\doibase 10.1103/PhysRevX.3.041018} {\bibfield  {journal} {\bibinfo  {journal} {Phys. Rev. X}\ }\textbf {\bibinfo {volume} {3}},\ \bibinfo {pages} {041018} (\bibinfo {year} {2013})},\ \Eprint {http://arxiv.org/abs/1306.2162} {arXiv:1306.2162 [cond-mat.quant-gas]} \BibitemShut {NoStop}%
\bibitem [{\citenamefont {Zohar}\ \emph {et~al.}(2013{\natexlab{a}})\citenamefont {Zohar}, \citenamefont {Cirac},\ and\ \citenamefont {Reznik}}]{Zohar:2013zla}%
  \BibitemOpen
  \bibfield  {author} {\bibinfo {author} {\bibfnamefont {E.}~\bibnamefont {Zohar}}, \bibinfo {author} {\bibfnamefont {J.~I.}\ \bibnamefont {Cirac}}, \ and\ \bibinfo {author} {\bibfnamefont {B.}~\bibnamefont {Reznik}},\ }\href {\doibase 10.1103/PhysRevA.88.023617} {\bibfield  {journal} {\bibinfo  {journal} {Phys. Rev. A}\ }\textbf {\bibinfo {volume} {88}},\ \bibinfo {pages} {023617} (\bibinfo {year} {2013}{\natexlab{a}})},\ \Eprint {http://arxiv.org/abs/1303.5040} {arXiv:1303.5040 [quant-ph]} \BibitemShut {NoStop}%
\bibitem [{\citenamefont {K\"uhn}\ \emph {et~al.}(2014)\citenamefont {K\"uhn}, \citenamefont {Cirac},\ and\ \citenamefont {Ba\~nuls}}]{Kuhn:2014rha}%
  \BibitemOpen
  \bibfield  {author} {\bibinfo {author} {\bibfnamefont {S.}~\bibnamefont {K\"uhn}}, \bibinfo {author} {\bibfnamefont {J.~I.}\ \bibnamefont {Cirac}}, \ and\ \bibinfo {author} {\bibfnamefont {M.-C.}\ \bibnamefont {Ba\~nuls}},\ }\href {\doibase 10.1103/PhysRevA.90.042305} {\bibfield  {journal} {\bibinfo  {journal} {Phys. Rev. A}\ }\textbf {\bibinfo {volume} {90}},\ \bibinfo {pages} {042305} (\bibinfo {year} {2014})},\ \Eprint {http://arxiv.org/abs/1407.4995} {arXiv:1407.4995 [quant-ph]} \BibitemShut {NoStop}%
\bibitem [{\citenamefont {Kasper}\ \emph {et~al.}(2016)\citenamefont {Kasper}, \citenamefont {Hebenstreit}, \citenamefont {Oberthaler},\ and\ \citenamefont {Berges}}]{Kasper:2015cca}%
  \BibitemOpen
  \bibfield  {author} {\bibinfo {author} {\bibfnamefont {V.}~\bibnamefont {Kasper}}, \bibinfo {author} {\bibfnamefont {F.}~\bibnamefont {Hebenstreit}}, \bibinfo {author} {\bibfnamefont {M.}~\bibnamefont {Oberthaler}}, \ and\ \bibinfo {author} {\bibfnamefont {J.}~\bibnamefont {Berges}},\ }\href {\doibase 10.1016/j.physletb.2016.07.036} {\bibfield  {journal} {\bibinfo  {journal} {Phys. Lett. B}\ }\textbf {\bibinfo {volume} {760}},\ \bibinfo {pages} {742} (\bibinfo {year} {2016})},\ \Eprint {http://arxiv.org/abs/1506.01238} {arXiv:1506.01238 [cond-mat.quant-gas]} \BibitemShut {NoStop}%
\bibitem [{\citenamefont {Zohar}\ \emph {et~al.}(2016)\citenamefont {Zohar}, \citenamefont {Cirac},\ and\ \citenamefont {Reznik}}]{Zohar:2015hwa}%
  \BibitemOpen
  \bibfield  {author} {\bibinfo {author} {\bibfnamefont {E.}~\bibnamefont {Zohar}}, \bibinfo {author} {\bibfnamefont {J.~I.}\ \bibnamefont {Cirac}}, \ and\ \bibinfo {author} {\bibfnamefont {B.}~\bibnamefont {Reznik}},\ }\href {\doibase 10.1088/0034-4885/79/1/014401} {\bibfield  {journal} {\bibinfo  {journal} {Rept. Prog. Phys.}\ }\textbf {\bibinfo {volume} {79}},\ \bibinfo {pages} {014401} (\bibinfo {year} {2016})},\ \Eprint {http://arxiv.org/abs/1503.02312} {arXiv:1503.02312 [quant-ph]} \BibitemShut {NoStop}%
\bibitem [{\citenamefont {{Martinez}}\ \emph {et~al.}(2016)\citenamefont {{Martinez}}, \citenamefont {{Muschik}}, \citenamefont {{Schindler}}, \citenamefont {{Nigg}}, \citenamefont {{Erhard}}, \citenamefont {{Heyl}}, \citenamefont {{Hauke}}, \citenamefont {{Dalmonte}}, \citenamefont {{Monz}}, \citenamefont {{Zoller}},\ and\ \citenamefont {{Blatt}}}]{Martinez:2016yna}%
  \BibitemOpen
  \bibfield  {author} {\bibinfo {author} {\bibfnamefont {E.~A.}\ \bibnamefont {{Martinez}}}, \bibinfo {author} {\bibfnamefont {C.~A.}\ \bibnamefont {{Muschik}}}, \bibinfo {author} {\bibfnamefont {P.}~\bibnamefont {{Schindler}}}, \bibinfo {author} {\bibfnamefont {D.}~\bibnamefont {{Nigg}}}, \bibinfo {author} {\bibfnamefont {A.}~\bibnamefont {{Erhard}}}, \bibinfo {author} {\bibfnamefont {M.}~\bibnamefont {{Heyl}}}, \bibinfo {author} {\bibfnamefont {P.}~\bibnamefont {{Hauke}}}, \bibinfo {author} {\bibfnamefont {M.}~\bibnamefont {{Dalmonte}}}, \bibinfo {author} {\bibfnamefont {T.}~\bibnamefont {{Monz}}}, \bibinfo {author} {\bibfnamefont {P.}~\bibnamefont {{Zoller}}}, \ and\ \bibinfo {author} {\bibfnamefont {R.}~\bibnamefont {{Blatt}}},\ }\href {\doibase 10.1038/nature18318} {\bibfield  {journal} {\bibinfo  {journal} {Nature}\ }\textbf {\bibinfo {volume} {534}},\ \bibinfo {pages} {516} (\bibinfo {year} {2016})},\ \Eprint {http://arxiv.org/abs/1605.04570} {arXiv:1605.04570 [quant-ph]} \BibitemShut {NoStop}%
\bibitem [{\citenamefont {Yang}\ \emph {et~al.}(2016)\citenamefont {Yang}, \citenamefont {Giri}, \citenamefont {Johanning}, \citenamefont {Wunderlich}, \citenamefont {Zoller},\ and\ \citenamefont {Hauke}}]{Yang:2016hjn}%
  \BibitemOpen
  \bibfield  {author} {\bibinfo {author} {\bibfnamefont {D.}~\bibnamefont {Yang}}, \bibinfo {author} {\bibfnamefont {G.~S.}\ \bibnamefont {Giri}}, \bibinfo {author} {\bibfnamefont {M.}~\bibnamefont {Johanning}}, \bibinfo {author} {\bibfnamefont {C.}~\bibnamefont {Wunderlich}}, \bibinfo {author} {\bibfnamefont {P.}~\bibnamefont {Zoller}}, \ and\ \bibinfo {author} {\bibfnamefont {P.}~\bibnamefont {Hauke}},\ }\href {\doibase 10.1103/PhysRevA.94.052321} {\bibfield  {journal} {\bibinfo  {journal} {Phys. Rev. A}\ }\textbf {\bibinfo {volume} {94}},\ \bibinfo {pages} {052321} (\bibinfo {year} {2016})},\ \Eprint {http://arxiv.org/abs/1604.03124} {arXiv:1604.03124 [quant-ph]} \BibitemShut {NoStop}%
\bibitem [{\citenamefont {Kokail}\ \emph {et~al.}(2019)\citenamefont {Kokail} \emph {et~al.}}]{Kokail:2018eiw}%
  \BibitemOpen
  \bibfield  {author} {\bibinfo {author} {\bibfnamefont {C.}~\bibnamefont {Kokail}} \emph {et~al.},\ }\href {\doibase 10.1038/s41586-019-1177-4} {\bibfield  {journal} {\bibinfo  {journal} {Nature}\ }\textbf {\bibinfo {volume} {569}},\ \bibinfo {pages} {355} (\bibinfo {year} {2019})},\ \Eprint {http://arxiv.org/abs/1810.03421} {arXiv:1810.03421 [quant-ph]} \BibitemShut {NoStop}%
\bibitem [{\citenamefont {Klco}\ \emph {et~al.}(2018)\citenamefont {Klco}, \citenamefont {Dumitrescu}, \citenamefont {McCaskey}, \citenamefont {Morris}, \citenamefont {Pooser}, \citenamefont {Sanz}, \citenamefont {Solano}, \citenamefont {Lougovski},\ and\ \citenamefont {Savage}}]{Klco:2018kyo}%
  \BibitemOpen
  \bibfield  {author} {\bibinfo {author} {\bibfnamefont {N.}~\bibnamefont {Klco}}, \bibinfo {author} {\bibfnamefont {E.~F.}\ \bibnamefont {Dumitrescu}}, \bibinfo {author} {\bibfnamefont {A.~J.}\ \bibnamefont {McCaskey}}, \bibinfo {author} {\bibfnamefont {T.~D.}\ \bibnamefont {Morris}}, \bibinfo {author} {\bibfnamefont {R.~C.}\ \bibnamefont {Pooser}}, \bibinfo {author} {\bibfnamefont {M.}~\bibnamefont {Sanz}}, \bibinfo {author} {\bibfnamefont {E.}~\bibnamefont {Solano}}, \bibinfo {author} {\bibfnamefont {P.}~\bibnamefont {Lougovski}}, \ and\ \bibinfo {author} {\bibfnamefont {M.~J.}\ \bibnamefont {Savage}},\ }\href {\doibase 10.1103/PhysRevA.98.032331} {\bibfield  {journal} {\bibinfo  {journal} {Phys. Rev. A}\ }\textbf {\bibinfo {volume} {98}},\ \bibinfo {pages} {032331} (\bibinfo {year} {2018})},\ \Eprint {http://arxiv.org/abs/1803.03326} {arXiv:1803.03326 [quant-ph]} \BibitemShut {NoStop}%
\bibitem [{\citenamefont {Lu}\ \emph {et~al.}(2019)\citenamefont {Lu} \emph {et~al.}}]{Lu:2018pjk}%
  \BibitemOpen
  \bibfield  {author} {\bibinfo {author} {\bibfnamefont {H.-H.}\ \bibnamefont {Lu}} \emph {et~al.},\ }\href {\doibase 10.1103/PhysRevA.100.012320} {\bibfield  {journal} {\bibinfo  {journal} {Phys. Rev. A}\ }\textbf {\bibinfo {volume} {100}},\ \bibinfo {pages} {012320} (\bibinfo {year} {2019})},\ \Eprint {http://arxiv.org/abs/1810.03959} {arXiv:1810.03959 [quant-ph]} \BibitemShut {NoStop}%
\bibitem [{\citenamefont {Kaplan}\ and\ \citenamefont {Stryker}(2020)}]{Kaplan:2018vnj}%
  \BibitemOpen
  \bibfield  {author} {\bibinfo {author} {\bibfnamefont {D.~B.}\ \bibnamefont {Kaplan}}\ and\ \bibinfo {author} {\bibfnamefont {J.~R.}\ \bibnamefont {Stryker}},\ }\href {\doibase 10.1103/PhysRevD.102.094515} {\bibfield  {journal} {\bibinfo  {journal} {Phys. Rev. D}\ }\textbf {\bibinfo {volume} {102}},\ \bibinfo {pages} {094515} (\bibinfo {year} {2020})},\ \Eprint {http://arxiv.org/abs/1806.08797} {arXiv:1806.08797 [hep-lat]} \BibitemShut {NoStop}%
\bibitem [{\citenamefont {Mil}\ \emph {et~al.}(2020)\citenamefont {Mil}, \citenamefont {Zache}, \citenamefont {Hegde}, \citenamefont {Xia}, \citenamefont {Bhatt}, \citenamefont {Oberthaler}, \citenamefont {Hauke}, \citenamefont {Berges},\ and\ \citenamefont {Jendrzejewski}}]{Mil:2019pbt}%
  \BibitemOpen
  \bibfield  {author} {\bibinfo {author} {\bibfnamefont {A.}~\bibnamefont {Mil}}, \bibinfo {author} {\bibfnamefont {T.~V.}\ \bibnamefont {Zache}}, \bibinfo {author} {\bibfnamefont {A.}~\bibnamefont {Hegde}}, \bibinfo {author} {\bibfnamefont {A.}~\bibnamefont {Xia}}, \bibinfo {author} {\bibfnamefont {R.~P.}\ \bibnamefont {Bhatt}}, \bibinfo {author} {\bibfnamefont {M.~K.}\ \bibnamefont {Oberthaler}}, \bibinfo {author} {\bibfnamefont {P.}~\bibnamefont {Hauke}}, \bibinfo {author} {\bibfnamefont {J.}~\bibnamefont {Berges}}, \ and\ \bibinfo {author} {\bibfnamefont {F.}~\bibnamefont {Jendrzejewski}},\ }\href {\doibase 10.1126/science.aaz5312} {\bibfield  {journal} {\bibinfo  {journal} {Science}\ }\textbf {\bibinfo {volume} {367}},\ \bibinfo {pages} {1128} (\bibinfo {year} {2020})},\ \Eprint {http://arxiv.org/abs/1909.07641} {arXiv:1909.07641 [cond-mat.quant-gas]} \BibitemShut {NoStop}%
\bibitem [{\citenamefont {Davoudi}\ \emph {et~al.}(2020)\citenamefont {Davoudi}, \citenamefont {Hafezi}, \citenamefont {Monroe}, \citenamefont {Pagano}, \citenamefont {Seif},\ and\ \citenamefont {Shaw}}]{Davoudi:2019bhy}%
  \BibitemOpen
  \bibfield  {author} {\bibinfo {author} {\bibfnamefont {Z.}~\bibnamefont {Davoudi}}, \bibinfo {author} {\bibfnamefont {M.}~\bibnamefont {Hafezi}}, \bibinfo {author} {\bibfnamefont {C.}~\bibnamefont {Monroe}}, \bibinfo {author} {\bibfnamefont {G.}~\bibnamefont {Pagano}}, \bibinfo {author} {\bibfnamefont {A.}~\bibnamefont {Seif}}, \ and\ \bibinfo {author} {\bibfnamefont {A.}~\bibnamefont {Shaw}},\ }\href {\doibase 10.1103/PhysRevResearch.2.023015} {\bibfield  {journal} {\bibinfo  {journal} {Phys. Rev. Res.}\ }\textbf {\bibinfo {volume} {2}},\ \bibinfo {pages} {023015} (\bibinfo {year} {2020})},\ \Eprint {http://arxiv.org/abs/1908.03210} {arXiv:1908.03210 [quant-ph]} \BibitemShut {NoStop}%
\bibitem [{\citenamefont {Surace}\ \emph {et~al.}(2020)\citenamefont {Surace}, \citenamefont {Mazza}, \citenamefont {Giudici}, \citenamefont {Lerose}, \citenamefont {Gambassi},\ and\ \citenamefont {Dalmonte}}]{Surace:2019dtp}%
  \BibitemOpen
  \bibfield  {author} {\bibinfo {author} {\bibfnamefont {F.~M.}\ \bibnamefont {Surace}}, \bibinfo {author} {\bibfnamefont {P.~P.}\ \bibnamefont {Mazza}}, \bibinfo {author} {\bibfnamefont {G.}~\bibnamefont {Giudici}}, \bibinfo {author} {\bibfnamefont {A.}~\bibnamefont {Lerose}}, \bibinfo {author} {\bibfnamefont {A.}~\bibnamefont {Gambassi}}, \ and\ \bibinfo {author} {\bibfnamefont {M.}~\bibnamefont {Dalmonte}},\ }\href {\doibase 10.1103/PhysRevX.10.021041} {\bibfield  {journal} {\bibinfo  {journal} {Phys. Rev. X}\ }\textbf {\bibinfo {volume} {10}},\ \bibinfo {pages} {021041} (\bibinfo {year} {2020})},\ \Eprint {http://arxiv.org/abs/1902.09551} {arXiv:1902.09551 [cond-mat.quant-gas]} \BibitemShut {NoStop}%
\bibitem [{\citenamefont {Haase}\ \emph {et~al.}(2021)\citenamefont {Haase}, \citenamefont {Dellantonio}, \citenamefont {Celi}, \citenamefont {Paulson}, \citenamefont {Kan}, \citenamefont {Jansen},\ and\ \citenamefont {Muschik}}]{Haase:2020kaj}%
  \BibitemOpen
  \bibfield  {author} {\bibinfo {author} {\bibfnamefont {J.~F.}\ \bibnamefont {Haase}}, \bibinfo {author} {\bibfnamefont {L.}~\bibnamefont {Dellantonio}}, \bibinfo {author} {\bibfnamefont {A.}~\bibnamefont {Celi}}, \bibinfo {author} {\bibfnamefont {D.}~\bibnamefont {Paulson}}, \bibinfo {author} {\bibfnamefont {A.}~\bibnamefont {Kan}}, \bibinfo {author} {\bibfnamefont {K.}~\bibnamefont {Jansen}}, \ and\ \bibinfo {author} {\bibfnamefont {C.~A.}\ \bibnamefont {Muschik}},\ }\href {\doibase 10.22331/q-2021-02-04-393} {\bibfield  {journal} {\bibinfo  {journal} {Quantum}\ }\textbf {\bibinfo {volume} {5}},\ \bibinfo {pages} {393} (\bibinfo {year} {2021})},\ \Eprint {http://arxiv.org/abs/2006.14160} {arXiv:2006.14160 [quant-ph]} \BibitemShut {NoStop}%
\bibitem [{\citenamefont {Luo}\ \emph {et~al.}(2020)\citenamefont {Luo}, \citenamefont {Shen}, \citenamefont {Highman}, \citenamefont {Clark}, \citenamefont {DeMarco}, \citenamefont {El-Khadra},\ and\ \citenamefont {Gadway}}]{Luo:2019vmi}%
  \BibitemOpen
  \bibfield  {author} {\bibinfo {author} {\bibfnamefont {D.}~\bibnamefont {Luo}}, \bibinfo {author} {\bibfnamefont {J.}~\bibnamefont {Shen}}, \bibinfo {author} {\bibfnamefont {M.}~\bibnamefont {Highman}}, \bibinfo {author} {\bibfnamefont {B.~K.}\ \bibnamefont {Clark}}, \bibinfo {author} {\bibfnamefont {B.}~\bibnamefont {DeMarco}}, \bibinfo {author} {\bibfnamefont {A.~X.}\ \bibnamefont {El-Khadra}}, \ and\ \bibinfo {author} {\bibfnamefont {B.}~\bibnamefont {Gadway}},\ }\href {\doibase 10.1103/PhysRevA.102.032617} {\bibfield  {journal} {\bibinfo  {journal} {Phys. Rev. A}\ }\textbf {\bibinfo {volume} {102}},\ \bibinfo {pages} {032617} (\bibinfo {year} {2020})},\ \Eprint {http://arxiv.org/abs/1912.11488} {arXiv:1912.11488 [quant-ph]} \BibitemShut {NoStop}%
\bibitem [{\citenamefont {Shaw}\ \emph {et~al.}(2020)\citenamefont {Shaw}, \citenamefont {Lougovski}, \citenamefont {Stryker},\ and\ \citenamefont {Wiebe}}]{Shaw:2020udc}%
  \BibitemOpen
  \bibfield  {author} {\bibinfo {author} {\bibfnamefont {A.~F.}\ \bibnamefont {Shaw}}, \bibinfo {author} {\bibfnamefont {P.}~\bibnamefont {Lougovski}}, \bibinfo {author} {\bibfnamefont {J.~R.}\ \bibnamefont {Stryker}}, \ and\ \bibinfo {author} {\bibfnamefont {N.}~\bibnamefont {Wiebe}},\ }\href {\doibase 10.22331/q-2020-08-10-306} {\bibfield  {journal} {\bibinfo  {journal} {Quantum}\ }\textbf {\bibinfo {volume} {4}},\ \bibinfo {pages} {306} (\bibinfo {year} {2020})},\ \Eprint {http://arxiv.org/abs/2002.11146} {arXiv:2002.11146 [quant-ph]} \BibitemShut {NoStop}%
\bibitem [{\citenamefont {Yang}\ \emph {et~al.}(2020)\citenamefont {Yang}, \citenamefont {Sun}, \citenamefont {Ott}, \citenamefont {Wang}, \citenamefont {Zache}, \citenamefont {Halimeh}, \citenamefont {Yuan}, \citenamefont {Hauke},\ and\ \citenamefont {Pan}}]{Yang:2020yer}%
  \BibitemOpen
  \bibfield  {author} {\bibinfo {author} {\bibfnamefont {B.}~\bibnamefont {Yang}}, \bibinfo {author} {\bibfnamefont {H.}~\bibnamefont {Sun}}, \bibinfo {author} {\bibfnamefont {R.}~\bibnamefont {Ott}}, \bibinfo {author} {\bibfnamefont {H.-Y.}\ \bibnamefont {Wang}}, \bibinfo {author} {\bibfnamefont {T.~V.}\ \bibnamefont {Zache}}, \bibinfo {author} {\bibfnamefont {J.~C.}\ \bibnamefont {Halimeh}}, \bibinfo {author} {\bibfnamefont {Z.-S.}\ \bibnamefont {Yuan}}, \bibinfo {author} {\bibfnamefont {P.}~\bibnamefont {Hauke}}, \ and\ \bibinfo {author} {\bibfnamefont {J.-W.}\ \bibnamefont {Pan}},\ }\href {\doibase 10.1038/s41586-020-2910-8} {\bibfield  {journal} {\bibinfo  {journal} {Nature}\ }\textbf {\bibinfo {volume} {587}},\ \bibinfo {pages} {392} (\bibinfo {year} {2020})},\ \Eprint {http://arxiv.org/abs/2003.08945} {arXiv:2003.08945 [cond-mat.quant-gas]} \BibitemShut {NoStop}%
\bibitem [{\citenamefont {Ott}\ \emph {et~al.}(2021)\citenamefont {Ott}, \citenamefont {Zache}, \citenamefont {Jendrzejewski},\ and\ \citenamefont {Berges}}]{Ott:2020ycj}%
  \BibitemOpen
  \bibfield  {author} {\bibinfo {author} {\bibfnamefont {R.}~\bibnamefont {Ott}}, \bibinfo {author} {\bibfnamefont {T.~V.}\ \bibnamefont {Zache}}, \bibinfo {author} {\bibfnamefont {F.}~\bibnamefont {Jendrzejewski}}, \ and\ \bibinfo {author} {\bibfnamefont {J.}~\bibnamefont {Berges}},\ }\href {\doibase 10.1103/PhysRevLett.127.130504} {\bibfield  {journal} {\bibinfo  {journal} {Phys. Rev. Lett.}\ }\textbf {\bibinfo {volume} {127}},\ \bibinfo {pages} {130504} (\bibinfo {year} {2021})},\ \Eprint {http://arxiv.org/abs/2012.10432} {arXiv:2012.10432 [cond-mat.quant-gas]} \BibitemShut {NoStop}%
\bibitem [{\citenamefont {Paulson}\ \emph {et~al.}(2021)\citenamefont {Paulson} \emph {et~al.}}]{Paulson:2020zjd}%
  \BibitemOpen
  \bibfield  {author} {\bibinfo {author} {\bibfnamefont {D.}~\bibnamefont {Paulson}} \emph {et~al.},\ }\href {\doibase 10.1103/PRXQuantum.2.030334} {\bibfield  {journal} {\bibinfo  {journal} {PRX Quantum}\ }\textbf {\bibinfo {volume} {2}},\ \bibinfo {pages} {030334} (\bibinfo {year} {2021})},\ \Eprint {http://arxiv.org/abs/2008.09252} {arXiv:2008.09252 [quant-ph]} \BibitemShut {NoStop}%
\bibitem [{\citenamefont {Nguyen}\ \emph {et~al.}(2022)\citenamefont {Nguyen}, \citenamefont {Tran}, \citenamefont {Zhu}, \citenamefont {Green}, \citenamefont {Alderete}, \citenamefont {Davoudi},\ and\ \citenamefont {Linke}}]{Nguyen:2021hyk}%
  \BibitemOpen
  \bibfield  {author} {\bibinfo {author} {\bibfnamefont {N.~H.}\ \bibnamefont {Nguyen}}, \bibinfo {author} {\bibfnamefont {M.~C.}\ \bibnamefont {Tran}}, \bibinfo {author} {\bibfnamefont {Y.}~\bibnamefont {Zhu}}, \bibinfo {author} {\bibfnamefont {A.~M.}\ \bibnamefont {Green}}, \bibinfo {author} {\bibfnamefont {C.~H.}\ \bibnamefont {Alderete}}, \bibinfo {author} {\bibfnamefont {Z.}~\bibnamefont {Davoudi}}, \ and\ \bibinfo {author} {\bibfnamefont {N.~M.}\ \bibnamefont {Linke}},\ }\href {\doibase 10.1103/PRXQuantum.3.020324} {\bibfield  {journal} {\bibinfo  {journal} {PRX Quantum}\ }\textbf {\bibinfo {volume} {3}},\ \bibinfo {pages} {020324} (\bibinfo {year} {2022})},\ \Eprint {http://arxiv.org/abs/2112.14262} {arXiv:2112.14262 [quant-ph]} \BibitemShut {NoStop}%
\bibitem [{\citenamefont {Zhou}\ \emph {et~al.}(2022)\citenamefont {Zhou}, \citenamefont {Su}, \citenamefont {Halimeh}, \citenamefont {Ott}, \citenamefont {Sun}, \citenamefont {Hauke}, \citenamefont {Yang}, \citenamefont {Yuan}, \citenamefont {Berges},\ and\ \citenamefont {Pan}}]{Zhou:2021kdl}%
  \BibitemOpen
  \bibfield  {author} {\bibinfo {author} {\bibfnamefont {Z.-Y.}\ \bibnamefont {Zhou}}, \bibinfo {author} {\bibfnamefont {G.-X.}\ \bibnamefont {Su}}, \bibinfo {author} {\bibfnamefont {J.~C.}\ \bibnamefont {Halimeh}}, \bibinfo {author} {\bibfnamefont {R.}~\bibnamefont {Ott}}, \bibinfo {author} {\bibfnamefont {H.}~\bibnamefont {Sun}}, \bibinfo {author} {\bibfnamefont {P.}~\bibnamefont {Hauke}}, \bibinfo {author} {\bibfnamefont {B.}~\bibnamefont {Yang}}, \bibinfo {author} {\bibfnamefont {Z.-S.}\ \bibnamefont {Yuan}}, \bibinfo {author} {\bibfnamefont {J.}~\bibnamefont {Berges}}, \ and\ \bibinfo {author} {\bibfnamefont {J.-W.}\ \bibnamefont {Pan}},\ }\href {\doibase 10.1126/science.abl6277} {\bibfield  {journal} {\bibinfo  {journal} {Science}\ }\textbf {\bibinfo {volume} {377}},\ \bibinfo {pages} {311} (\bibinfo {year} {2022})},\ \Eprint {http://arxiv.org/abs/2107.13563} {arXiv:2107.13563 [cond-mat.quant-gas]} \BibitemShut {NoStop}%
\bibitem [{\citenamefont {Riechert}\ \emph {et~al.}(2022)\citenamefont {Riechert}, \citenamefont {Halimeh}, \citenamefont {Kasper}, \citenamefont {Bretheau}, \citenamefont {Zohar}, \citenamefont {Hauke},\ and\ \citenamefont {Jendrzejewski}}]{Riechert:2021ink}%
  \BibitemOpen
  \bibfield  {author} {\bibinfo {author} {\bibfnamefont {H.}~\bibnamefont {Riechert}}, \bibinfo {author} {\bibfnamefont {J.~C.}\ \bibnamefont {Halimeh}}, \bibinfo {author} {\bibfnamefont {V.}~\bibnamefont {Kasper}}, \bibinfo {author} {\bibfnamefont {L.}~\bibnamefont {Bretheau}}, \bibinfo {author} {\bibfnamefont {E.}~\bibnamefont {Zohar}}, \bibinfo {author} {\bibfnamefont {P.}~\bibnamefont {Hauke}}, \ and\ \bibinfo {author} {\bibfnamefont {F.}~\bibnamefont {Jendrzejewski}},\ }\href {\doibase 10.1103/PhysRevB.105.205141} {\bibfield  {journal} {\bibinfo  {journal} {Phys. Rev. B}\ }\textbf {\bibinfo {volume} {105}},\ \bibinfo {pages} {205141} (\bibinfo {year} {2022})},\ \Eprint {http://arxiv.org/abs/2108.01086} {arXiv:2108.01086 [cond-mat.mes-hall]} \BibitemShut {NoStop}%
\bibitem [{\citenamefont {Bauer}\ and\ \citenamefont {Grabowska}(2023)}]{Bauer:2021gek}%
  \BibitemOpen
  \bibfield  {author} {\bibinfo {author} {\bibfnamefont {C.~W.}\ \bibnamefont {Bauer}}\ and\ \bibinfo {author} {\bibfnamefont {D.~M.}\ \bibnamefont {Grabowska}},\ }\href {\doibase 10.1103/PhysRevD.107.L031503} {\bibfield  {journal} {\bibinfo  {journal} {Phys. Rev. D}\ }\textbf {\bibinfo {volume} {107}},\ \bibinfo {pages} {L031503} (\bibinfo {year} {2023})},\ \Eprint {http://arxiv.org/abs/2111.08015} {arXiv:2111.08015 [hep-ph]} \BibitemShut {NoStop}%
\bibitem [{\citenamefont {Kane}\ \emph {et~al.}(2022)\citenamefont {Kane}, \citenamefont {Grabowska}, \citenamefont {Nachman},\ and\ \citenamefont {Bauer}}]{Kane:2022ejm}%
  \BibitemOpen
  \bibfield  {author} {\bibinfo {author} {\bibfnamefont {C.}~\bibnamefont {Kane}}, \bibinfo {author} {\bibfnamefont {D.~M.}\ \bibnamefont {Grabowska}}, \bibinfo {author} {\bibfnamefont {B.}~\bibnamefont {Nachman}}, \ and\ \bibinfo {author} {\bibfnamefont {C.~W.}\ \bibnamefont {Bauer}},\ }\href@noop {} {\  (\bibinfo {year} {2022})},\ \Eprint {http://arxiv.org/abs/2211.10497} {arXiv:2211.10497 [quant-ph]} \BibitemShut {NoStop}%
\bibitem [{\citenamefont {Grabowska}\ \emph {et~al.}(2022)\citenamefont {Grabowska}, \citenamefont {Kane}, \citenamefont {Nachman},\ and\ \citenamefont {Bauer}}]{Grabowska:2022uos}%
  \BibitemOpen
  \bibfield  {author} {\bibinfo {author} {\bibfnamefont {D.~M.}\ \bibnamefont {Grabowska}}, \bibinfo {author} {\bibfnamefont {C.}~\bibnamefont {Kane}}, \bibinfo {author} {\bibfnamefont {B.}~\bibnamefont {Nachman}}, \ and\ \bibinfo {author} {\bibfnamefont {C.~W.}\ \bibnamefont {Bauer}},\ }\href@noop {} {\  (\bibinfo {year} {2022})},\ \Eprint {http://arxiv.org/abs/2208.03333} {arXiv:2208.03333 [quant-ph]} \BibitemShut {NoStop}%
\bibitem [{\citenamefont {Zhang}\ \emph {et~al.}(2023)\citenamefont {Zhang}, \citenamefont {Liu}, \citenamefont {Cheng}, \citenamefont {He}, \citenamefont {Wang}, \citenamefont {Wang}, \citenamefont {Zhu}, \citenamefont {Su}, \citenamefont {Zhou}, \citenamefont {Zheng}, \citenamefont {Sun}, \citenamefont {Yang}, \citenamefont {Hauke}, \citenamefont {Zheng}, \citenamefont {Halimeh}, \citenamefont {Yuan},\ and\ \citenamefont {Pan}}]{zhang2023observation}%
  \BibitemOpen
  \bibfield  {author} {\bibinfo {author} {\bibfnamefont {W.-Y.}\ \bibnamefont {Zhang}}, \bibinfo {author} {\bibfnamefont {Y.}~\bibnamefont {Liu}}, \bibinfo {author} {\bibfnamefont {Y.}~\bibnamefont {Cheng}}, \bibinfo {author} {\bibfnamefont {M.-G.}\ \bibnamefont {He}}, \bibinfo {author} {\bibfnamefont {H.-Y.}\ \bibnamefont {Wang}}, \bibinfo {author} {\bibfnamefont {T.-Y.}\ \bibnamefont {Wang}}, \bibinfo {author} {\bibfnamefont {Z.-H.}\ \bibnamefont {Zhu}}, \bibinfo {author} {\bibfnamefont {G.-X.}\ \bibnamefont {Su}}, \bibinfo {author} {\bibfnamefont {Z.-Y.}\ \bibnamefont {Zhou}}, \bibinfo {author} {\bibfnamefont {Y.-G.}\ \bibnamefont {Zheng}}, \bibinfo {author} {\bibfnamefont {H.}~\bibnamefont {Sun}}, \bibinfo {author} {\bibfnamefont {B.}~\bibnamefont {Yang}}, \bibinfo {author} {\bibfnamefont {P.}~\bibnamefont {Hauke}}, \bibinfo {author} {\bibfnamefont {W.}~\bibnamefont {Zheng}}, \bibinfo {author} {\bibfnamefont {J.~C.}\ \bibnamefont {Halimeh}}, \bibinfo {author} {\bibfnamefont {Z.-S.}\ \bibnamefont {Yuan}},
  \ and\ \bibinfo {author} {\bibfnamefont {J.-W.}\ \bibnamefont {Pan}},\ }\href@noop {} {\enquote {\bibinfo {title} {Observation of microscopic confinement dynamics by a tunable topological $\theta$-angle},}\ } (\bibinfo {year} {2023}),\ \Eprint {http://arxiv.org/abs/2306.11794} {arXiv:2306.11794 [cond-mat.quant-gas]} \BibitemShut {NoStop}%
\bibitem [{\citenamefont {Farrell}\ \emph {et~al.}(2024{\natexlab{a}})\citenamefont {Farrell}, \citenamefont {Illa}, \citenamefont {Ciavarella},\ and\ \citenamefont {Savage}}]{Farrell:2023fgd}%
  \BibitemOpen
  \bibfield  {author} {\bibinfo {author} {\bibfnamefont {R.~C.}\ \bibnamefont {Farrell}}, \bibinfo {author} {\bibfnamefont {M.}~\bibnamefont {Illa}}, \bibinfo {author} {\bibfnamefont {A.~N.}\ \bibnamefont {Ciavarella}}, \ and\ \bibinfo {author} {\bibfnamefont {M.~J.}\ \bibnamefont {Savage}},\ }\href {\doibase 10.1103/PRXQuantum.5.020315} {\bibfield  {journal} {\bibinfo  {journal} {PRX Quantum}\ }\textbf {\bibinfo {volume} {5}},\ \bibinfo {pages} {020315} (\bibinfo {year} {2024}{\natexlab{a}})},\ \Eprint {http://arxiv.org/abs/2308.04481} {arXiv:2308.04481 [quant-ph]} \BibitemShut {NoStop}%
\bibitem [{\citenamefont {Nagano}\ \emph {et~al.}(2023)\citenamefont {Nagano}, \citenamefont {Bapat},\ and\ \citenamefont {Bauer}}]{Nagano:2023uaq}%
  \BibitemOpen
  \bibfield  {author} {\bibinfo {author} {\bibfnamefont {L.}~\bibnamefont {Nagano}}, \bibinfo {author} {\bibfnamefont {A.}~\bibnamefont {Bapat}}, \ and\ \bibinfo {author} {\bibfnamefont {C.~W.}\ \bibnamefont {Bauer}},\ }\href {\doibase 10.1103/PhysRevD.108.034501} {\bibfield  {journal} {\bibinfo  {journal} {Phys. Rev. D}\ }\textbf {\bibinfo {volume} {108}},\ \bibinfo {pages} {034501} (\bibinfo {year} {2023})},\ \Eprint {http://arxiv.org/abs/2302.10933} {arXiv:2302.10933 [hep-ph]} \BibitemShut {NoStop}%
\bibitem [{\citenamefont {Gustafson}\ \emph {et~al.}(2025)\citenamefont {Gustafson} \emph {et~al.}}]{Gustafson:2024bww}%
  \BibitemOpen
  \bibfield  {author} {\bibinfo {author} {\bibfnamefont {E.}~\bibnamefont {Gustafson}} \emph {et~al.},\ }\href {\doibase 10.1103/PhysRevApplied.23.064002} {\bibfield  {journal} {\bibinfo  {journal} {Phys. Rev. Applied}\ }\textbf {\bibinfo {volume} {23}},\ \bibinfo {pages} {064002} (\bibinfo {year} {2025})},\ \Eprint {http://arxiv.org/abs/2408.12641} {arXiv:2408.12641 [quant-ph]} \BibitemShut {NoStop}%
\bibitem [{\citenamefont {Crane}\ \emph {et~al.}(2024)\citenamefont {Crane} \emph {et~al.}}]{Crane:2024tlj}%
  \BibitemOpen
  \bibfield  {author} {\bibinfo {author} {\bibfnamefont {E.}~\bibnamefont {Crane}} \emph {et~al.},\ }\href@noop {} {\  (\bibinfo {year} {2024})},\ \Eprint {http://arxiv.org/abs/2409.03747} {arXiv:2409.03747 [quant-ph]} \BibitemShut {NoStop}%
\bibitem [{\citenamefont {Zohar}\ \emph {et~al.}(2013{\natexlab{b}})\citenamefont {Zohar}, \citenamefont {Cirac},\ and\ \citenamefont {Reznik}}]{Zohar:2012xf}%
  \BibitemOpen
  \bibfield  {author} {\bibinfo {author} {\bibfnamefont {E.}~\bibnamefont {Zohar}}, \bibinfo {author} {\bibfnamefont {J.~I.}\ \bibnamefont {Cirac}}, \ and\ \bibinfo {author} {\bibfnamefont {B.}~\bibnamefont {Reznik}},\ }\href {\doibase 10.1103/PhysRevLett.110.125304} {\bibfield  {journal} {\bibinfo  {journal} {Phys. Rev. Lett.}\ }\textbf {\bibinfo {volume} {110}},\ \bibinfo {pages} {125304} (\bibinfo {year} {2013}{\natexlab{b}})},\ \Eprint {http://arxiv.org/abs/1211.2241} {arXiv:1211.2241 [quant-ph]} \BibitemShut {NoStop}%
\bibitem [{\citenamefont {Stannigel}\ \emph {et~al.}(2014)\citenamefont {Stannigel}, \citenamefont {Hauke}, \citenamefont {Marcos}, \citenamefont {Hafezi}, \citenamefont {Diehl}, \citenamefont {Dalmonte},\ and\ \citenamefont {Zoller}}]{Stannigel:2013zka}%
  \BibitemOpen
  \bibfield  {author} {\bibinfo {author} {\bibfnamefont {K.}~\bibnamefont {Stannigel}}, \bibinfo {author} {\bibfnamefont {P.}~\bibnamefont {Hauke}}, \bibinfo {author} {\bibfnamefont {D.}~\bibnamefont {Marcos}}, \bibinfo {author} {\bibfnamefont {M.}~\bibnamefont {Hafezi}}, \bibinfo {author} {\bibfnamefont {S.}~\bibnamefont {Diehl}}, \bibinfo {author} {\bibfnamefont {M.}~\bibnamefont {Dalmonte}}, \ and\ \bibinfo {author} {\bibfnamefont {P.}~\bibnamefont {Zoller}},\ }\href {\doibase 10.1103/PhysRevLett.112.120406} {\bibfield  {journal} {\bibinfo  {journal} {Phys. Rev. Lett.}\ }\textbf {\bibinfo {volume} {112}},\ \bibinfo {pages} {120406} (\bibinfo {year} {2014})},\ \Eprint {http://arxiv.org/abs/1308.0528} {arXiv:1308.0528 [quant-ph]} \BibitemShut {NoStop}%
\bibitem [{\citenamefont {Mezzacapo}\ \emph {et~al.}(2015)\citenamefont {Mezzacapo}, \citenamefont {Rico}, \citenamefont {Sabin}, \citenamefont {Egusquiza}, \citenamefont {Lamata},\ and\ \citenamefont {Solano}}]{Mezzacapo:2015bra}%
  \BibitemOpen
  \bibfield  {author} {\bibinfo {author} {\bibfnamefont {A.}~\bibnamefont {Mezzacapo}}, \bibinfo {author} {\bibfnamefont {E.}~\bibnamefont {Rico}}, \bibinfo {author} {\bibfnamefont {C.}~\bibnamefont {Sabin}}, \bibinfo {author} {\bibfnamefont {I.~L.}\ \bibnamefont {Egusquiza}}, \bibinfo {author} {\bibfnamefont {L.}~\bibnamefont {Lamata}}, \ and\ \bibinfo {author} {\bibfnamefont {E.}~\bibnamefont {Solano}},\ }\href {\doibase 10.1103/PhysRevLett.115.240502} {\bibfield  {journal} {\bibinfo  {journal} {Phys. Rev. Lett.}\ }\textbf {\bibinfo {volume} {115}},\ \bibinfo {pages} {240502} (\bibinfo {year} {2015})},\ \Eprint {http://arxiv.org/abs/1505.04720} {arXiv:1505.04720 [quant-ph]} \BibitemShut {NoStop}%
\bibitem [{\citenamefont {Mathur}\ and\ \citenamefont {Sreeraj}(2015)}]{Mathur:2015wba}%
  \BibitemOpen
  \bibfield  {author} {\bibinfo {author} {\bibfnamefont {M.}~\bibnamefont {Mathur}}\ and\ \bibinfo {author} {\bibfnamefont {T.~P.}\ \bibnamefont {Sreeraj}},\ }\href {\doibase 10.1103/PhysRevD.92.125018} {\bibfield  {journal} {\bibinfo  {journal} {Phys. Rev. D}\ }\textbf {\bibinfo {volume} {92}},\ \bibinfo {pages} {125018} (\bibinfo {year} {2015})},\ \Eprint {http://arxiv.org/abs/1509.04033} {arXiv:1509.04033 [hep-lat]} \BibitemShut {NoStop}%
\bibitem [{\citenamefont {Raychowdhury}\ and\ \citenamefont {Stryker}(2020{\natexlab{a}})}]{Raychowdhury:2018osk}%
  \BibitemOpen
  \bibfield  {author} {\bibinfo {author} {\bibfnamefont {I.}~\bibnamefont {Raychowdhury}}\ and\ \bibinfo {author} {\bibfnamefont {J.~R.}\ \bibnamefont {Stryker}},\ }\href {\doibase 10.1103/PhysRevResearch.2.033039} {\bibfield  {journal} {\bibinfo  {journal} {Phys. Rev. Res.}\ }\textbf {\bibinfo {volume} {2}},\ \bibinfo {pages} {033039} (\bibinfo {year} {2020}{\natexlab{a}})},\ \Eprint {http://arxiv.org/abs/1812.07554} {arXiv:1812.07554 [hep-lat]} \BibitemShut {NoStop}%
\bibitem [{\citenamefont {Raychowdhury}\ and\ \citenamefont {Stryker}(2020{\natexlab{b}})}]{Raychowdhury:2019iki}%
  \BibitemOpen
  \bibfield  {author} {\bibinfo {author} {\bibfnamefont {I.}~\bibnamefont {Raychowdhury}}\ and\ \bibinfo {author} {\bibfnamefont {J.~R.}\ \bibnamefont {Stryker}},\ }\href {\doibase 10.1103/PhysRevD.101.114502} {\bibfield  {journal} {\bibinfo  {journal} {Phys. Rev. D}\ }\textbf {\bibinfo {volume} {101}},\ \bibinfo {pages} {114502} (\bibinfo {year} {2020}{\natexlab{b}})},\ \Eprint {http://arxiv.org/abs/1912.06133} {arXiv:1912.06133 [hep-lat]} \BibitemShut {NoStop}%
\bibitem [{\citenamefont {Klco}\ \emph {et~al.}(2020)\citenamefont {Klco}, \citenamefont {Stryker},\ and\ \citenamefont {Savage}}]{Klco:2019evd}%
  \BibitemOpen
  \bibfield  {author} {\bibinfo {author} {\bibfnamefont {N.}~\bibnamefont {Klco}}, \bibinfo {author} {\bibfnamefont {J.~R.}\ \bibnamefont {Stryker}}, \ and\ \bibinfo {author} {\bibfnamefont {M.~J.}\ \bibnamefont {Savage}},\ }\href {\doibase 10.1103/PhysRevD.101.074512} {\bibfield  {journal} {\bibinfo  {journal} {Phys. Rev. D}\ }\textbf {\bibinfo {volume} {101}},\ \bibinfo {pages} {074512} (\bibinfo {year} {2020})},\ \Eprint {http://arxiv.org/abs/1908.06935} {arXiv:1908.06935 [quant-ph]} \BibitemShut {NoStop}%
\bibitem [{\citenamefont {Dasgupta}\ and\ \citenamefont {Raychowdhury}(2022)}]{Dasgupta:2020itb}%
  \BibitemOpen
  \bibfield  {author} {\bibinfo {author} {\bibfnamefont {R.}~\bibnamefont {Dasgupta}}\ and\ \bibinfo {author} {\bibfnamefont {I.}~\bibnamefont {Raychowdhury}},\ }\href {\doibase 10.1103/PhysRevA.105.023322} {\bibfield  {journal} {\bibinfo  {journal} {Phys. Rev. A}\ }\textbf {\bibinfo {volume} {105}},\ \bibinfo {pages} {023322} (\bibinfo {year} {2022})},\ \Eprint {http://arxiv.org/abs/2009.13969} {arXiv:2009.13969 [hep-lat]} \BibitemShut {NoStop}%
\bibitem [{\citenamefont {Davoudi}\ \emph {et~al.}(2021)\citenamefont {Davoudi}, \citenamefont {Raychowdhury},\ and\ \citenamefont {Shaw}}]{Davoudi:2020yln}%
  \BibitemOpen
  \bibfield  {author} {\bibinfo {author} {\bibfnamefont {Z.}~\bibnamefont {Davoudi}}, \bibinfo {author} {\bibfnamefont {I.}~\bibnamefont {Raychowdhury}}, \ and\ \bibinfo {author} {\bibfnamefont {A.}~\bibnamefont {Shaw}},\ }\href {\doibase 10.1103/PhysRevD.104.074505} {\bibfield  {journal} {\bibinfo  {journal} {Phys. Rev. D}\ }\textbf {\bibinfo {volume} {104}},\ \bibinfo {pages} {074505} (\bibinfo {year} {2021})},\ \Eprint {http://arxiv.org/abs/2009.11802} {arXiv:2009.11802 [hep-lat]} \BibitemShut {NoStop}%
\bibitem [{\citenamefont {Atas}\ \emph {et~al.}(2021)\citenamefont {Atas}, \citenamefont {Zhang}, \citenamefont {Lewis}, \citenamefont {Jahanpour}, \citenamefont {Haase},\ and\ \citenamefont {Muschik}}]{Atas:2021ext}%
  \BibitemOpen
  \bibfield  {author} {\bibinfo {author} {\bibfnamefont {Y.~Y.}\ \bibnamefont {Atas}}, \bibinfo {author} {\bibfnamefont {J.}~\bibnamefont {Zhang}}, \bibinfo {author} {\bibfnamefont {R.}~\bibnamefont {Lewis}}, \bibinfo {author} {\bibfnamefont {A.}~\bibnamefont {Jahanpour}}, \bibinfo {author} {\bibfnamefont {J.~F.}\ \bibnamefont {Haase}}, \ and\ \bibinfo {author} {\bibfnamefont {C.~A.}\ \bibnamefont {Muschik}},\ }\href {\doibase 10.1038/s41467-021-26825-4} {\bibfield  {journal} {\bibinfo  {journal} {Nature Commun.}\ }\textbf {\bibinfo {volume} {12}},\ \bibinfo {pages} {6499} (\bibinfo {year} {2021})},\ \Eprint {http://arxiv.org/abs/2102.08920} {arXiv:2102.08920 [quant-ph]} \BibitemShut {NoStop}%
\bibitem [{\citenamefont {A~Rahman}\ \emph {et~al.}(2021)\citenamefont {A~Rahman}, \citenamefont {Lewis}, \citenamefont {Mendicelli},\ and\ \citenamefont {Powell}}]{ARahman:2021ktn}%
  \BibitemOpen
  \bibfield  {author} {\bibinfo {author} {\bibfnamefont {S.}~\bibnamefont {A~Rahman}}, \bibinfo {author} {\bibfnamefont {R.}~\bibnamefont {Lewis}}, \bibinfo {author} {\bibfnamefont {E.}~\bibnamefont {Mendicelli}}, \ and\ \bibinfo {author} {\bibfnamefont {S.}~\bibnamefont {Powell}},\ }\href {\doibase 10.1103/PhysRevD.104.034501} {\bibfield  {journal} {\bibinfo  {journal} {Phys. Rev. D}\ }\textbf {\bibinfo {volume} {104}},\ \bibinfo {pages} {034501} (\bibinfo {year} {2021})},\ \Eprint {http://arxiv.org/abs/2103.08661} {arXiv:2103.08661 [hep-lat]} \BibitemShut {NoStop}%
\bibitem [{\citenamefont {Osborne}\ \emph {et~al.}(2022)\citenamefont {Osborne}, \citenamefont {McCulloch}, \citenamefont {Yang}, \citenamefont {Hauke},\ and\ \citenamefont {Halimeh}}]{Osborne:2022jxq}%
  \BibitemOpen
  \bibfield  {author} {\bibinfo {author} {\bibfnamefont {J.}~\bibnamefont {Osborne}}, \bibinfo {author} {\bibfnamefont {I.~P.}\ \bibnamefont {McCulloch}}, \bibinfo {author} {\bibfnamefont {B.}~\bibnamefont {Yang}}, \bibinfo {author} {\bibfnamefont {P.}~\bibnamefont {Hauke}}, \ and\ \bibinfo {author} {\bibfnamefont {J.~C.}\ \bibnamefont {Halimeh}},\ }\href@noop {} {\  (\bibinfo {year} {2022})},\ \Eprint {http://arxiv.org/abs/2211.01380} {arXiv:2211.01380 [cond-mat.quant-gas]} \BibitemShut {NoStop}%
\bibitem [{\citenamefont {Halimeh}\ \emph {et~al.}(2022)\citenamefont {Halimeh}, \citenamefont {Lang},\ and\ \citenamefont {Hauke}}]{halimeh2022gauge}%
  \BibitemOpen
  \bibfield  {author} {\bibinfo {author} {\bibfnamefont {J.~C.}\ \bibnamefont {Halimeh}}, \bibinfo {author} {\bibfnamefont {H.}~\bibnamefont {Lang}}, \ and\ \bibinfo {author} {\bibfnamefont {P.}~\bibnamefont {Hauke}},\ }\href {\doibase 10.1088/1367-2630/ac5564} {\bibfield  {journal} {\bibinfo  {journal} {New Journal of Physics}\ }\textbf {\bibinfo {volume} {24}},\ \bibinfo {pages} {033015} (\bibinfo {year} {2022})}\BibitemShut {NoStop}%
\bibitem [{\citenamefont {A~Rahman}\ \emph {et~al.}(2022)\citenamefont {A~Rahman}, \citenamefont {Lewis}, \citenamefont {Mendicelli},\ and\ \citenamefont {Powell}}]{ARahman:2022tkr}%
  \BibitemOpen
  \bibfield  {author} {\bibinfo {author} {\bibfnamefont {S.}~\bibnamefont {A~Rahman}}, \bibinfo {author} {\bibfnamefont {R.}~\bibnamefont {Lewis}}, \bibinfo {author} {\bibfnamefont {E.}~\bibnamefont {Mendicelli}}, \ and\ \bibinfo {author} {\bibfnamefont {S.}~\bibnamefont {Powell}},\ }\href {\doibase 10.1103/PhysRevD.106.074502} {\bibfield  {journal} {\bibinfo  {journal} {Phys. Rev. D}\ }\textbf {\bibinfo {volume} {106}},\ \bibinfo {pages} {074502} (\bibinfo {year} {2022})},\ \Eprint {http://arxiv.org/abs/2205.09247} {arXiv:2205.09247 [hep-lat]} \BibitemShut {NoStop}%
\bibitem [{\citenamefont {Zache}\ \emph {et~al.}(2023)\citenamefont {Zache}, \citenamefont {Gonz\'alez-Cuadra},\ and\ \citenamefont {Zoller}}]{zache2023quantum}%
  \BibitemOpen
  \bibfield  {author} {\bibinfo {author} {\bibfnamefont {T.~V.}\ \bibnamefont {Zache}}, \bibinfo {author} {\bibfnamefont {D.}~\bibnamefont {Gonz\'alez-Cuadra}}, \ and\ \bibinfo {author} {\bibfnamefont {P.}~\bibnamefont {Zoller}},\ }\href {\doibase 10.1103/PhysRevLett.131.171902} {\bibfield  {journal} {\bibinfo  {journal} {Phys. Rev. Lett.}\ }\textbf {\bibinfo {volume} {131}},\ \bibinfo {pages} {171902} (\bibinfo {year} {2023})}\BibitemShut {NoStop}%
\bibitem [{\citenamefont {Alexandru}\ \emph {et~al.}(2024)\citenamefont {Alexandru}, \citenamefont {Bedaque}, \citenamefont {Carosso}, \citenamefont {Cervia}, \citenamefont {Murairi},\ and\ \citenamefont {Sheng}}]{Alexandru:2023qzd}%
  \BibitemOpen
  \bibfield  {author} {\bibinfo {author} {\bibfnamefont {A.}~\bibnamefont {Alexandru}}, \bibinfo {author} {\bibfnamefont {P.~F.}\ \bibnamefont {Bedaque}}, \bibinfo {author} {\bibfnamefont {A.}~\bibnamefont {Carosso}}, \bibinfo {author} {\bibfnamefont {M.~J.}\ \bibnamefont {Cervia}}, \bibinfo {author} {\bibfnamefont {E.~M.}\ \bibnamefont {Murairi}}, \ and\ \bibinfo {author} {\bibfnamefont {A.}~\bibnamefont {Sheng}},\ }\href {\doibase 10.1103/PhysRevD.109.094502} {\bibfield  {journal} {\bibinfo  {journal} {Phys. Rev. D}\ }\textbf {\bibinfo {volume} {109}},\ \bibinfo {pages} {094502} (\bibinfo {year} {2024})},\ \Eprint {http://arxiv.org/abs/2308.05253} {arXiv:2308.05253 [hep-lat]} \BibitemShut {NoStop}%
\bibitem [{\citenamefont {D'Andrea}\ \emph {et~al.}(2024)\citenamefont {D'Andrea}, \citenamefont {Bauer}, \citenamefont {Grabowska},\ and\ \citenamefont {Freytsis}}]{DAndrea:2023qnr}%
  \BibitemOpen
  \bibfield  {author} {\bibinfo {author} {\bibfnamefont {I.}~\bibnamefont {D'Andrea}}, \bibinfo {author} {\bibfnamefont {C.~W.}\ \bibnamefont {Bauer}}, \bibinfo {author} {\bibfnamefont {D.~M.}\ \bibnamefont {Grabowska}}, \ and\ \bibinfo {author} {\bibfnamefont {M.}~\bibnamefont {Freytsis}},\ }\href {\doibase 10.1103/PhysRevD.109.074501} {\bibfield  {journal} {\bibinfo  {journal} {Phys. Rev. D}\ }\textbf {\bibinfo {volume} {109}},\ \bibinfo {pages} {074501} (\bibinfo {year} {2024})},\ \Eprint {http://arxiv.org/abs/2307.11829} {arXiv:2307.11829 [hep-ph]} \BibitemShut {NoStop}%
\bibitem [{\citenamefont {Turro}\ \emph {et~al.}(2024)\citenamefont {Turro}, \citenamefont {Ciavarella},\ and\ \citenamefont {Yao}}]{Turro:2024pxu}%
  \BibitemOpen
  \bibfield  {author} {\bibinfo {author} {\bibfnamefont {F.}~\bibnamefont {Turro}}, \bibinfo {author} {\bibfnamefont {A.}~\bibnamefont {Ciavarella}}, \ and\ \bibinfo {author} {\bibfnamefont {X.}~\bibnamefont {Yao}},\ }\href {\doibase 10.1103/PhysRevD.109.114511} {\bibfield  {journal} {\bibinfo  {journal} {Phys. Rev. D}\ }\textbf {\bibinfo {volume} {109}},\ \bibinfo {pages} {114511} (\bibinfo {year} {2024})},\ \Eprint {http://arxiv.org/abs/2402.04221} {arXiv:2402.04221 [hep-lat]} \BibitemShut {NoStop}%
\bibitem [{\citenamefont {Grabowska}\ \emph {et~al.}(2024)\citenamefont {Grabowska}, \citenamefont {Kane},\ and\ \citenamefont {Bauer}}]{Grabowska:2024emw}%
  \BibitemOpen
  \bibfield  {author} {\bibinfo {author} {\bibfnamefont {D.~M.}\ \bibnamefont {Grabowska}}, \bibinfo {author} {\bibfnamefont {C.~F.}\ \bibnamefont {Kane}}, \ and\ \bibinfo {author} {\bibfnamefont {C.~W.}\ \bibnamefont {Bauer}},\ }\href@noop {} {\  (\bibinfo {year} {2024})},\ \Eprint {http://arxiv.org/abs/2409.10610} {arXiv:2409.10610 [quant-ph]} \BibitemShut {NoStop}%
\bibitem [{\citenamefont {Burbano}\ and\ \citenamefont {Bauer}(2024)}]{Burbano:2024uvn}%
  \BibitemOpen
  \bibfield  {author} {\bibinfo {author} {\bibfnamefont {I.~M.}\ \bibnamefont {Burbano}}\ and\ \bibinfo {author} {\bibfnamefont {C.~W.}\ \bibnamefont {Bauer}},\ }\href@noop {} {\  (\bibinfo {year} {2024})},\ \Eprint {http://arxiv.org/abs/2409.13812} {arXiv:2409.13812 [hep-lat]} \BibitemShut {NoStop}%
\bibitem [{\citenamefont {Anishetty}\ \emph {et~al.}(2010)\citenamefont {Anishetty}, \citenamefont {Mathur},\ and\ \citenamefont {Raychowdhury}}]{Anishetty:2009nh}%
  \BibitemOpen
  \bibfield  {author} {\bibinfo {author} {\bibfnamefont {R.}~\bibnamefont {Anishetty}}, \bibinfo {author} {\bibfnamefont {M.}~\bibnamefont {Mathur}}, \ and\ \bibinfo {author} {\bibfnamefont {I.}~\bibnamefont {Raychowdhury}},\ }\href {\doibase 10.1088/1751-8113/43/3/035403} {\bibfield  {journal} {\bibinfo  {journal} {J. Phys. A}\ }\textbf {\bibinfo {volume} {43}},\ \bibinfo {pages} {035403} (\bibinfo {year} {2010})},\ \Eprint {http://arxiv.org/abs/0909.2394} {arXiv:0909.2394 [hep-lat]} \BibitemShut {NoStop}%
\bibitem [{\citenamefont {Alexandru}\ \emph {et~al.}(2019)\citenamefont {Alexandru}, \citenamefont {Bedaque}, \citenamefont {Harmalkar}, \citenamefont {Lamm}, \citenamefont {Lawrence},\ and\ \citenamefont {Warrington}}]{Alexandru:2019nsa}%
  \BibitemOpen
  \bibfield  {author} {\bibinfo {author} {\bibfnamefont {A.}~\bibnamefont {Alexandru}}, \bibinfo {author} {\bibfnamefont {P.~F.}\ \bibnamefont {Bedaque}}, \bibinfo {author} {\bibfnamefont {S.}~\bibnamefont {Harmalkar}}, \bibinfo {author} {\bibfnamefont {H.}~\bibnamefont {Lamm}}, \bibinfo {author} {\bibfnamefont {S.}~\bibnamefont {Lawrence}}, \ and\ \bibinfo {author} {\bibfnamefont {N.~C.}\ \bibnamefont {Warrington}} (\bibinfo {collaboration} {NuQS}),\ }\href {\doibase 10.1103/PhysRevD.100.114501} {\bibfield  {journal} {\bibinfo  {journal} {Phys. Rev. D}\ }\textbf {\bibinfo {volume} {100}},\ \bibinfo {pages} {114501} (\bibinfo {year} {2019})},\ \Eprint {http://arxiv.org/abs/1906.11213} {arXiv:1906.11213 [hep-lat]} \BibitemShut {NoStop}%
\bibitem [{\citenamefont {Ciavarella}\ \emph {et~al.}(2021)\citenamefont {Ciavarella}, \citenamefont {Klco},\ and\ \citenamefont {Savage}}]{Ciavarella:2021nmj}%
  \BibitemOpen
  \bibfield  {author} {\bibinfo {author} {\bibfnamefont {A.}~\bibnamefont {Ciavarella}}, \bibinfo {author} {\bibfnamefont {N.}~\bibnamefont {Klco}}, \ and\ \bibinfo {author} {\bibfnamefont {M.~J.}\ \bibnamefont {Savage}},\ }\href {\doibase 10.1103/PhysRevD.103.094501} {\bibfield  {journal} {\bibinfo  {journal} {Phys. Rev. D}\ }\textbf {\bibinfo {volume} {103}},\ \bibinfo {pages} {094501} (\bibinfo {year} {2021})},\ \Eprint {http://arxiv.org/abs/2101.10227} {arXiv:2101.10227 [quant-ph]} \BibitemShut {NoStop}%
\bibitem [{\citenamefont {Farrell}\ \emph {et~al.}(2023{\natexlab{a}})\citenamefont {Farrell}, \citenamefont {Chernyshev}, \citenamefont {Powell}, \citenamefont {Zemlevskiy}, \citenamefont {Illa},\ and\ \citenamefont {Savage}}]{Farrell:2022wyt}%
  \BibitemOpen
  \bibfield  {author} {\bibinfo {author} {\bibfnamefont {R.~C.}\ \bibnamefont {Farrell}}, \bibinfo {author} {\bibfnamefont {I.~A.}\ \bibnamefont {Chernyshev}}, \bibinfo {author} {\bibfnamefont {S.~J.~M.}\ \bibnamefont {Powell}}, \bibinfo {author} {\bibfnamefont {N.~A.}\ \bibnamefont {Zemlevskiy}}, \bibinfo {author} {\bibfnamefont {M.}~\bibnamefont {Illa}}, \ and\ \bibinfo {author} {\bibfnamefont {M.~J.}\ \bibnamefont {Savage}},\ }\href {\doibase 10.1103/PhysRevD.107.054512} {\bibfield  {journal} {\bibinfo  {journal} {Phys. Rev. D}\ }\textbf {\bibinfo {volume} {107}},\ \bibinfo {pages} {054512} (\bibinfo {year} {2023}{\natexlab{a}})},\ \Eprint {http://arxiv.org/abs/2207.01731} {arXiv:2207.01731 [quant-ph]} \BibitemShut {NoStop}%
\bibitem [{\citenamefont {Farrell}\ \emph {et~al.}(2023{\natexlab{b}})\citenamefont {Farrell}, \citenamefont {Chernyshev}, \citenamefont {Powell}, \citenamefont {Zemlevskiy}, \citenamefont {Illa},\ and\ \citenamefont {Savage}}]{Farrell:2022vyh}%
  \BibitemOpen
  \bibfield  {author} {\bibinfo {author} {\bibfnamefont {R.~C.}\ \bibnamefont {Farrell}}, \bibinfo {author} {\bibfnamefont {I.~A.}\ \bibnamefont {Chernyshev}}, \bibinfo {author} {\bibfnamefont {S.~J.~M.}\ \bibnamefont {Powell}}, \bibinfo {author} {\bibfnamefont {N.~A.}\ \bibnamefont {Zemlevskiy}}, \bibinfo {author} {\bibfnamefont {M.}~\bibnamefont {Illa}}, \ and\ \bibinfo {author} {\bibfnamefont {M.~J.}\ \bibnamefont {Savage}},\ }\href {\doibase 10.1103/PhysRevD.107.054513} {\bibfield  {journal} {\bibinfo  {journal} {Phys. Rev. D}\ }\textbf {\bibinfo {volume} {107}},\ \bibinfo {pages} {054513} (\bibinfo {year} {2023}{\natexlab{b}})},\ \Eprint {http://arxiv.org/abs/2209.10781} {arXiv:2209.10781 [quant-ph]} \BibitemShut {NoStop}%
\bibitem [{\citenamefont {Atas}\ \emph {et~al.}(2023)\citenamefont {Atas}, \citenamefont {Haase}, \citenamefont {Zhang}, \citenamefont {Wei}, \citenamefont {Pfaendler}, \citenamefont {Lewis},\ and\ \citenamefont {Muschik}}]{Atas:2022dqm}%
  \BibitemOpen
  \bibfield  {author} {\bibinfo {author} {\bibfnamefont {Y.~Y.}\ \bibnamefont {Atas}}, \bibinfo {author} {\bibfnamefont {J.~F.}\ \bibnamefont {Haase}}, \bibinfo {author} {\bibfnamefont {J.}~\bibnamefont {Zhang}}, \bibinfo {author} {\bibfnamefont {V.}~\bibnamefont {Wei}}, \bibinfo {author} {\bibfnamefont {S.~M.~L.}\ \bibnamefont {Pfaendler}}, \bibinfo {author} {\bibfnamefont {R.}~\bibnamefont {Lewis}}, \ and\ \bibinfo {author} {\bibfnamefont {C.~A.}\ \bibnamefont {Muschik}},\ }\href {\doibase 10.1103/PhysRevResearch.5.033184} {\bibfield  {journal} {\bibinfo  {journal} {Phys. Rev. Res.}\ }\textbf {\bibinfo {volume} {5}},\ \bibinfo {pages} {033184} (\bibinfo {year} {2023})},\ \Eprint {http://arxiv.org/abs/2207.03473} {arXiv:2207.03473 [quant-ph]} \BibitemShut {NoStop}%
\bibitem [{\citenamefont {Ciavarella}\ and\ \citenamefont {Chernyshev}(2022)}]{Ciavarella:2021lel}%
  \BibitemOpen
  \bibfield  {author} {\bibinfo {author} {\bibfnamefont {A.~N.}\ \bibnamefont {Ciavarella}}\ and\ \bibinfo {author} {\bibfnamefont {I.~A.}\ \bibnamefont {Chernyshev}},\ }\href {\doibase 10.1103/PhysRevD.105.074504} {\bibfield  {journal} {\bibinfo  {journal} {Phys. Rev. D}\ }\textbf {\bibinfo {volume} {105}},\ \bibinfo {pages} {074504} (\bibinfo {year} {2022})},\ \Eprint {http://arxiv.org/abs/2112.09083} {arXiv:2112.09083 [quant-ph]} \BibitemShut {NoStop}%
\bibitem [{\citenamefont {Hayata}\ and\ \citenamefont {Hidaka}(2023)}]{hayata2023qdeformedformulationhamiltonian}%
  \BibitemOpen
  \bibfield  {author} {\bibinfo {author} {\bibfnamefont {T.}~\bibnamefont {Hayata}}\ and\ \bibinfo {author} {\bibfnamefont {Y.}~\bibnamefont {Hidaka}},\ }\href {https://arxiv.org/abs/2306.12324} {\enquote {\bibinfo {title} {$q$ deformed formulation of hamiltonian su(3) yang-mills theory},}\ } (\bibinfo {year} {2023}),\ \Eprint {http://arxiv.org/abs/2306.12324} {arXiv:2306.12324 [hep-lat]} \BibitemShut {NoStop}%
\bibitem [{\citenamefont {Farrell}\ \emph {et~al.}(2024{\natexlab{b}})\citenamefont {Farrell}, \citenamefont {Illa}, \citenamefont {Ciavarella},\ and\ \citenamefont {Savage}}]{Farrell:2024fit}%
  \BibitemOpen
  \bibfield  {author} {\bibinfo {author} {\bibfnamefont {R.~C.}\ \bibnamefont {Farrell}}, \bibinfo {author} {\bibfnamefont {M.}~\bibnamefont {Illa}}, \bibinfo {author} {\bibfnamefont {A.~N.}\ \bibnamefont {Ciavarella}}, \ and\ \bibinfo {author} {\bibfnamefont {M.~J.}\ \bibnamefont {Savage}},\ }\href {\doibase 10.1103/PhysRevD.109.114510} {\bibfield  {journal} {\bibinfo  {journal} {Phys. Rev. D}\ }\textbf {\bibinfo {volume} {109}},\ \bibinfo {pages} {114510} (\bibinfo {year} {2024}{\natexlab{b}})},\ \Eprint {http://arxiv.org/abs/2401.08044} {arXiv:2401.08044 [quant-ph]} \BibitemShut {NoStop}%
\bibitem [{\citenamefont {Ciavarella}\ and\ \citenamefont {Bauer}(2024)}]{Ciavarella:2024fzw}%
  \BibitemOpen
  \bibfield  {author} {\bibinfo {author} {\bibfnamefont {A.~N.}\ \bibnamefont {Ciavarella}}\ and\ \bibinfo {author} {\bibfnamefont {C.~W.}\ \bibnamefont {Bauer}},\ }\href {\doibase 10.1103/PhysRevLett.133.111901} {\bibfield  {journal} {\bibinfo  {journal} {Phys. Rev. Lett.}\ }\textbf {\bibinfo {volume} {133}},\ \bibinfo {pages} {111901} (\bibinfo {year} {2024})},\ \Eprint {http://arxiv.org/abs/2402.10265} {arXiv:2402.10265 [hep-ph]} \BibitemShut {NoStop}%
\bibitem [{\citenamefont {Gustafson}\ \emph {et~al.}(2024)\citenamefont {Gustafson}, \citenamefont {Ji}, \citenamefont {Lamm}, \citenamefont {Murairi}, \citenamefont {Perez},\ and\ \citenamefont {Zhu}}]{Gustafson:2024kym}%
  \BibitemOpen
  \bibfield  {author} {\bibinfo {author} {\bibfnamefont {E.~J.}\ \bibnamefont {Gustafson}}, \bibinfo {author} {\bibfnamefont {Y.}~\bibnamefont {Ji}}, \bibinfo {author} {\bibfnamefont {H.}~\bibnamefont {Lamm}}, \bibinfo {author} {\bibfnamefont {E.~M.}\ \bibnamefont {Murairi}}, \bibinfo {author} {\bibfnamefont {S.~O.}\ \bibnamefont {Perez}}, \ and\ \bibinfo {author} {\bibfnamefont {S.}~\bibnamefont {Zhu}},\ }\href {\doibase 10.1103/PhysRevD.110.034515} {\bibfield  {journal} {\bibinfo  {journal} {Phys. Rev. D}\ }\textbf {\bibinfo {volume} {110}},\ \bibinfo {pages} {034515} (\bibinfo {year} {2024})},\ \Eprint {http://arxiv.org/abs/2405.05973} {arXiv:2405.05973 [hep-lat]} \BibitemShut {NoStop}%
\bibitem [{\citenamefont {Balaji}\ \emph {et~al.}(2025)\citenamefont {Balaji}, \citenamefont {Conefrey-Shinozaki}, \citenamefont {Draper}, \citenamefont {Elhaderi}, \citenamefont {Gupta}, \citenamefont {Hidalgo}, \citenamefont {Lytle},\ and\ \citenamefont {Rinaldi}}]{Balaji:2025afl}%
  \BibitemOpen
  \bibfield  {author} {\bibinfo {author} {\bibfnamefont {P.}~\bibnamefont {Balaji}}, \bibinfo {author} {\bibfnamefont {C.}~\bibnamefont {Conefrey-Shinozaki}}, \bibinfo {author} {\bibfnamefont {P.}~\bibnamefont {Draper}}, \bibinfo {author} {\bibfnamefont {J.~K.}\ \bibnamefont {Elhaderi}}, \bibinfo {author} {\bibfnamefont {D.}~\bibnamefont {Gupta}}, \bibinfo {author} {\bibfnamefont {L.}~\bibnamefont {Hidalgo}}, \bibinfo {author} {\bibfnamefont {A.}~\bibnamefont {Lytle}}, \ and\ \bibinfo {author} {\bibfnamefont {E.}~\bibnamefont {Rinaldi}},\ }\href@noop {} {\  (\bibinfo {year} {2025})},\ \Eprint {http://arxiv.org/abs/2503.08866} {arXiv:2503.08866 [hep-lat]} \BibitemShut {NoStop}%
\bibitem [{\citenamefont {Ciavarella}\ \emph {et~al.}(2025{\natexlab{a}})\citenamefont {Ciavarella}, \citenamefont {Burbano},\ and\ \citenamefont {Bauer}}]{Ciavarella:2025bsg}%
  \BibitemOpen
  \bibfield  {author} {\bibinfo {author} {\bibfnamefont {A.~N.}\ \bibnamefont {Ciavarella}}, \bibinfo {author} {\bibfnamefont {I.~M.}\ \bibnamefont {Burbano}}, \ and\ \bibinfo {author} {\bibfnamefont {C.~W.}\ \bibnamefont {Bauer}},\ }\href@noop {} {\  (\bibinfo {year} {2025}{\natexlab{a}})},\ \Eprint {http://arxiv.org/abs/2503.11888} {arXiv:2503.11888 [hep-lat]} \BibitemShut {NoStop}%
\bibitem [{\citenamefont {Kreshchuk}\ \emph {et~al.}(2021)\citenamefont {Kreshchuk}, \citenamefont {Jia}, \citenamefont {Kirby}, \citenamefont {Goldstein}, \citenamefont {Vary},\ and\ \citenamefont {Love}}]{Kreshchuk:2020kcz}%
  \BibitemOpen
  \bibfield  {author} {\bibinfo {author} {\bibfnamefont {M.}~\bibnamefont {Kreshchuk}}, \bibinfo {author} {\bibfnamefont {S.}~\bibnamefont {Jia}}, \bibinfo {author} {\bibfnamefont {W.~M.}\ \bibnamefont {Kirby}}, \bibinfo {author} {\bibfnamefont {G.}~\bibnamefont {Goldstein}}, \bibinfo {author} {\bibfnamefont {J.~P.}\ \bibnamefont {Vary}}, \ and\ \bibinfo {author} {\bibfnamefont {P.~J.}\ \bibnamefont {Love}},\ }\href {\doibase 10.3390/e23050597} {\bibfield  {journal} {\bibinfo  {journal} {Entropy}\ }\textbf {\bibinfo {volume} {23}},\ \bibinfo {pages} {597} (\bibinfo {year} {2021})},\ \Eprint {http://arxiv.org/abs/2009.07885} {arXiv:2009.07885 [quant-ph]} \BibitemShut {NoStop}%
\bibitem [{\citenamefont {Kreshchuk}\ \emph {et~al.}(2022)\citenamefont {Kreshchuk}, \citenamefont {Kirby}, \citenamefont {Goldstein}, \citenamefont {Beauchemin},\ and\ \citenamefont {Love}}]{kreshchuk2022quantum}%
  \BibitemOpen
  \bibfield  {author} {\bibinfo {author} {\bibfnamefont {M.}~\bibnamefont {Kreshchuk}}, \bibinfo {author} {\bibfnamefont {W.~M.}\ \bibnamefont {Kirby}}, \bibinfo {author} {\bibfnamefont {G.}~\bibnamefont {Goldstein}}, \bibinfo {author} {\bibfnamefont {H.}~\bibnamefont {Beauchemin}}, \ and\ \bibinfo {author} {\bibfnamefont {P.~J.}\ \bibnamefont {Love}},\ }\href {\doibase 10.1103/PhysRevA.105.032418} {\bibfield  {journal} {\bibinfo  {journal} {Phys. Rev. A}\ }\textbf {\bibinfo {volume} {105}},\ \bibinfo {pages} {032418} (\bibinfo {year} {2022})}\BibitemShut {NoStop}%
\bibitem [{\citenamefont {Kreshchuk}\ \emph {et~al.}(2023)\citenamefont {Kreshchuk}, \citenamefont {Vary},\ and\ \citenamefont {Love}}]{kreshchuk2023simulatingscatteringcompositeparticles}%
  \BibitemOpen
  \bibfield  {author} {\bibinfo {author} {\bibfnamefont {M.}~\bibnamefont {Kreshchuk}}, \bibinfo {author} {\bibfnamefont {J.~P.}\ \bibnamefont {Vary}}, \ and\ \bibinfo {author} {\bibfnamefont {P.~J.}\ \bibnamefont {Love}},\ }\href {https://arxiv.org/abs/2310.13742} {\enquote {\bibinfo {title} {Simulating scattering of composite particles},}\ } (\bibinfo {year} {2023}),\ \Eprint {http://arxiv.org/abs/2310.13742} {arXiv:2310.13742 [quant-ph]} \BibitemShut {NoStop}%
\bibitem [{\citenamefont {Bauer}\ \emph {et~al.}(2021)\citenamefont {Bauer}, \citenamefont {de~Jong}, \citenamefont {Nachman},\ and\ \citenamefont {Provasoli}}]{Bauer:2019qxa}%
  \BibitemOpen
  \bibfield  {author} {\bibinfo {author} {\bibfnamefont {C.~W.}\ \bibnamefont {Bauer}}, \bibinfo {author} {\bibfnamefont {W.~A.}\ \bibnamefont {de~Jong}}, \bibinfo {author} {\bibfnamefont {B.}~\bibnamefont {Nachman}}, \ and\ \bibinfo {author} {\bibfnamefont {D.}~\bibnamefont {Provasoli}},\ }\href {\doibase 10.1103/PhysRevLett.126.062001} {\bibfield  {journal} {\bibinfo  {journal} {Phys. Rev. Lett.}\ }\textbf {\bibinfo {volume} {126}},\ \bibinfo {pages} {062001} (\bibinfo {year} {2021})},\ \Eprint {http://arxiv.org/abs/1904.03196} {arXiv:1904.03196 [hep-ph]} \BibitemShut {NoStop}%
\bibitem [{\citenamefont {Bauer}\ \emph {et~al.}(2024)\citenamefont {Bauer}, \citenamefont {Chigusa},\ and\ \citenamefont {Yamazaki}}]{Bauer:2023ujy}%
  \BibitemOpen
  \bibfield  {author} {\bibinfo {author} {\bibfnamefont {C.~W.}\ \bibnamefont {Bauer}}, \bibinfo {author} {\bibfnamefont {S.}~\bibnamefont {Chigusa}}, \ and\ \bibinfo {author} {\bibfnamefont {M.}~\bibnamefont {Yamazaki}},\ }\href {\doibase 10.1103/PhysRevA.109.032432} {\bibfield  {journal} {\bibinfo  {journal} {Phys. Rev. A}\ }\textbf {\bibinfo {volume} {109}},\ \bibinfo {pages} {032432} (\bibinfo {year} {2024})},\ \Eprint {http://arxiv.org/abs/2310.19881} {arXiv:2310.19881 [hep-ph]} \BibitemShut {NoStop}%
\bibitem [{\citenamefont {Chigusa}\ and\ \citenamefont {Yamazaki}(2022)}]{Chigusa:2022act}%
  \BibitemOpen
  \bibfield  {author} {\bibinfo {author} {\bibfnamefont {S.}~\bibnamefont {Chigusa}}\ and\ \bibinfo {author} {\bibfnamefont {M.}~\bibnamefont {Yamazaki}},\ }\href {\doibase 10.1016/j.physletb.2022.137466} {\bibfield  {journal} {\bibinfo  {journal} {Phys. Lett. B}\ }\textbf {\bibinfo {volume} {834}},\ \bibinfo {pages} {137466} (\bibinfo {year} {2022})},\ \Eprint {http://arxiv.org/abs/2204.12500} {arXiv:2204.12500 [hep-ph]} \BibitemShut {NoStop}%
\bibitem [{\citenamefont {Carena}\ \emph {et~al.}(2022{\natexlab{a}})\citenamefont {Carena}, \citenamefont {Lamm}, \citenamefont {Li},\ and\ \citenamefont {Liu}}]{Carena:2022kpg}%
  \BibitemOpen
  \bibfield  {author} {\bibinfo {author} {\bibfnamefont {M.}~\bibnamefont {Carena}}, \bibinfo {author} {\bibfnamefont {H.}~\bibnamefont {Lamm}}, \bibinfo {author} {\bibfnamefont {Y.-Y.}\ \bibnamefont {Li}}, \ and\ \bibinfo {author} {\bibfnamefont {W.}~\bibnamefont {Liu}},\ }\href {\doibase 10.1103/PhysRevLett.129.051601} {\bibfield  {journal} {\bibinfo  {journal} {Phys. Rev. Lett.}\ }\textbf {\bibinfo {volume} {129}},\ \bibinfo {pages} {051601} (\bibinfo {year} {2022}{\natexlab{a}})},\ \Eprint {http://arxiv.org/abs/2203.02823} {arXiv:2203.02823 [hep-lat]} \BibitemShut {NoStop}%
\bibitem [{\citenamefont {Gustafson}\ and\ \citenamefont {Van~de Water}(2024)}]{Gustafson:2023aai}%
  \BibitemOpen
  \bibfield  {author} {\bibinfo {author} {\bibfnamefont {E.}~\bibnamefont {Gustafson}}\ and\ \bibinfo {author} {\bibfnamefont {R.}~\bibnamefont {Van~de Water}},\ }\href {\doibase 10.22323/1.453.0215} {\bibfield  {journal} {\bibinfo  {journal} {PoS}\ }\textbf {\bibinfo {volume} {LATTICE2023}},\ \bibinfo {pages} {215} (\bibinfo {year} {2024})},\ \Eprint {http://arxiv.org/abs/2402.04317} {arXiv:2402.04317 [hep-lat]} \BibitemShut {NoStop}%
\bibitem [{\citenamefont {Ciavarella}(2023)}]{Ciavarella:2023mfc}%
  \BibitemOpen
  \bibfield  {author} {\bibinfo {author} {\bibfnamefont {A.~N.}\ \bibnamefont {Ciavarella}},\ }\href {\doibase 10.1103/PhysRevD.108.094513} {\bibfield  {journal} {\bibinfo  {journal} {Phys. Rev. D}\ }\textbf {\bibinfo {volume} {108}},\ \bibinfo {pages} {094513} (\bibinfo {year} {2023})},\ \Eprint {http://arxiv.org/abs/2307.05593} {arXiv:2307.05593 [hep-lat]} \BibitemShut {NoStop}%
\bibitem [{\citenamefont {Illa}\ \emph {et~al.}(2025{\natexlab{a}})\citenamefont {Illa}, \citenamefont {Savage},\ and\ \citenamefont {Yao}}]{Illa:2025dou}%
  \BibitemOpen
  \bibfield  {author} {\bibinfo {author} {\bibfnamefont {M.}~\bibnamefont {Illa}}, \bibinfo {author} {\bibfnamefont {M.~J.}\ \bibnamefont {Savage}}, \ and\ \bibinfo {author} {\bibfnamefont {X.}~\bibnamefont {Yao}},\ }\href@noop {} {\  (\bibinfo {year} {2025}{\natexlab{a}})},\ \Eprint {http://arxiv.org/abs/2503.09688} {arXiv:2503.09688 [hep-lat]} \BibitemShut {NoStop}%
\bibitem [{\citenamefont {Illa}\ \emph {et~al.}(2025{\natexlab{b}})\citenamefont {Illa}, \citenamefont {Savage},\ and\ \citenamefont {Yao}}]{Illa:2025njz}%
  \BibitemOpen
  \bibfield  {author} {\bibinfo {author} {\bibfnamefont {M.}~\bibnamefont {Illa}}, \bibinfo {author} {\bibfnamefont {M.~J.}\ \bibnamefont {Savage}}, \ and\ \bibinfo {author} {\bibfnamefont {X.}~\bibnamefont {Yao}},\ }\href@noop {} {\  (\bibinfo {year} {2025}{\natexlab{b}})},\ \Eprint {http://arxiv.org/abs/2504.21575} {arXiv:2504.21575 [quant-ph]} \BibitemShut {NoStop}%
\bibitem [{\citenamefont {Roggero}\ \emph {et~al.}(2020)\citenamefont {Roggero}, \citenamefont {Li}, \citenamefont {Carlson}, \citenamefont {Gupta},\ and\ \citenamefont {Perdue}}]{Roggero:2019myu}%
  \BibitemOpen
  \bibfield  {author} {\bibinfo {author} {\bibfnamefont {A.}~\bibnamefont {Roggero}}, \bibinfo {author} {\bibfnamefont {A.~C.~Y.}\ \bibnamefont {Li}}, \bibinfo {author} {\bibfnamefont {J.}~\bibnamefont {Carlson}}, \bibinfo {author} {\bibfnamefont {R.}~\bibnamefont {Gupta}}, \ and\ \bibinfo {author} {\bibfnamefont {G.~N.}\ \bibnamefont {Perdue}},\ }\href {\doibase 10.1103/PhysRevD.101.074038} {\bibfield  {journal} {\bibinfo  {journal} {Phys. Rev. D}\ }\textbf {\bibinfo {volume} {101}},\ \bibinfo {pages} {074038} (\bibinfo {year} {2020})},\ \Eprint {http://arxiv.org/abs/1911.06368} {arXiv:1911.06368 [quant-ph]} \BibitemShut {NoStop}%
\bibitem [{\citenamefont {Watson}\ \emph {et~al.}(2023)\citenamefont {Watson}, \citenamefont {Bringewatt}, \citenamefont {Shaw}, \citenamefont {Childs}, \citenamefont {Gorshkov},\ and\ \citenamefont {Davoudi}}]{Watson:2023oov}%
  \BibitemOpen
  \bibfield  {author} {\bibinfo {author} {\bibfnamefont {J.~D.}\ \bibnamefont {Watson}}, \bibinfo {author} {\bibfnamefont {J.}~\bibnamefont {Bringewatt}}, \bibinfo {author} {\bibfnamefont {A.~F.}\ \bibnamefont {Shaw}}, \bibinfo {author} {\bibfnamefont {A.~M.}\ \bibnamefont {Childs}}, \bibinfo {author} {\bibfnamefont {A.~V.}\ \bibnamefont {Gorshkov}}, \ and\ \bibinfo {author} {\bibfnamefont {Z.}~\bibnamefont {Davoudi}},\ }\href@noop {} {\  (\bibinfo {year} {2023})},\ \Eprint {http://arxiv.org/abs/2312.05344} {arXiv:2312.05344 [quant-ph]} \BibitemShut {NoStop}%
\bibitem [{\citenamefont {Burbano}\ \emph {et~al.}(2025)\citenamefont {Burbano}, \citenamefont {Carrillo}, \citenamefont {Urek}, \citenamefont {Ciavarella},\ and\ \citenamefont {Brice\~no}}]{Burbano:2025pef}%
  \BibitemOpen
  \bibfield  {author} {\bibinfo {author} {\bibfnamefont {I.~M.}\ \bibnamefont {Burbano}}, \bibinfo {author} {\bibfnamefont {M.~A.}\ \bibnamefont {Carrillo}}, \bibinfo {author} {\bibfnamefont {R.}~\bibnamefont {Urek}}, \bibinfo {author} {\bibfnamefont {A.~N.}\ \bibnamefont {Ciavarella}}, \ and\ \bibinfo {author} {\bibfnamefont {R.~A.}\ \bibnamefont {Brice\~no}},\ }\href@noop {} {\  (\bibinfo {year} {2025})},\ \Eprint {http://arxiv.org/abs/2506.06511} {arXiv:2506.06511 [hep-lat]} \BibitemShut {NoStop}%
\bibitem [{\citenamefont {Tagliacozzo}\ \emph {et~al.}(2013)\citenamefont {Tagliacozzo}, \citenamefont {Celi}, \citenamefont {Orland},\ and\ \citenamefont {Lewenstein}}]{Tagliacozzo:2012df}%
  \BibitemOpen
  \bibfield  {author} {\bibinfo {author} {\bibfnamefont {L.}~\bibnamefont {Tagliacozzo}}, \bibinfo {author} {\bibfnamefont {A.}~\bibnamefont {Celi}}, \bibinfo {author} {\bibfnamefont {P.}~\bibnamefont {Orland}}, \ and\ \bibinfo {author} {\bibfnamefont {M.}~\bibnamefont {Lewenstein}},\ }\href {\doibase 10.1038/ncomms3615} {\bibfield  {journal} {\bibinfo  {journal} {Nature Commun.}\ }\textbf {\bibinfo {volume} {4}},\ \bibinfo {pages} {2615} (\bibinfo {year} {2013})},\ \Eprint {http://arxiv.org/abs/1211.2704} {arXiv:1211.2704 [cond-mat.quant-gas]} \BibitemShut {NoStop}%
\bibitem [{\citenamefont {Bazavov}\ \emph {et~al.}(2015)\citenamefont {Bazavov}, \citenamefont {Meurice}, \citenamefont {Tsai}, \citenamefont {Unmuth-Yockey},\ and\ \citenamefont {Zhang}}]{Bazavov:2015kka}%
  \BibitemOpen
  \bibfield  {author} {\bibinfo {author} {\bibfnamefont {A.}~\bibnamefont {Bazavov}}, \bibinfo {author} {\bibfnamefont {Y.}~\bibnamefont {Meurice}}, \bibinfo {author} {\bibfnamefont {S.-W.}\ \bibnamefont {Tsai}}, \bibinfo {author} {\bibfnamefont {J.}~\bibnamefont {Unmuth-Yockey}}, \ and\ \bibinfo {author} {\bibfnamefont {J.}~\bibnamefont {Zhang}},\ }\href {\doibase 10.1103/PhysRevD.92.076003} {\bibfield  {journal} {\bibinfo  {journal} {Phys. Rev. D}\ }\textbf {\bibinfo {volume} {92}},\ \bibinfo {pages} {076003} (\bibinfo {year} {2015})},\ \Eprint {http://arxiv.org/abs/1503.08354} {arXiv:1503.08354 [hep-lat]} \BibitemShut {NoStop}%
\bibitem [{\citenamefont {Jordan}\ \emph {et~al.}(2018)\citenamefont {Jordan}, \citenamefont {Krovi}, \citenamefont {Lee},\ and\ \citenamefont {Preskill}}]{Jordan:2017lea}%
  \BibitemOpen
  \bibfield  {author} {\bibinfo {author} {\bibfnamefont {S.~P.}\ \bibnamefont {Jordan}}, \bibinfo {author} {\bibfnamefont {H.}~\bibnamefont {Krovi}}, \bibinfo {author} {\bibfnamefont {K.~S.~M.}\ \bibnamefont {Lee}}, \ and\ \bibinfo {author} {\bibfnamefont {J.}~\bibnamefont {Preskill}},\ }\href {\doibase 10.22331/q-2018-01-08-44} {\bibfield  {journal} {\bibinfo  {journal} {Quantum}\ }\textbf {\bibinfo {volume} {2}},\ \bibinfo {pages} {44} (\bibinfo {year} {2018})},\ \Eprint {http://arxiv.org/abs/1703.00454} {arXiv:1703.00454 [quant-ph]} \BibitemShut {NoStop}%
\bibitem [{\citenamefont {Gonz\'alez-Cuadra}\ \emph {et~al.}(2017)\citenamefont {Gonz\'alez-Cuadra}, \citenamefont {Zohar},\ and\ \citenamefont {Cirac}}]{Gonzalez-Cuadra:2017lvz}%
  \BibitemOpen
  \bibfield  {author} {\bibinfo {author} {\bibfnamefont {D.}~\bibnamefont {Gonz\'alez-Cuadra}}, \bibinfo {author} {\bibfnamefont {E.}~\bibnamefont {Zohar}}, \ and\ \bibinfo {author} {\bibfnamefont {J.~I.}\ \bibnamefont {Cirac}},\ }\href {\doibase 10.1088/1367-2630/aa6f37} {\bibfield  {journal} {\bibinfo  {journal} {New J. Phys.}\ }\textbf {\bibinfo {volume} {19}},\ \bibinfo {pages} {063038} (\bibinfo {year} {2017})},\ \Eprint {http://arxiv.org/abs/1702.05492} {arXiv:1702.05492 [quant-ph]} \BibitemShut {NoStop}%
\bibitem [{\citenamefont {G\"org}\ \emph {et~al.}(2019)\citenamefont {G\"org}, \citenamefont {Sandholzer}, \citenamefont {Minguzzi}, \citenamefont {Desbuquois}, \citenamefont {Messer},\ and\ \citenamefont {Esslinger}}]{Gorg:2018xyc}%
  \BibitemOpen
  \bibfield  {author} {\bibinfo {author} {\bibfnamefont {F.}~\bibnamefont {G\"org}}, \bibinfo {author} {\bibfnamefont {K.}~\bibnamefont {Sandholzer}}, \bibinfo {author} {\bibfnamefont {J.}~\bibnamefont {Minguzzi}}, \bibinfo {author} {\bibfnamefont {R.}~\bibnamefont {Desbuquois}}, \bibinfo {author} {\bibfnamefont {M.}~\bibnamefont {Messer}}, \ and\ \bibinfo {author} {\bibfnamefont {T.}~\bibnamefont {Esslinger}},\ }\href {\doibase 10.1038/s41567-019-0615-4} {\bibfield  {journal} {\bibinfo  {journal} {Nature Phys.}\ }\textbf {\bibinfo {volume} {15}},\ \bibinfo {pages} {1161} (\bibinfo {year} {2019})},\ \Eprint {http://arxiv.org/abs/1812.05895} {arXiv:1812.05895 [cond-mat.quant-gas]} \BibitemShut {NoStop}%
\bibitem [{\citenamefont {Lamm}\ \emph {et~al.}(2019)\citenamefont {Lamm}, \citenamefont {Lawrence},\ and\ \citenamefont {Yamauchi}}]{Lamm:2019bik}%
  \BibitemOpen
  \bibfield  {author} {\bibinfo {author} {\bibfnamefont {H.}~\bibnamefont {Lamm}}, \bibinfo {author} {\bibfnamefont {S.}~\bibnamefont {Lawrence}}, \ and\ \bibinfo {author} {\bibfnamefont {Y.}~\bibnamefont {Yamauchi}} (\bibinfo {collaboration} {NuQS}),\ }\href {\doibase 10.1103/PhysRevD.100.034518} {\bibfield  {journal} {\bibinfo  {journal} {Phys. Rev. D}\ }\textbf {\bibinfo {volume} {100}},\ \bibinfo {pages} {034518} (\bibinfo {year} {2019})},\ \Eprint {http://arxiv.org/abs/1903.08807} {arXiv:1903.08807 [hep-lat]} \BibitemShut {NoStop}%
\bibitem [{\citenamefont {Zohar}\ and\ \citenamefont {Cirac}(2019)}]{Zohar:2019ygc}%
  \BibitemOpen
  \bibfield  {author} {\bibinfo {author} {\bibfnamefont {E.}~\bibnamefont {Zohar}}\ and\ \bibinfo {author} {\bibfnamefont {J.~I.}\ \bibnamefont {Cirac}},\ }\href {\doibase 10.1103/PhysRevD.99.114511} {\bibfield  {journal} {\bibinfo  {journal} {Phys. Rev. D}\ }\textbf {\bibinfo {volume} {99}},\ \bibinfo {pages} {114511} (\bibinfo {year} {2019})},\ \Eprint {http://arxiv.org/abs/1905.00652} {arXiv:1905.00652 [quant-ph]} \BibitemShut {NoStop}%
\bibitem [{\citenamefont {Buser}\ \emph {et~al.}(2021)\citenamefont {Buser}, \citenamefont {Gharibyan}, \citenamefont {Hanada}, \citenamefont {Honda},\ and\ \citenamefont {Liu}}]{Buser:2020cvn}%
  \BibitemOpen
  \bibfield  {author} {\bibinfo {author} {\bibfnamefont {A.~J.}\ \bibnamefont {Buser}}, \bibinfo {author} {\bibfnamefont {H.}~\bibnamefont {Gharibyan}}, \bibinfo {author} {\bibfnamefont {M.}~\bibnamefont {Hanada}}, \bibinfo {author} {\bibfnamefont {M.}~\bibnamefont {Honda}}, \ and\ \bibinfo {author} {\bibfnamefont {J.}~\bibnamefont {Liu}},\ }\href {\doibase 10.1007/JHEP09(2021)034} {\bibfield  {journal} {\bibinfo  {journal} {JHEP}\ }\textbf {\bibinfo {volume} {09}},\ \bibinfo {pages} {034} (\bibinfo {year} {2021})},\ \Eprint {http://arxiv.org/abs/2011.06576} {arXiv:2011.06576 [hep-th]} \BibitemShut {NoStop}%
\bibitem [{\citenamefont {Barata}\ \emph {et~al.}(2021)\citenamefont {Barata}, \citenamefont {Mueller}, \citenamefont {Tarasov},\ and\ \citenamefont {Venugopalan}}]{Barata:2020jtq}%
  \BibitemOpen
  \bibfield  {author} {\bibinfo {author} {\bibfnamefont {J.~a.}\ \bibnamefont {Barata}}, \bibinfo {author} {\bibfnamefont {N.}~\bibnamefont {Mueller}}, \bibinfo {author} {\bibfnamefont {A.}~\bibnamefont {Tarasov}}, \ and\ \bibinfo {author} {\bibfnamefont {R.}~\bibnamefont {Venugopalan}},\ }\href {\doibase 10.1103/PhysRevA.103.042410} {\bibfield  {journal} {\bibinfo  {journal} {Phys. Rev. A}\ }\textbf {\bibinfo {volume} {103}},\ \bibinfo {pages} {042410} (\bibinfo {year} {2021})},\ \Eprint {http://arxiv.org/abs/2012.00020} {arXiv:2012.00020 [hep-th]} \BibitemShut {NoStop}%
\bibitem [{\citenamefont {Stryker}(2021)}]{Stryker:2021asy}%
  \BibitemOpen
  \bibfield  {author} {\bibinfo {author} {\bibfnamefont {J.~R.}\ \bibnamefont {Stryker}},\ }\href@noop {} {\  (\bibinfo {year} {2021})},\ \Eprint {http://arxiv.org/abs/2105.11548} {arXiv:2105.11548 [hep-lat]} \BibitemShut {NoStop}%
\bibitem [{\citenamefont {Davoudi}\ \emph {et~al.}(2024{\natexlab{a}})\citenamefont {Davoudi}, \citenamefont {Jarzynski}, \citenamefont {Mueller}, \citenamefont {Oruganti}, \citenamefont {Powers},\ and\ \citenamefont {Halpern}}]{Davoudi:2024osg}%
  \BibitemOpen
  \bibfield  {author} {\bibinfo {author} {\bibfnamefont {Z.}~\bibnamefont {Davoudi}}, \bibinfo {author} {\bibfnamefont {C.}~\bibnamefont {Jarzynski}}, \bibinfo {author} {\bibfnamefont {N.}~\bibnamefont {Mueller}}, \bibinfo {author} {\bibfnamefont {G.}~\bibnamefont {Oruganti}}, \bibinfo {author} {\bibfnamefont {C.}~\bibnamefont {Powers}}, \ and\ \bibinfo {author} {\bibfnamefont {N.~Y.}\ \bibnamefont {Halpern}},\ }\href {\doibase 10.1103/PhysRevLett.133.250402} {\bibfield  {journal} {\bibinfo  {journal} {Phys. Rev. Lett.}\ }\textbf {\bibinfo {volume} {133}},\ \bibinfo {pages} {250402} (\bibinfo {year} {2024}{\natexlab{a}})},\ \Eprint {http://arxiv.org/abs/2404.02965} {arXiv:2404.02965 [quant-ph]} \BibitemShut {NoStop}%
\bibitem [{\citenamefont {Mueller}\ \emph {et~al.}(2024)\citenamefont {Mueller}, \citenamefont {Wang}, \citenamefont {Katz}, \citenamefont {Davoudi},\ and\ \citenamefont {Cetina}}]{Mueller:2024mmk}%
  \BibitemOpen
  \bibfield  {author} {\bibinfo {author} {\bibfnamefont {N.}~\bibnamefont {Mueller}}, \bibinfo {author} {\bibfnamefont {T.}~\bibnamefont {Wang}}, \bibinfo {author} {\bibfnamefont {O.}~\bibnamefont {Katz}}, \bibinfo {author} {\bibfnamefont {Z.}~\bibnamefont {Davoudi}}, \ and\ \bibinfo {author} {\bibfnamefont {M.}~\bibnamefont {Cetina}},\ }\href@noop {} {\  (\bibinfo {year} {2024})},\ \Eprint {http://arxiv.org/abs/2408.00069} {arXiv:2408.00069 [quant-ph]} \BibitemShut {NoStop}%
\bibitem [{\citenamefont {Davoudi}\ \emph {et~al.}(2025{\natexlab{a}})\citenamefont {Davoudi}, \citenamefont {Jarzynski}, \citenamefont {Mueller}, \citenamefont {Oruganti}, \citenamefont {Powers},\ and\ \citenamefont {Halpern}}]{Davoudi:2025tbi}%
  \BibitemOpen
  \bibfield  {author} {\bibinfo {author} {\bibfnamefont {Z.}~\bibnamefont {Davoudi}}, \bibinfo {author} {\bibfnamefont {C.}~\bibnamefont {Jarzynski}}, \bibinfo {author} {\bibfnamefont {N.}~\bibnamefont {Mueller}}, \bibinfo {author} {\bibfnamefont {G.}~\bibnamefont {Oruganti}}, \bibinfo {author} {\bibfnamefont {C.}~\bibnamefont {Powers}}, \ and\ \bibinfo {author} {\bibfnamefont {N.~Y.}\ \bibnamefont {Halpern}},\ }\href@noop {} {\  (\bibinfo {year} {2025}{\natexlab{a}})},\ \Eprint {http://arxiv.org/abs/2502.19418} {arXiv:2502.19418 [quant-ph]} \BibitemShut {NoStop}%
\bibitem [{\citenamefont {Surace}\ \emph {et~al.}(2024)\citenamefont {Surace} \emph {et~al.}}]{Surace:2024bht}%
  \BibitemOpen
  \bibfield  {author} {\bibinfo {author} {\bibfnamefont {F.~M.}\ \bibnamefont {Surace}} \emph {et~al.},\ }\href@noop {} {\  (\bibinfo {year} {2024})},\ \Eprint {http://arxiv.org/abs/2411.10652} {arXiv:2411.10652 [quant-ph]} \BibitemShut {NoStop}%
\bibitem [{\citenamefont {De}\ \emph {et~al.}(2024)\citenamefont {De} \emph {et~al.}}]{De:2024smi}%
  \BibitemOpen
  \bibfield  {author} {\bibinfo {author} {\bibfnamefont {A.}~\bibnamefont {De}} \emph {et~al.},\ }\href@noop {} {\  (\bibinfo {year} {2024})},\ \Eprint {http://arxiv.org/abs/2410.13815} {arXiv:2410.13815 [quant-ph]} \BibitemShut {NoStop}%
\bibitem [{\citenamefont {Davoudi}\ \emph {et~al.}(2024{\natexlab{b}})\citenamefont {Davoudi}, \citenamefont {Hsieh},\ and\ \citenamefont {Kadam}}]{Davoudi:2024wyv}%
  \BibitemOpen
  \bibfield  {author} {\bibinfo {author} {\bibfnamefont {Z.}~\bibnamefont {Davoudi}}, \bibinfo {author} {\bibfnamefont {C.-C.}\ \bibnamefont {Hsieh}}, \ and\ \bibinfo {author} {\bibfnamefont {S.~V.}\ \bibnamefont {Kadam}},\ }\href {\doibase 10.22331/q-2024-11-11-1520} {\bibfield  {journal} {\bibinfo  {journal} {Quantum}\ }\textbf {\bibinfo {volume} {8}},\ \bibinfo {pages} {1520} (\bibinfo {year} {2024}{\natexlab{b}})},\ \Eprint {http://arxiv.org/abs/2402.00840} {arXiv:2402.00840 [quant-ph]} \BibitemShut {NoStop}%
\bibitem [{\citenamefont {Bennewitz}\ \emph {et~al.}(2024)\citenamefont {Bennewitz} \emph {et~al.}}]{Bennewitz:2024ixi}%
  \BibitemOpen
  \bibfield  {author} {\bibinfo {author} {\bibfnamefont {E.~R.}\ \bibnamefont {Bennewitz}} \emph {et~al.},\ }\href@noop {} {\  (\bibinfo {year} {2024})},\ \Eprint {http://arxiv.org/abs/2403.07061} {arXiv:2403.07061 [quant-ph]} \BibitemShut {NoStop}%
\bibitem [{\citenamefont {Mueller}\ \emph {et~al.}(2023)\citenamefont {Mueller}, \citenamefont {Carolan}, \citenamefont {Connelly}, \citenamefont {Davoudi}, \citenamefont {Dumitrescu},\ and\ \citenamefont {Yeter-Aydeniz}}]{Mueller:2022xbg}%
  \BibitemOpen
  \bibfield  {author} {\bibinfo {author} {\bibfnamefont {N.}~\bibnamefont {Mueller}}, \bibinfo {author} {\bibfnamefont {J.~A.}\ \bibnamefont {Carolan}}, \bibinfo {author} {\bibfnamefont {A.}~\bibnamefont {Connelly}}, \bibinfo {author} {\bibfnamefont {Z.}~\bibnamefont {Davoudi}}, \bibinfo {author} {\bibfnamefont {E.~F.}\ \bibnamefont {Dumitrescu}}, \ and\ \bibinfo {author} {\bibfnamefont {K.}~\bibnamefont {Yeter-Aydeniz}},\ }\href {\doibase 10.1103/PRXQuantum.4.030323} {\bibfield  {journal} {\bibinfo  {journal} {PRX Quantum}\ }\textbf {\bibinfo {volume} {4}},\ \bibinfo {pages} {030323} (\bibinfo {year} {2023})},\ \Eprint {http://arxiv.org/abs/2210.03089} {arXiv:2210.03089 [quant-ph]} \BibitemShut {NoStop}%
\bibitem [{\citenamefont {K\"urk\c{c}\"uoglu}\ \emph {et~al.}(2024)\citenamefont {K\"urk\c{c}\"uoglu}, \citenamefont {Lamm},\ and\ \citenamefont {Maestri}}]{Kurkcuoglu:2024cfv}%
  \BibitemOpen
  \bibfield  {author} {\bibinfo {author} {\bibfnamefont {D.~M.}\ \bibnamefont {K\"urk\c{c}\"uoglu}}, \bibinfo {author} {\bibfnamefont {H.}~\bibnamefont {Lamm}}, \ and\ \bibinfo {author} {\bibfnamefont {A.}~\bibnamefont {Maestri}},\ }\href@noop {} {\  (\bibinfo {year} {2024})},\ \Eprint {http://arxiv.org/abs/2410.16414} {arXiv:2410.16414 [quant-ph]} \BibitemShut {NoStop}%
\bibitem [{\citenamefont {Carena}\ \emph {et~al.}(2024)\citenamefont {Carena}, \citenamefont {Lamm}, \citenamefont {Li},\ and\ \citenamefont {Liu}}]{Carena:2024dzu}%
  \BibitemOpen
  \bibfield  {author} {\bibinfo {author} {\bibfnamefont {M.}~\bibnamefont {Carena}}, \bibinfo {author} {\bibfnamefont {H.}~\bibnamefont {Lamm}}, \bibinfo {author} {\bibfnamefont {Y.-Y.}\ \bibnamefont {Li}}, \ and\ \bibinfo {author} {\bibfnamefont {W.}~\bibnamefont {Liu}},\ }\href {\doibase 10.1103/PhysRevD.110.054516} {\bibfield  {journal} {\bibinfo  {journal} {Phys. Rev. D}\ }\textbf {\bibinfo {volume} {110}},\ \bibinfo {pages} {054516} (\bibinfo {year} {2024})},\ \Eprint {http://arxiv.org/abs/2402.16780} {arXiv:2402.16780 [hep-lat]} \BibitemShut {NoStop}%
\bibitem [{\citenamefont {Ciavarella}(2025)}]{Ciavarella:2024lsp}%
  \BibitemOpen
  \bibfield  {author} {\bibinfo {author} {\bibfnamefont {A.~N.}\ \bibnamefont {Ciavarella}},\ }\href {\doibase 10.1103/PhysRevD.111.054501} {\bibfield  {journal} {\bibinfo  {journal} {Phys. Rev. D}\ }\textbf {\bibinfo {volume} {111}},\ \bibinfo {pages} {054501} (\bibinfo {year} {2025})},\ \Eprint {http://arxiv.org/abs/2411.05915} {arXiv:2411.05915 [quant-ph]} \BibitemShut {NoStop}%
\bibitem [{\citenamefont {Ciavarella}\ \emph {et~al.}(2025{\natexlab{b}})\citenamefont {Ciavarella}, \citenamefont {Bauer},\ and\ \citenamefont {Halimeh}}]{Ciavarella:2025zqf}%
  \BibitemOpen
  \bibfield  {author} {\bibinfo {author} {\bibfnamefont {A.~N.}\ \bibnamefont {Ciavarella}}, \bibinfo {author} {\bibfnamefont {C.~W.}\ \bibnamefont {Bauer}}, \ and\ \bibinfo {author} {\bibfnamefont {J.~C.}\ \bibnamefont {Halimeh}},\ }\href@noop {} {\  (\bibinfo {year} {2025}{\natexlab{b}})},\ \Eprint {http://arxiv.org/abs/2502.03533} {arXiv:2502.03533 [quant-ph]} \BibitemShut {NoStop}%
\bibitem [{\citenamefont {Ingoldby}\ \emph {et~al.}(2025)\citenamefont {Ingoldby}, \citenamefont {Spannowsky}, \citenamefont {Sypchenko}, \citenamefont {Williams},\ and\ \citenamefont {Wingate}}]{Ingoldby:2025bdb}%
  \BibitemOpen
  \bibfield  {author} {\bibinfo {author} {\bibfnamefont {J.}~\bibnamefont {Ingoldby}}, \bibinfo {author} {\bibfnamefont {M.}~\bibnamefont {Spannowsky}}, \bibinfo {author} {\bibfnamefont {T.}~\bibnamefont {Sypchenko}}, \bibinfo {author} {\bibfnamefont {S.}~\bibnamefont {Williams}}, \ and\ \bibinfo {author} {\bibfnamefont {M.}~\bibnamefont {Wingate}},\ }\href@noop {} {\  (\bibinfo {year} {2025})},\ \Eprint {http://arxiv.org/abs/2505.03878} {arXiv:2505.03878 [quant-ph]} \BibitemShut {NoStop}%
\bibitem [{\citenamefont {Gupta}\ \emph {et~al.}(2025)\citenamefont {Gupta}, \citenamefont {White},\ and\ \citenamefont {Davoudi}}]{Gupta:2025xti}%
  \BibitemOpen
  \bibfield  {author} {\bibinfo {author} {\bibfnamefont {N.}~\bibnamefont {Gupta}}, \bibinfo {author} {\bibfnamefont {C.~D.}\ \bibnamefont {White}}, \ and\ \bibinfo {author} {\bibfnamefont {Z.}~\bibnamefont {Davoudi}},\ }\href@noop {} {\  (\bibinfo {year} {2025})},\ \Eprint {http://arxiv.org/abs/2506.02313} {arXiv:2506.02313 [quant-ph]} \BibitemShut {NoStop}%
\bibitem [{\citenamefont {Bauer}(2025)}]{Bauer:2025nzf}%
  \BibitemOpen
  \bibfield  {author} {\bibinfo {author} {\bibfnamefont {C.~W.}\ \bibnamefont {Bauer}},\ }\href@noop {} {\  (\bibinfo {year} {2025})},\ \Eprint {http://arxiv.org/abs/2503.16602} {arXiv:2503.16602 [hep-ph]} \BibitemShut {NoStop}%
\bibitem [{\citenamefont {Alvi}\ \emph {et~al.}(2023)\citenamefont {Alvi}, \citenamefont {Bauer},\ and\ \citenamefont {Nachman}}]{Alvi:2022fkk}%
  \BibitemOpen
  \bibfield  {author} {\bibinfo {author} {\bibfnamefont {S.}~\bibnamefont {Alvi}}, \bibinfo {author} {\bibfnamefont {C.~W.}\ \bibnamefont {Bauer}}, \ and\ \bibinfo {author} {\bibfnamefont {B.}~\bibnamefont {Nachman}},\ }\href {\doibase 10.1007/JHEP02(2023)220} {\bibfield  {journal} {\bibinfo  {journal} {JHEP}\ }\textbf {\bibinfo {volume} {02}},\ \bibinfo {pages} {220} (\bibinfo {year} {2023})},\ \Eprint {http://arxiv.org/abs/2206.08391} {arXiv:2206.08391 [hep-ph]} \BibitemShut {NoStop}%
\bibitem [{\citenamefont {Kan}\ and\ \citenamefont {Nam}(2021)}]{Kan:2021xfc}%
  \BibitemOpen
  \bibfield  {author} {\bibinfo {author} {\bibfnamefont {A.}~\bibnamefont {Kan}}\ and\ \bibinfo {author} {\bibfnamefont {Y.}~\bibnamefont {Nam}},\ }\href@noop {} {\  (\bibinfo {year} {2021})},\ \Eprint {http://arxiv.org/abs/2107.12769} {arXiv:2107.12769 [quant-ph]} \BibitemShut {NoStop}%
\bibitem [{\citenamefont {Alexandru}\ \emph {et~al.}(2023)\citenamefont {Alexandru}, \citenamefont {Bedaque}, \citenamefont {Carosso}, \citenamefont {Cervia},\ and\ \citenamefont {Sheng}}]{Alexandru:2022son}%
  \BibitemOpen
  \bibfield  {author} {\bibinfo {author} {\bibfnamefont {A.}~\bibnamefont {Alexandru}}, \bibinfo {author} {\bibfnamefont {P.~F.}\ \bibnamefont {Bedaque}}, \bibinfo {author} {\bibfnamefont {A.}~\bibnamefont {Carosso}}, \bibinfo {author} {\bibfnamefont {M.~J.}\ \bibnamefont {Cervia}}, \ and\ \bibinfo {author} {\bibfnamefont {A.}~\bibnamefont {Sheng}},\ }\href {\doibase 10.1103/PhysRevD.107.034503} {\bibfield  {journal} {\bibinfo  {journal} {Phys. Rev. D}\ }\textbf {\bibinfo {volume} {107}},\ \bibinfo {pages} {034503} (\bibinfo {year} {2023})},\ \Eprint {http://arxiv.org/abs/2209.00098} {arXiv:2209.00098 [hep-lat]} \BibitemShut {NoStop}%
\bibitem [{\citenamefont {Tong}\ \emph {et~al.}(2022)\citenamefont {Tong}, \citenamefont {Albert}, \citenamefont {McClean}, \citenamefont {Preskill},\ and\ \citenamefont {Su}}]{tong2022provably}%
  \BibitemOpen
  \bibfield  {author} {\bibinfo {author} {\bibfnamefont {Y.}~\bibnamefont {Tong}}, \bibinfo {author} {\bibfnamefont {V.~V.}\ \bibnamefont {Albert}}, \bibinfo {author} {\bibfnamefont {J.~R.}\ \bibnamefont {McClean}}, \bibinfo {author} {\bibfnamefont {J.}~\bibnamefont {Preskill}}, \ and\ \bibinfo {author} {\bibfnamefont {Y.}~\bibnamefont {Su}},\ }\href {\doibase 10.22331/q-2022-09-22-816} {\bibfield  {journal} {\bibinfo  {journal} {Quantum}\ }\textbf {\bibinfo {volume} {6}},\ \bibinfo {pages} {816} (\bibinfo {year} {2022})}\BibitemShut {NoStop}%
\bibitem [{\citenamefont {Davoudi}\ \emph {et~al.}(2023)\citenamefont {Davoudi}, \citenamefont {Shaw},\ and\ \citenamefont {Stryker}}]{Davoudi:2022xmb}%
  \BibitemOpen
  \bibfield  {author} {\bibinfo {author} {\bibfnamefont {Z.}~\bibnamefont {Davoudi}}, \bibinfo {author} {\bibfnamefont {A.~F.}\ \bibnamefont {Shaw}}, \ and\ \bibinfo {author} {\bibfnamefont {J.~R.}\ \bibnamefont {Stryker}},\ }\href {\doibase 10.22331/q-2023-12-20-1213} {\bibfield  {journal} {\bibinfo  {journal} {Quantum}\ }\textbf {\bibinfo {volume} {7}},\ \bibinfo {pages} {1213} (\bibinfo {year} {2023})},\ \Eprint {http://arxiv.org/abs/2212.14030} {arXiv:2212.14030 [hep-lat]} \BibitemShut {NoStop}%
\bibitem [{\citenamefont {Kane}\ \emph {et~al.}(2024)\citenamefont {Kane}, \citenamefont {Gomes},\ and\ \citenamefont {Kreshchuk}}]{Kane:2023jdo}%
  \BibitemOpen
  \bibfield  {author} {\bibinfo {author} {\bibfnamefont {C.~F.}\ \bibnamefont {Kane}}, \bibinfo {author} {\bibfnamefont {N.}~\bibnamefont {Gomes}}, \ and\ \bibinfo {author} {\bibfnamefont {M.}~\bibnamefont {Kreshchuk}},\ }\href {\doibase 10.1103/PhysRevA.110.012455} {\bibfield  {journal} {\bibinfo  {journal} {Phys. Rev. A}\ }\textbf {\bibinfo {volume} {110}},\ \bibinfo {pages} {012455} (\bibinfo {year} {2024})},\ \Eprint {http://arxiv.org/abs/2310.13757} {arXiv:2310.13757 [quant-ph]} \BibitemShut {NoStop}%
\bibitem [{\citenamefont {Hariprakash}\ \emph {et~al.}(2025)\citenamefont {Hariprakash}, \citenamefont {Modi}, \citenamefont {Kreshchuk}, \citenamefont {Kane},\ and\ \citenamefont {Bauer}}]{Hariprakash:2023tla}%
  \BibitemOpen
  \bibfield  {author} {\bibinfo {author} {\bibfnamefont {S.}~\bibnamefont {Hariprakash}}, \bibinfo {author} {\bibfnamefont {N.~S.}\ \bibnamefont {Modi}}, \bibinfo {author} {\bibfnamefont {M.}~\bibnamefont {Kreshchuk}}, \bibinfo {author} {\bibfnamefont {C.~F.}\ \bibnamefont {Kane}}, \ and\ \bibinfo {author} {\bibfnamefont {C.~W.}\ \bibnamefont {Bauer}},\ }\href {\doibase 10.1103/PhysRevA.111.022419} {\bibfield  {journal} {\bibinfo  {journal} {Phys. Rev. A}\ }\textbf {\bibinfo {volume} {111}},\ \bibinfo {pages} {022419} (\bibinfo {year} {2025})},\ \Eprint {http://arxiv.org/abs/2312.11637} {arXiv:2312.11637 [quant-ph]} \BibitemShut {NoStop}%
\bibitem [{\citenamefont {Rhodes}\ \emph {et~al.}(2024)\citenamefont {Rhodes}, \citenamefont {Kreshchuk},\ and\ \citenamefont {Pathak}}]{Rhodes:2024zbr}%
  \BibitemOpen
  \bibfield  {author} {\bibinfo {author} {\bibfnamefont {M.~L.}\ \bibnamefont {Rhodes}}, \bibinfo {author} {\bibfnamefont {M.}~\bibnamefont {Kreshchuk}}, \ and\ \bibinfo {author} {\bibfnamefont {S.}~\bibnamefont {Pathak}},\ }\href {\doibase 10.1103/PRXQuantum.5.040347} {\bibfield  {journal} {\bibinfo  {journal} {PRX Quantum}\ }\textbf {\bibinfo {volume} {5}},\ \bibinfo {pages} {040347} (\bibinfo {year} {2024})},\ \Eprint {http://arxiv.org/abs/2405.10416} {arXiv:2405.10416 [quant-ph]} \BibitemShut {NoStop}%
\bibitem [{\citenamefont {Du}\ and\ \citenamefont {Vary}(2025)}]{Du:2024ixj}%
  \BibitemOpen
  \bibfield  {author} {\bibinfo {author} {\bibfnamefont {W.}~\bibnamefont {Du}}\ and\ \bibinfo {author} {\bibfnamefont {J.~P.}\ \bibnamefont {Vary}},\ }\href {\doibase 10.1103/PhysRevD.111.016013} {\bibfield  {journal} {\bibinfo  {journal} {Phys. Rev. D}\ }\textbf {\bibinfo {volume} {111}},\ \bibinfo {pages} {016013} (\bibinfo {year} {2025})},\ \Eprint {http://arxiv.org/abs/2407.13672} {arXiv:2407.13672 [quant-ph]} \BibitemShut {NoStop}%
\bibitem [{\citenamefont {Li}\ \emph {et~al.}(2024)\citenamefont {Li}, \citenamefont {Grabowska},\ and\ \citenamefont {Savage}}]{Li:2024lrl}%
  \BibitemOpen
  \bibfield  {author} {\bibinfo {author} {\bibfnamefont {Z.}~\bibnamefont {Li}}, \bibinfo {author} {\bibfnamefont {D.~M.}\ \bibnamefont {Grabowska}}, \ and\ \bibinfo {author} {\bibfnamefont {M.~J.}\ \bibnamefont {Savage}},\ }\href@noop {} {\  (\bibinfo {year} {2024})},\ \Eprint {http://arxiv.org/abs/2407.13835} {arXiv:2407.13835 [quant-ph]} \BibitemShut {NoStop}%
\bibitem [{\citenamefont {Kane}\ \emph {et~al.}(2025)\citenamefont {Kane}, \citenamefont {Hariprakash}, \citenamefont {Modi}, \citenamefont {Kreshchuk},\ and\ \citenamefont {Bauer}}]{Kane:2024odt}%
  \BibitemOpen
  \bibfield  {author} {\bibinfo {author} {\bibfnamefont {C.~F.}\ \bibnamefont {Kane}}, \bibinfo {author} {\bibfnamefont {S.}~\bibnamefont {Hariprakash}}, \bibinfo {author} {\bibfnamefont {N.~S.}\ \bibnamefont {Modi}}, \bibinfo {author} {\bibfnamefont {M.}~\bibnamefont {Kreshchuk}}, \ and\ \bibinfo {author} {\bibfnamefont {C.~W.}\ \bibnamefont {Bauer}},\ }\href {\doibase 10.22331/q-2025-05-15-1747} {\bibfield  {journal} {\bibinfo  {journal} {Quantum}\ }\textbf {\bibinfo {volume} {9}},\ \bibinfo {pages} {1747} (\bibinfo {year} {2025})},\ \Eprint {http://arxiv.org/abs/2408.16824} {arXiv:2408.16824 [quant-ph]} \BibitemShut {NoStop}%
\bibitem [{\citenamefont {Murairi}\ \emph {et~al.}(2024)\citenamefont {Murairi}, \citenamefont {Sohaib~Alam}, \citenamefont {Lamm}, \citenamefont {Hadfield},\ and\ \citenamefont {Gustafson}}]{Murairi:2024xpc}%
  \BibitemOpen
  \bibfield  {author} {\bibinfo {author} {\bibfnamefont {E.~M.}\ \bibnamefont {Murairi}}, \bibinfo {author} {\bibfnamefont {M.}~\bibnamefont {Sohaib~Alam}}, \bibinfo {author} {\bibfnamefont {H.}~\bibnamefont {Lamm}}, \bibinfo {author} {\bibfnamefont {S.}~\bibnamefont {Hadfield}}, \ and\ \bibinfo {author} {\bibfnamefont {E.}~\bibnamefont {Gustafson}},\ }\href {\doibase 10.1103/PhysRevD.110.074501} {\bibfield  {journal} {\bibinfo  {journal} {Phys. Rev. D}\ }\textbf {\bibinfo {volume} {110}},\ \bibinfo {pages} {074501} (\bibinfo {year} {2024})},\ \Eprint {http://arxiv.org/abs/2408.00075} {arXiv:2408.00075 [quant-ph]} \BibitemShut {NoStop}%
\bibitem [{\citenamefont {Assi}\ and\ \citenamefont {Lamm}(2024)}]{Assi:2024pdn}%
  \BibitemOpen
  \bibfield  {author} {\bibinfo {author} {\bibfnamefont {B.}~\bibnamefont {Assi}}\ and\ \bibinfo {author} {\bibfnamefont {H.}~\bibnamefont {Lamm}},\ }\href {\doibase 10.1103/PhysRevD.110.074511} {\bibfield  {journal} {\bibinfo  {journal} {Phys. Rev. D}\ }\textbf {\bibinfo {volume} {110}},\ \bibinfo {pages} {074511} (\bibinfo {year} {2024})},\ \Eprint {http://arxiv.org/abs/2405.12204} {arXiv:2405.12204 [hep-lat]} \BibitemShut {NoStop}%
\bibitem [{\citenamefont {Lamm}\ \emph {et~al.}(2024)\citenamefont {Lamm}, \citenamefont {Li}, \citenamefont {Shu}, \citenamefont {Wang},\ and\ \citenamefont {Xu}}]{Lamm:2024jnl}%
  \BibitemOpen
  \bibfield  {author} {\bibinfo {author} {\bibfnamefont {H.}~\bibnamefont {Lamm}}, \bibinfo {author} {\bibfnamefont {Y.-Y.}\ \bibnamefont {Li}}, \bibinfo {author} {\bibfnamefont {J.}~\bibnamefont {Shu}}, \bibinfo {author} {\bibfnamefont {Y.-L.}\ \bibnamefont {Wang}}, \ and\ \bibinfo {author} {\bibfnamefont {B.}~\bibnamefont {Xu}},\ }\href {\doibase 10.1103/PhysRevD.110.054505} {\bibfield  {journal} {\bibinfo  {journal} {Phys. Rev. D}\ }\textbf {\bibinfo {volume} {110}},\ \bibinfo {pages} {054505} (\bibinfo {year} {2024})},\ \Eprint {http://arxiv.org/abs/2405.12890} {arXiv:2405.12890 [hep-lat]} \BibitemShut {NoStop}%
\bibitem [{\citenamefont {Gomes}\ \emph {et~al.}(2024)\citenamefont {Gomes}, \citenamefont {Lim},\ and\ \citenamefont {Wiebe}}]{Gomes:2024tup}%
  \BibitemOpen
  \bibfield  {author} {\bibinfo {author} {\bibfnamefont {N.}~\bibnamefont {Gomes}}, \bibinfo {author} {\bibfnamefont {H.}~\bibnamefont {Lim}}, \ and\ \bibinfo {author} {\bibfnamefont {N.}~\bibnamefont {Wiebe}},\ }\href@noop {} {\  (\bibinfo {year} {2024})},\ \Eprint {http://arxiv.org/abs/2408.03254} {arXiv:2408.03254 [quant-ph]} \BibitemShut {NoStop}%
\bibitem [{\citenamefont {Hardy}\ \emph {et~al.}(2024)\citenamefont {Hardy} \emph {et~al.}}]{Hardy:2024ric}%
  \BibitemOpen
  \bibfield  {author} {\bibinfo {author} {\bibfnamefont {A.}~\bibnamefont {Hardy}} \emph {et~al.},\ }\href@noop {} {\  (\bibinfo {year} {2024})},\ \Eprint {http://arxiv.org/abs/2407.13819} {arXiv:2407.13819 [quant-ph]} \BibitemShut {NoStop}%
\bibitem [{\citenamefont {Simon}\ \emph {et~al.}(2025)\citenamefont {Simon}, \citenamefont {Gustin}, \citenamefont {Serafin}, \citenamefont {Ralli}, \citenamefont {Goldstein},\ and\ \citenamefont {Love}}]{Simon:2025pbo}%
  \BibitemOpen
  \bibfield  {author} {\bibinfo {author} {\bibfnamefont {W.~A.}\ \bibnamefont {Simon}}, \bibinfo {author} {\bibfnamefont {C.~M.}\ \bibnamefont {Gustin}}, \bibinfo {author} {\bibfnamefont {K.}~\bibnamefont {Serafin}}, \bibinfo {author} {\bibfnamefont {A.}~\bibnamefont {Ralli}}, \bibinfo {author} {\bibfnamefont {G.~R.}\ \bibnamefont {Goldstein}}, \ and\ \bibinfo {author} {\bibfnamefont {P.~J.}\ \bibnamefont {Love}},\ }\href@noop {} {\  (\bibinfo {year} {2025})},\ \Eprint {http://arxiv.org/abs/2503.11641} {arXiv:2503.11641 [quant-ph]} \BibitemShut {NoStop}%
\bibitem [{\citenamefont {Decker}\ \emph {et~al.}(2025)\citenamefont {Decker} \emph {et~al.}}]{Decker:2025jts}%
  \BibitemOpen
  \bibfield  {author} {\bibinfo {author} {\bibfnamefont {E.}~\bibnamefont {Decker}} \emph {et~al.},\ }\href@noop {} {\  (\bibinfo {year} {2025})},\ \Eprint {http://arxiv.org/abs/2504.07214} {arXiv:2504.07214 [quant-ph]} \BibitemShut {NoStop}%
\bibitem [{\citenamefont {Davoudi}\ \emph {et~al.}(2025{\natexlab{b}})\citenamefont {Davoudi}, \citenamefont {Hsieh},\ and\ \citenamefont {Kadam}}]{Davoudi:2025rdv}%
  \BibitemOpen
  \bibfield  {author} {\bibinfo {author} {\bibfnamefont {Z.}~\bibnamefont {Davoudi}}, \bibinfo {author} {\bibfnamefont {C.-C.}\ \bibnamefont {Hsieh}}, \ and\ \bibinfo {author} {\bibfnamefont {S.~V.}\ \bibnamefont {Kadam}},\ }\href@noop {} {\  (\bibinfo {year} {2025}{\natexlab{b}})},\ \Eprint {http://arxiv.org/abs/2505.20408} {arXiv:2505.20408 [quant-ph]} \BibitemShut {NoStop}%
\bibitem [{\citenamefont {Pichler}\ \emph {et~al.}(2016)\citenamefont {Pichler}, \citenamefont {Dalmonte}, \citenamefont {Rico}, \citenamefont {Zoller},\ and\ \citenamefont {Montangero}}]{Pichler:2015yqa}%
  \BibitemOpen
  \bibfield  {author} {\bibinfo {author} {\bibfnamefont {T.}~\bibnamefont {Pichler}}, \bibinfo {author} {\bibfnamefont {M.}~\bibnamefont {Dalmonte}}, \bibinfo {author} {\bibfnamefont {E.}~\bibnamefont {Rico}}, \bibinfo {author} {\bibfnamefont {P.}~\bibnamefont {Zoller}}, \ and\ \bibinfo {author} {\bibfnamefont {S.}~\bibnamefont {Montangero}},\ }\href {\doibase 10.1103/PhysRevX.6.011023} {\bibfield  {journal} {\bibinfo  {journal} {Phys. Rev. X}\ }\textbf {\bibinfo {volume} {6}},\ \bibinfo {pages} {011023} (\bibinfo {year} {2016})},\ \Eprint {http://arxiv.org/abs/1505.04440} {arXiv:1505.04440 [cond-mat.quant-gas]} \BibitemShut {NoStop}%
\bibitem [{\citenamefont {Ba\~nuls}\ \emph {et~al.}(2017)\citenamefont {Ba\~nuls}, \citenamefont {Cichy}, \citenamefont {Cirac}, \citenamefont {Jansen},\ and\ \citenamefont {K\"uhn}}]{Banuls:2017ena}%
  \BibitemOpen
  \bibfield  {author} {\bibinfo {author} {\bibfnamefont {M.~C.}\ \bibnamefont {Ba\~nuls}}, \bibinfo {author} {\bibfnamefont {K.}~\bibnamefont {Cichy}}, \bibinfo {author} {\bibfnamefont {J.~I.}\ \bibnamefont {Cirac}}, \bibinfo {author} {\bibfnamefont {K.}~\bibnamefont {Jansen}}, \ and\ \bibinfo {author} {\bibfnamefont {S.}~\bibnamefont {K\"uhn}},\ }\href {\doibase 10.1103/PhysRevX.7.041046} {\bibfield  {journal} {\bibinfo  {journal} {Phys. Rev. X}\ }\textbf {\bibinfo {volume} {7}},\ \bibinfo {pages} {041046} (\bibinfo {year} {2017})},\ \Eprint {http://arxiv.org/abs/1707.06434} {arXiv:1707.06434 [hep-lat]} \BibitemShut {NoStop}%
\bibitem [{\citenamefont {Ba\~nuls}\ \emph {et~al.}(2018)\citenamefont {Ba\~nuls}, \citenamefont {Cichy}, \citenamefont {Cirac}, \citenamefont {Jansen},\ and\ \citenamefont {K\"uhn}}]{Banuls:2018jag}%
  \BibitemOpen
  \bibfield  {author} {\bibinfo {author} {\bibfnamefont {M.~C.}\ \bibnamefont {Ba\~nuls}}, \bibinfo {author} {\bibfnamefont {K.}~\bibnamefont {Cichy}}, \bibinfo {author} {\bibfnamefont {J.~I.}\ \bibnamefont {Cirac}}, \bibinfo {author} {\bibfnamefont {K.}~\bibnamefont {Jansen}}, \ and\ \bibinfo {author} {\bibfnamefont {S.}~\bibnamefont {K\"uhn}},\ }\href {\doibase 10.22323/1.334.0022} {\bibfield  {journal} {\bibinfo  {journal} {PoS}\ }\textbf {\bibinfo {volume} {LATTICE2018}},\ \bibinfo {pages} {022} (\bibinfo {year} {2018})},\ \Eprint {http://arxiv.org/abs/1810.12838} {arXiv:1810.12838 [hep-lat]} \BibitemShut {NoStop}%
\bibitem [{\citenamefont {Ba\~nuls}(2023)}]{Banuls:2022vxp}%
  \BibitemOpen
  \bibfield  {author} {\bibinfo {author} {\bibfnamefont {M.~C.}\ \bibnamefont {Ba\~nuls}},\ }\href {\doibase 10.1146/annurev-conmatphys-040721-022705} {\bibfield  {journal} {\bibinfo  {journal} {Ann. Rev. Condensed Matter Phys.}\ }\textbf {\bibinfo {volume} {14}},\ \bibinfo {pages} {173} (\bibinfo {year} {2023})},\ \Eprint {http://arxiv.org/abs/2205.10345} {arXiv:2205.10345 [quant-ph]} \BibitemShut {NoStop}%
\bibitem [{\citenamefont {Mathew}\ \emph {et~al.}(2025)\citenamefont {Mathew}, \citenamefont {Gupta}, \citenamefont {Kadam}, \citenamefont {Bapat}, \citenamefont {Stryker}, \citenamefont {Davoudi},\ and\ \citenamefont {Raychowdhury}}]{Mathew:2025fim}%
  \BibitemOpen
  \bibfield  {author} {\bibinfo {author} {\bibfnamefont {E.}~\bibnamefont {Mathew}}, \bibinfo {author} {\bibfnamefont {N.}~\bibnamefont {Gupta}}, \bibinfo {author} {\bibfnamefont {S.~V.}\ \bibnamefont {Kadam}}, \bibinfo {author} {\bibfnamefont {A.}~\bibnamefont {Bapat}}, \bibinfo {author} {\bibfnamefont {J.}~\bibnamefont {Stryker}}, \bibinfo {author} {\bibfnamefont {Z.}~\bibnamefont {Davoudi}}, \ and\ \bibinfo {author} {\bibfnamefont {I.}~\bibnamefont {Raychowdhury}},\ }\href {\doibase 10.22323/1.466.0472} {\bibfield  {journal} {\bibinfo  {journal} {PoS}\ }\textbf {\bibinfo {volume} {LATTICE2024}},\ \bibinfo {pages} {472} (\bibinfo {year} {2025})},\ \Eprint {http://arxiv.org/abs/2501.18301} {arXiv:2501.18301 [hep-lat]} \BibitemShut {NoStop}%
\bibitem [{\citenamefont {Belyansky}\ \emph {et~al.}(2024)\citenamefont {Belyansky}, \citenamefont {Whitsitt}, \citenamefont {Mueller}, \citenamefont {Fahimniya}, \citenamefont {Bennewitz}, \citenamefont {Davoudi},\ and\ \citenamefont {Gorshkov}}]{Belyansky:2023rgh}%
  \BibitemOpen
  \bibfield  {author} {\bibinfo {author} {\bibfnamefont {R.}~\bibnamefont {Belyansky}}, \bibinfo {author} {\bibfnamefont {S.}~\bibnamefont {Whitsitt}}, \bibinfo {author} {\bibfnamefont {N.}~\bibnamefont {Mueller}}, \bibinfo {author} {\bibfnamefont {A.}~\bibnamefont {Fahimniya}}, \bibinfo {author} {\bibfnamefont {E.~R.}\ \bibnamefont {Bennewitz}}, \bibinfo {author} {\bibfnamefont {Z.}~\bibnamefont {Davoudi}}, \ and\ \bibinfo {author} {\bibfnamefont {A.~V.}\ \bibnamefont {Gorshkov}},\ }\href {\doibase 10.1103/PhysRevLett.132.091903} {\bibfield  {journal} {\bibinfo  {journal} {Phys. Rev. Lett.}\ }\textbf {\bibinfo {volume} {132}},\ \bibinfo {pages} {091903} (\bibinfo {year} {2024})},\ \Eprint {http://arxiv.org/abs/2307.02522} {arXiv:2307.02522 [quant-ph]} \BibitemShut {NoStop}%
\bibitem [{\citenamefont {Clemente}\ \emph {et~al.}(2022)\citenamefont {Clemente}, \citenamefont {Crippa},\ and\ \citenamefont {Jansen}}]{Clemente:2022cka}%
  \BibitemOpen
  \bibfield  {author} {\bibinfo {author} {\bibfnamefont {G.}~\bibnamefont {Clemente}}, \bibinfo {author} {\bibfnamefont {A.}~\bibnamefont {Crippa}}, \ and\ \bibinfo {author} {\bibfnamefont {K.}~\bibnamefont {Jansen}},\ }\href {\doibase 10.1103/PhysRevD.106.114511} {\bibfield  {journal} {\bibinfo  {journal} {Phys. Rev. D}\ }\textbf {\bibinfo {volume} {106}},\ \bibinfo {pages} {114511} (\bibinfo {year} {2022})},\ \Eprint {http://arxiv.org/abs/2206.12454} {arXiv:2206.12454 [hep-lat]} \BibitemShut {NoStop}%
\bibitem [{\citenamefont {Crippa}\ \emph {et~al.}(2024)\citenamefont {Crippa}, \citenamefont {Romiti}, \citenamefont {Funcke}, \citenamefont {Jansen}, \citenamefont {K\"uhn}, \citenamefont {Stornati},\ and\ \citenamefont {Urbach}}]{Crippa:2024cqr}%
  \BibitemOpen
  \bibfield  {author} {\bibinfo {author} {\bibfnamefont {A.}~\bibnamefont {Crippa}}, \bibinfo {author} {\bibfnamefont {S.}~\bibnamefont {Romiti}}, \bibinfo {author} {\bibfnamefont {L.}~\bibnamefont {Funcke}}, \bibinfo {author} {\bibfnamefont {K.}~\bibnamefont {Jansen}}, \bibinfo {author} {\bibfnamefont {S.}~\bibnamefont {K\"uhn}}, \bibinfo {author} {\bibfnamefont {P.}~\bibnamefont {Stornati}}, \ and\ \bibinfo {author} {\bibfnamefont {C.}~\bibnamefont {Urbach}},\ }\href@noop {} {\  (\bibinfo {year} {2024})},\ \Eprint {http://arxiv.org/abs/2404.17545} {arXiv:2404.17545 [hep-lat]} \BibitemShut {NoStop}%
\bibitem [{\citenamefont {Carena}\ \emph {et~al.}(2021)\citenamefont {Carena}, \citenamefont {Lamm}, \citenamefont {Li},\ and\ \citenamefont {Liu}}]{Carena:2021ltu}%
  \BibitemOpen
  \bibfield  {author} {\bibinfo {author} {\bibfnamefont {M.}~\bibnamefont {Carena}}, \bibinfo {author} {\bibfnamefont {H.}~\bibnamefont {Lamm}}, \bibinfo {author} {\bibfnamefont {Y.-Y.}\ \bibnamefont {Li}}, \ and\ \bibinfo {author} {\bibfnamefont {W.}~\bibnamefont {Liu}},\ }\href {\doibase 10.1103/PhysRevD.104.094519} {\bibfield  {journal} {\bibinfo  {journal} {Phys. Rev. D}\ }\textbf {\bibinfo {volume} {104}},\ \bibinfo {pages} {094519} (\bibinfo {year} {2021})},\ \Eprint {http://arxiv.org/abs/2107.01166} {arXiv:2107.01166 [hep-lat]} \BibitemShut {NoStop}%
\bibitem [{\citenamefont {Funcke}\ \emph {et~al.}(2023)\citenamefont {Funcke}, \citenamefont {Gro\ss{}}, \citenamefont {Jansen}, \citenamefont {K\"uhn}, \citenamefont {Romiti},\ and\ \citenamefont {Urbach}}]{Funcke:2022opx}%
  \BibitemOpen
  \bibfield  {author} {\bibinfo {author} {\bibfnamefont {L.}~\bibnamefont {Funcke}}, \bibinfo {author} {\bibfnamefont {C.~F.}\ \bibnamefont {Gro\ss{}}}, \bibinfo {author} {\bibfnamefont {K.}~\bibnamefont {Jansen}}, \bibinfo {author} {\bibfnamefont {S.}~\bibnamefont {K\"uhn}}, \bibinfo {author} {\bibfnamefont {S.}~\bibnamefont {Romiti}}, \ and\ \bibinfo {author} {\bibfnamefont {C.}~\bibnamefont {Urbach}},\ }\href {\doibase 10.22323/1.430.0292} {\bibfield  {journal} {\bibinfo  {journal} {PoS}\ }\textbf {\bibinfo {volume} {LATTICE2022}},\ \bibinfo {pages} {292} (\bibinfo {year} {2023})},\ \Eprint {http://arxiv.org/abs/2212.09627} {arXiv:2212.09627 [hep-lat]} \BibitemShut {NoStop}%
\bibitem [{\citenamefont {Gro\ss{}}\ \emph {et~al.}(2025)\citenamefont {Gro\ss{}}, \citenamefont {Romiti}, \citenamefont {Funcke}, \citenamefont {Jansen}, \citenamefont {Kan}, \citenamefont {K\"uhn},\ and\ \citenamefont {Urbach}}]{Gross:2025qae}%
  \BibitemOpen
  \bibfield  {author} {\bibinfo {author} {\bibfnamefont {C.~F.}\ \bibnamefont {Gro\ss{}}}, \bibinfo {author} {\bibfnamefont {S.}~\bibnamefont {Romiti}}, \bibinfo {author} {\bibfnamefont {L.}~\bibnamefont {Funcke}}, \bibinfo {author} {\bibfnamefont {K.}~\bibnamefont {Jansen}}, \bibinfo {author} {\bibfnamefont {A.}~\bibnamefont {Kan}}, \bibinfo {author} {\bibfnamefont {S.}~\bibnamefont {K\"uhn}}, \ and\ \bibinfo {author} {\bibfnamefont {C.}~\bibnamefont {Urbach}},\ }\href@noop {} {\  (\bibinfo {year} {2025})},\ \Eprint {http://arxiv.org/abs/2503.11480} {arXiv:2503.11480 [hep-lat]} \BibitemShut {NoStop}%
\bibitem [{\citenamefont {Carena}\ \emph {et~al.}(2022{\natexlab{b}})\citenamefont {Carena}, \citenamefont {Gustafson}, \citenamefont {Lamm}, \citenamefont {Li},\ and\ \citenamefont {Liu}}]{Carena:2022hpz}%
  \BibitemOpen
  \bibfield  {author} {\bibinfo {author} {\bibfnamefont {M.}~\bibnamefont {Carena}}, \bibinfo {author} {\bibfnamefont {E.~J.}\ \bibnamefont {Gustafson}}, \bibinfo {author} {\bibfnamefont {H.}~\bibnamefont {Lamm}}, \bibinfo {author} {\bibfnamefont {Y.-Y.}\ \bibnamefont {Li}}, \ and\ \bibinfo {author} {\bibfnamefont {W.}~\bibnamefont {Liu}},\ }\href {\doibase 10.1103/PhysRevD.106.114504} {\bibfield  {journal} {\bibinfo  {journal} {Phys. Rev. D}\ }\textbf {\bibinfo {volume} {106}},\ \bibinfo {pages} {114504} (\bibinfo {year} {2022}{\natexlab{b}})},\ \Eprint {http://arxiv.org/abs/2208.10417} {arXiv:2208.10417 [hep-lat]} \BibitemShut {NoStop}%
\bibitem [{\citenamefont {Childs}\ and\ \citenamefont {Wiebe}(2012)}]{Childs:2012gwh}%
  \BibitemOpen
  \bibfield  {author} {\bibinfo {author} {\bibfnamefont {A.~M.}\ \bibnamefont {Childs}}\ and\ \bibinfo {author} {\bibfnamefont {N.}~\bibnamefont {Wiebe}},\ }\href {\doibase 10.26421/QIC12.11-12-1} {\bibfield  {journal} {\bibinfo  {journal} {Quant. Inf. Comput.}\ }\textbf {\bibinfo {volume} {12}},\ \bibinfo {pages} {0901} (\bibinfo {year} {2012})},\ \Eprint {http://arxiv.org/abs/1202.5822} {arXiv:1202.5822 [quant-ph]} \BibitemShut {NoStop}%
\bibitem [{\citenamefont {Low}\ and\ \citenamefont {Chuang}(2017)}]{Low:2016sck}%
  \BibitemOpen
  \bibfield  {author} {\bibinfo {author} {\bibfnamefont {G.~H.}\ \bibnamefont {Low}}\ and\ \bibinfo {author} {\bibfnamefont {I.~L.}\ \bibnamefont {Chuang}},\ }\href {\doibase 10.1103/PhysRevLett.118.010501} {\bibfield  {journal} {\bibinfo  {journal} {Phys. Rev. Lett.}\ }\textbf {\bibinfo {volume} {118}},\ \bibinfo {pages} {010501} (\bibinfo {year} {2017})},\ \Eprint {http://arxiv.org/abs/1606.02685} {arXiv:1606.02685 [quant-ph]} \BibitemShut {NoStop}%
\bibitem [{\citenamefont {Low}\ and\ \citenamefont {Chuang}(2019)}]{Low:2016znh}%
  \BibitemOpen
  \bibfield  {author} {\bibinfo {author} {\bibfnamefont {G.~H.}\ \bibnamefont {Low}}\ and\ \bibinfo {author} {\bibfnamefont {I.~L.}\ \bibnamefont {Chuang}},\ }\href {\doibase 10.22331/q-2019-07-12-163} {\bibfield  {journal} {\bibinfo  {journal} {Quantum}\ }\textbf {\bibinfo {volume} {3}},\ \bibinfo {pages} {163} (\bibinfo {year} {2019})},\ \Eprint {http://arxiv.org/abs/1610.06546} {arXiv:1610.06546 [quant-ph]} \BibitemShut {NoStop}%
\bibitem [{\citenamefont {Campbell}(2019)}]{PhysRevLett.123.070503}%
  \BibitemOpen
  \bibfield  {author} {\bibinfo {author} {\bibfnamefont {E.}~\bibnamefont {Campbell}},\ }\href {\doibase 10.1103/PhysRevLett.123.070503} {\bibfield  {journal} {\bibinfo  {journal} {Phys. Rev. Lett.}\ }\textbf {\bibinfo {volume} {123}},\ \bibinfo {pages} {070503} (\bibinfo {year} {2019})}\BibitemShut {NoStop}%
\bibitem [{\citenamefont {Li}\ and\ \citenamefont {Benjamin}(2017)}]{Li:2016vmf}%
  \BibitemOpen
  \bibfield  {author} {\bibinfo {author} {\bibfnamefont {Y.}~\bibnamefont {Li}}\ and\ \bibinfo {author} {\bibfnamefont {S.~C.}\ \bibnamefont {Benjamin}},\ }\href {\doibase 10.1103/physrevx.7.021050} {\bibfield  {journal} {\bibinfo  {journal} {Phys. Rev. X}\ }\textbf {\bibinfo {volume} {7}},\ \bibinfo {pages} {021050} (\bibinfo {year} {2017})},\ \Eprint {http://arxiv.org/abs/1611.09301} {arXiv:1611.09301 [quant-ph]} \BibitemShut {NoStop}%
\bibitem [{\citenamefont {Haah}\ \emph {et~al.}(2021)\citenamefont {Haah}, \citenamefont {Hastings}, \citenamefont {Kothari},\ and\ \citenamefont {Low}}]{Haah:2018ekc}%
  \BibitemOpen
  \bibfield  {author} {\bibinfo {author} {\bibfnamefont {J.}~\bibnamefont {Haah}}, \bibinfo {author} {\bibfnamefont {M.~B.}\ \bibnamefont {Hastings}}, \bibinfo {author} {\bibfnamefont {R.}~\bibnamefont {Kothari}}, \ and\ \bibinfo {author} {\bibfnamefont {G.~H.}\ \bibnamefont {Low}},\ }\href {\doibase 10.1137/18M1231511} {\bibfield  {journal} {\bibinfo  {journal} {SIAM J. Comput.}\ }\textbf {\bibinfo {volume} {52}},\ \bibinfo {pages} {FOCS18} (\bibinfo {year} {2021})},\ \Eprint {http://arxiv.org/abs/1801.03922} {arXiv:1801.03922 [quant-ph]} \BibitemShut {NoStop}%
\bibitem [{\citenamefont {Yuan}\ \emph {et~al.}(2019)\citenamefont {Yuan}, \citenamefont {Endo}, \citenamefont {Zhao}, \citenamefont {Li},\ and\ \citenamefont {Benjamin}}]{Yuan:2018jdl}%
  \BibitemOpen
  \bibfield  {author} {\bibinfo {author} {\bibfnamefont {X.}~\bibnamefont {Yuan}}, \bibinfo {author} {\bibfnamefont {S.}~\bibnamefont {Endo}}, \bibinfo {author} {\bibfnamefont {Q.}~\bibnamefont {Zhao}}, \bibinfo {author} {\bibfnamefont {Y.}~\bibnamefont {Li}}, \ and\ \bibinfo {author} {\bibfnamefont {S.}~\bibnamefont {Benjamin}},\ }\href {\doibase 10.22331/q-2019-10-07-191} {\bibfield  {journal} {\bibinfo  {journal} {Quantum}\ }\textbf {\bibinfo {volume} {3}},\ \bibinfo {pages} {191} (\bibinfo {year} {2019})},\ \Eprint {http://arxiv.org/abs/1812.08767} {arXiv:1812.08767 [quant-ph]} \BibitemShut {NoStop}%
\bibitem [{\citenamefont {Motlagh}\ and\ \citenamefont {Wiebe}(2024)}]{Motlagh:2023oqc}%
  \BibitemOpen
  \bibfield  {author} {\bibinfo {author} {\bibfnamefont {D.}~\bibnamefont {Motlagh}}\ and\ \bibinfo {author} {\bibfnamefont {N.}~\bibnamefont {Wiebe}},\ }\href {\doibase 10.1103/PRXQuantum.5.020368} {\bibfield  {journal} {\bibinfo  {journal} {PRX Quantum}\ }\textbf {\bibinfo {volume} {5}},\ \bibinfo {pages} {020368} (\bibinfo {year} {2024})},\ \Eprint {http://arxiv.org/abs/2308.01501} {arXiv:2308.01501 [quant-ph]} \BibitemShut {NoStop}%
\bibitem [{\citenamefont {Zeng}\ \emph {et~al.}(2025)\citenamefont {Zeng}, \citenamefont {Sun}, \citenamefont {Jiang},\ and\ \citenamefont {Zhao}}]{Zeng:2022pim}%
  \BibitemOpen
  \bibfield  {author} {\bibinfo {author} {\bibfnamefont {P.}~\bibnamefont {Zeng}}, \bibinfo {author} {\bibfnamefont {J.}~\bibnamefont {Sun}}, \bibinfo {author} {\bibfnamefont {L.}~\bibnamefont {Jiang}}, \ and\ \bibinfo {author} {\bibfnamefont {Q.}~\bibnamefont {Zhao}},\ }\href {\doibase 10.1103/PRXQuantum.6.010359} {\bibfield  {journal} {\bibinfo  {journal} {PRX Quantum}\ }\textbf {\bibinfo {volume} {6}},\ \bibinfo {pages} {010359} (\bibinfo {year} {2025})},\ \Eprint {http://arxiv.org/abs/2212.04566} {arXiv:2212.04566 [quant-ph]} \BibitemShut {NoStop}%
\bibitem [{\citenamefont {Watson}\ and\ \citenamefont {Watkins}(2024)}]{Watson:2024yqs}%
  \BibitemOpen
  \bibfield  {author} {\bibinfo {author} {\bibfnamefont {J.~D.}\ \bibnamefont {Watson}}\ and\ \bibinfo {author} {\bibfnamefont {J.}~\bibnamefont {Watkins}},\ }\href@noop {} {\  (\bibinfo {year} {2024})},\ \Eprint {http://arxiv.org/abs/2408.14385} {arXiv:2408.14385 [quant-ph]} \BibitemShut {NoStop}%
\bibitem [{\citenamefont {Watson}(2024)}]{Watson:2024dvw}%
  \BibitemOpen
  \bibfield  {author} {\bibinfo {author} {\bibfnamefont {J.~D.}\ \bibnamefont {Watson}},\ }\href@noop {} {\  (\bibinfo {year} {2024})},\ \Eprint {http://arxiv.org/abs/2411.04240} {arXiv:2411.04240 [quant-ph]} \BibitemShut {NoStop}%
\bibitem [{\citenamefont {Chakraborty}\ \emph {et~al.}(2025)\citenamefont {Chakraborty}, \citenamefont {Hazra}, \citenamefont {Li}, \citenamefont {Shao}, \citenamefont {Wang},\ and\ \citenamefont {Zhang}}]{Chakraborty:2025sry}%
  \BibitemOpen
  \bibfield  {author} {\bibinfo {author} {\bibfnamefont {S.}~\bibnamefont {Chakraborty}}, \bibinfo {author} {\bibfnamefont {S.}~\bibnamefont {Hazra}}, \bibinfo {author} {\bibfnamefont {T.}~\bibnamefont {Li}}, \bibinfo {author} {\bibfnamefont {C.}~\bibnamefont {Shao}}, \bibinfo {author} {\bibfnamefont {X.}~\bibnamefont {Wang}}, \ and\ \bibinfo {author} {\bibfnamefont {Y.}~\bibnamefont {Zhang}},\ }\href@noop {} {\  (\bibinfo {year} {2025})},\ \Eprint {http://arxiv.org/abs/2504.02385} {arXiv:2504.02385 [quant-ph]} \BibitemShut {NoStop}%
\bibitem [{\citenamefont {Gily\'en}\ \emph {et~al.}(2018)\citenamefont {Gily\'en}, \citenamefont {Su}, \citenamefont {Low},\ and\ \citenamefont {Wiebe}}]{Gilyen:2018khw}%
  \BibitemOpen
  \bibfield  {author} {\bibinfo {author} {\bibfnamefont {A.}~\bibnamefont {Gily\'en}}, \bibinfo {author} {\bibfnamefont {Y.}~\bibnamefont {Su}}, \bibinfo {author} {\bibfnamefont {G.~H.}\ \bibnamefont {Low}}, \ and\ \bibinfo {author} {\bibfnamefont {N.}~\bibnamefont {Wiebe}},\ }in\ \href {\doibase 10.1145/3313276.3316366} {\emph {\bibinfo {booktitle} {{51st Annual ACM SIGACT Symposium on Theory of Computing}}}}\ (\bibinfo {year} {2018})\ \Eprint {http://arxiv.org/abs/1806.01838} {arXiv:1806.01838 [quant-ph]} \BibitemShut {NoStop}%
\bibitem [{\citenamefont {Creutz}(1983)}]{Creutz:1983njd}%
  \BibitemOpen
  \bibfield  {author} {\bibinfo {author} {\bibfnamefont {M.}~\bibnamefont {Creutz}},\ }\href {\doibase 10.1017/9781009290395} {\emph {\bibinfo {title} {{Quarks, Gluons and Lattices}}}}\ (\bibinfo  {publisher} {Oxford University Press},\ \bibinfo {year} {1983})\BibitemShut {NoStop}%
\bibitem [{\citenamefont {Hoshina}\ \emph {et~al.}(2020)\citenamefont {Hoshina}, \citenamefont {Fujii},\ and\ \citenamefont {Kikukawa}}]{Hoshina:2020gdy}%
  \BibitemOpen
  \bibfield  {author} {\bibinfo {author} {\bibfnamefont {H.}~\bibnamefont {Hoshina}}, \bibinfo {author} {\bibfnamefont {H.}~\bibnamefont {Fujii}}, \ and\ \bibinfo {author} {\bibfnamefont {Y.}~\bibnamefont {Kikukawa}},\ }\href {\doibase 10.22323/1.363.0190} {\bibfield  {journal} {\bibinfo  {journal} {PoS}\ }\textbf {\bibinfo {volume} {LATTICE2019}},\ \bibinfo {pages} {190} (\bibinfo {year} {2020})}\BibitemShut {NoStop}%
\bibitem [{\citenamefont {Kanwar}\ and\ \citenamefont {Wagman}(2021)}]{Kanwar:2021tkd}%
  \BibitemOpen
  \bibfield  {author} {\bibinfo {author} {\bibfnamefont {G.}~\bibnamefont {Kanwar}}\ and\ \bibinfo {author} {\bibfnamefont {M.~L.}\ \bibnamefont {Wagman}},\ }\href {\doibase 10.1103/PhysRevD.104.014513} {\bibfield  {journal} {\bibinfo  {journal} {Phys. Rev. D}\ }\textbf {\bibinfo {volume} {104}},\ \bibinfo {pages} {014513} (\bibinfo {year} {2021})},\ \Eprint {http://arxiv.org/abs/2103.02602} {arXiv:2103.02602 [hep-lat]} \BibitemShut {NoStop}%
\bibitem [{\citenamefont {Menotti}\ and\ \citenamefont {Onofri}(1981)}]{Menotti:1981ry}%
  \BibitemOpen
  \bibfield  {author} {\bibinfo {author} {\bibfnamefont {P.}~\bibnamefont {Menotti}}\ and\ \bibinfo {author} {\bibfnamefont {E.}~\bibnamefont {Onofri}},\ }\href {\doibase 10.1016/0550-3213(81)90560-5} {\bibfield  {journal} {\bibinfo  {journal} {Nucl. Phys. B}\ }\textbf {\bibinfo {volume} {190}},\ \bibinfo {pages} {288} (\bibinfo {year} {1981})}\BibitemShut {NoStop}%
\bibitem [{\citenamefont {Childs}\ and\ \citenamefont {Su}(2019)}]{PhysRevLett.123.050503}%
  \BibitemOpen
  \bibfield  {author} {\bibinfo {author} {\bibfnamefont {A.~M.}\ \bibnamefont {Childs}}\ and\ \bibinfo {author} {\bibfnamefont {Y.}~\bibnamefont {Su}},\ }\href {\doibase 10.1103/PhysRevLett.123.050503} {\bibfield  {journal} {\bibinfo  {journal} {Phys. Rev. Lett.}\ }\textbf {\bibinfo {volume} {123}},\ \bibinfo {pages} {050503} (\bibinfo {year} {2019})}\BibitemShut {NoStop}%
\bibitem [{\citenamefont {Heyl}\ \emph {et~al.}(2019)\citenamefont {Heyl}, \citenamefont {Hauke},\ and\ \citenamefont {Zoller}}]{doi:10.1126/sciadv.aau8342}%
  \BibitemOpen
  \bibfield  {author} {\bibinfo {author} {\bibfnamefont {M.}~\bibnamefont {Heyl}}, \bibinfo {author} {\bibfnamefont {P.}~\bibnamefont {Hauke}}, \ and\ \bibinfo {author} {\bibfnamefont {P.}~\bibnamefont {Zoller}},\ }\href {\doibase 10.1126/sciadv.aau8342} {\bibfield  {journal} {\bibinfo  {journal} {Science Advances}\ }\textbf {\bibinfo {volume} {5}},\ \bibinfo {pages} {eaau8342} (\bibinfo {year} {2019})},\ \Eprint {http://arxiv.org/abs/https://www.science.org/doi/pdf/10.1126/sciadv.aau8342} {https://www.science.org/doi/pdf/10.1126/sciadv.aau8342} \BibitemShut {NoStop}%
\bibitem [{\citenamefont {Childs}\ \emph {et~al.}(2021)\citenamefont {Childs}, \citenamefont {Su}, \citenamefont {Tran}, \citenamefont {Wiebe},\ and\ \citenamefont {Zhu}}]{Childs:2019hts}%
  \BibitemOpen
  \bibfield  {author} {\bibinfo {author} {\bibfnamefont {A.~M.}\ \bibnamefont {Childs}}, \bibinfo {author} {\bibfnamefont {Y.}~\bibnamefont {Su}}, \bibinfo {author} {\bibfnamefont {M.~C.}\ \bibnamefont {Tran}}, \bibinfo {author} {\bibfnamefont {N.}~\bibnamefont {Wiebe}}, \ and\ \bibinfo {author} {\bibfnamefont {S.}~\bibnamefont {Zhu}},\ }\href {\doibase 10.1103/PhysRevX.11.011020} {\bibfield  {journal} {\bibinfo  {journal} {Phys. Rev. X}\ }\textbf {\bibinfo {volume} {11}},\ \bibinfo {pages} {011020} (\bibinfo {year} {2021})},\ \Eprint {http://arxiv.org/abs/1912.08854} {arXiv:1912.08854 [quant-ph]} \BibitemShut {NoStop}%
\bibitem [{\citenamefont {Suchsland}\ \emph {et~al.}(2025)\citenamefont {Suchsland}, \citenamefont {Moessner},\ and\ \citenamefont {Claeys}}]{Suchsland:2023cmb}%
  \BibitemOpen
  \bibfield  {author} {\bibinfo {author} {\bibfnamefont {P.}~\bibnamefont {Suchsland}}, \bibinfo {author} {\bibfnamefont {R.}~\bibnamefont {Moessner}}, \ and\ \bibinfo {author} {\bibfnamefont {P.~W.}\ \bibnamefont {Claeys}},\ }\href {\doibase 10.1103/PhysRevB.111.014309} {\bibfield  {journal} {\bibinfo  {journal} {Phys. Rev. B}\ }\textbf {\bibinfo {volume} {111}},\ \bibinfo {pages} {014309} (\bibinfo {year} {2025})},\ \Eprint {http://arxiv.org/abs/2308.03851} {arXiv:2308.03851 [quant-ph]} \BibitemShut {NoStop}%
\bibitem [{\citenamefont {Berry}\ \emph {et~al.}(2015)\citenamefont {Berry}, \citenamefont {Childs},\ and\ \citenamefont {Kothari}}]{Berry:2015hst}%
  \BibitemOpen
  \bibfield  {author} {\bibinfo {author} {\bibfnamefont {D.~W.}\ \bibnamefont {Berry}}, \bibinfo {author} {\bibfnamefont {A.~M.}\ \bibnamefont {Childs}}, \ and\ \bibinfo {author} {\bibfnamefont {R.}~\bibnamefont {Kothari}}\ }(\bibinfo {year} {2015})\ \Eprint {http://arxiv.org/abs/1501.01715} {arXiv:1501.01715 [quant-ph]} \BibitemShut {NoStop}%
\bibitem [{\citenamefont {Umeda}\ \emph {et~al.}(2005)\citenamefont {Umeda}, \citenamefont {Nomura},\ and\ \citenamefont {Matsufuru}}]{Umeda:2002vr}%
  \BibitemOpen
  \bibfield  {author} {\bibinfo {author} {\bibfnamefont {T.}~\bibnamefont {Umeda}}, \bibinfo {author} {\bibfnamefont {K.}~\bibnamefont {Nomura}}, \ and\ \bibinfo {author} {\bibfnamefont {H.}~\bibnamefont {Matsufuru}},\ }\href {\doibase 10.1140/epjcd/s2004-01-002-1} {\bibfield  {journal} {\bibinfo  {journal} {Eur. Phys. J. C}\ }\textbf {\bibinfo {volume} {39S1}},\ \bibinfo {pages} {9} (\bibinfo {year} {2005})},\ \Eprint {http://arxiv.org/abs/hep-lat/0211003} {arXiv:hep-lat/0211003} \BibitemShut {NoStop}%
\bibitem [{\citenamefont {Nomura}\ \emph {et~al.}(2004)\citenamefont {Nomura}, \citenamefont {Umeda},\ and\ \citenamefont {Matsufuru}}]{Nomura:2003in}%
  \BibitemOpen
  \bibfield  {author} {\bibinfo {author} {\bibfnamefont {K.}~\bibnamefont {Nomura}}, \bibinfo {author} {\bibfnamefont {T.}~\bibnamefont {Umeda}}, \ and\ \bibinfo {author} {\bibfnamefont {H.}~\bibnamefont {Matsufuru}},\ }\href {\doibase 10.1016/S0920-5632(03)02591-X} {\bibfield  {journal} {\bibinfo  {journal} {Nucl. Phys. B Proc. Suppl.}\ }\textbf {\bibinfo {volume} {129}},\ \bibinfo {pages} {390} (\bibinfo {year} {2004})},\ \Eprint {http://arxiv.org/abs/hep-lat/0312010} {arXiv:hep-lat/0312010} \BibitemShut {NoStop}%
\bibitem [{\citenamefont {Aarts}\ \emph {et~al.}(2015)\citenamefont {Aarts}, \citenamefont {Allton}, \citenamefont {Amato}, \citenamefont {Giudice}, \citenamefont {Hands},\ and\ \citenamefont {Skullerud}}]{Aarts:2014nba}%
  \BibitemOpen
  \bibfield  {author} {\bibinfo {author} {\bibfnamefont {G.}~\bibnamefont {Aarts}}, \bibinfo {author} {\bibfnamefont {C.}~\bibnamefont {Allton}}, \bibinfo {author} {\bibfnamefont {A.}~\bibnamefont {Amato}}, \bibinfo {author} {\bibfnamefont {P.}~\bibnamefont {Giudice}}, \bibinfo {author} {\bibfnamefont {S.}~\bibnamefont {Hands}}, \ and\ \bibinfo {author} {\bibfnamefont {J.-I.}\ \bibnamefont {Skullerud}},\ }\href {\doibase 10.1007/JHEP02(2015)186} {\bibfield  {journal} {\bibinfo  {journal} {JHEP}\ }\textbf {\bibinfo {volume} {02}},\ \bibinfo {pages} {186} (\bibinfo {year} {2015})},\ \Eprint {http://arxiv.org/abs/1412.6411} {arXiv:1412.6411 [hep-lat]} \BibitemShut {NoStop}%
\bibitem [{\citenamefont {Klco}\ and\ \citenamefont {Savage}(2019)}]{Klco:2018zqz}%
  \BibitemOpen
  \bibfield  {author} {\bibinfo {author} {\bibfnamefont {N.}~\bibnamefont {Klco}}\ and\ \bibinfo {author} {\bibfnamefont {M.~J.}\ \bibnamefont {Savage}},\ }\href {\doibase 10.1103/PhysRevA.99.052335} {\bibfield  {journal} {\bibinfo  {journal} {Phys. Rev. A}\ }\textbf {\bibinfo {volume} {99}},\ \bibinfo {pages} {052335} (\bibinfo {year} {2019})},\ \Eprint {http://arxiv.org/abs/1808.10378} {arXiv:1808.10378 [quant-ph]} \BibitemShut {NoStop}%
\end{thebibliography}%

\appendix

\clearpage
\newpage
\onecolumngrid

\section{Scale setting and the Kogut-Susskind Hamiltonian}
\label{app:scale_set_and_HKS}

In this Appendix, we provide a pedagogical discussion regarding the need for scale setting in lattice gauge theory simulations. 
After expressing the Kogut-Susskind Hamiltonian in terms of dimensionless variables, we comment on the bare parameters of the theory, and show that one needs to include an additional parameter in the fermionic kinetic term.

To show that one must determine the lattice spacing by a scale setting procedure, we will use the following two observations:
\begin{enumerate}
    \item Computers (classical and quantum) can only work with dimensionless variables, \emph{i.e.} all inputs into any numerical simulation are dimensionless.
    \item The value of the lattice spacing $a$ is not independent of the bare parameters in the Hamiltonian, but is related through renormalization group equations.
\end{enumerate}
To isolate the logic of these two steps, we first consider the simple case of the harmonic oscillator and show that no such scale setting procedure is necessary.
While this example may appear trivial, it highlights the key difference between standard quantum mechanical Hamiltonians and lattice field theory Hamiltonians where the lattice is used as a UV regulator.
We then consider a U(1) LGT and show that the lattice spacing must be determined through a scale setting procedure.

The dimensionful harmonic oscillator Hamiltonian (with $\hbar=1$) is given by
\begin{equation}
    H = \omega (b^\dagger b + \frac{1}{2})\,,
\end{equation}
where $b$ and $b^\dagger$ are lowering and raising operators that satisfy $[b, b^\dagger] = 1$.
Aside from the Hamiltonian, the only dimensionful parameter is the frequency $\omega$.
If we define $\omega = \hat \omega/a$ where $a$ is a dimensionful parameter with constant magnitude (with the symbol $a$ being suggestively chosen for the LGT case), and $\hat \omega$ is a dimensionless parameter with magnitude $a \omega$, then the dimensionless Hamiltonian that we actually simulate is 
\begin{equation}
    \hat H \equiv a H = \hat \omega (a^\dagger a + \frac{1}{2})\,.
\end{equation}
Note that the parameter $a$ acts as a physical scale of the system.
The above equations shows that one actually calculates the energies $\hat E_n \equiv a E_n$ of the rescaled Hamiltonian $\hat{H} \ket{n} = \hat E_n \ket{n}$.
Converting to the physical energies of $H$ in this case can be done trivially by dividing by the known parameter $a$
\begin{equation}
    E_n = \hat E_n/a\,.
\end{equation}
Indeed, this step is so obvious and intuitive that it is usually done implicitly by simply setting the magnitude of $a=1$, which implies one can directly interpret the dimensionless energies $\hat E_n$ as the physical energies $E_n$. 

Consider now the more complicated scenario where the parameter $a = a(\hat \omega)$ is an \emph{unknown} function of $\hat \omega$ (or, equivalently, a function of $\omega$).
In order to convert $\widetilde{E}_n$ to the physical energies $E_n$, one must somehow determine the scale $a(\hat \omega)$.
One possibility is to perform a \emph{scale setting} procedure, where one determines $a(\hat \omega)$ by demanding the calculated dimensionless energy $\hat{E}_{n^*} = a(\hat \omega) E_{n^*}$ is equal to the experimentally known value $E_{n^*}^{{\rm phys}}$, 
\begin{equation}
    a(\hat \omega) \equiv \frac{[a(\hat \omega) E_{n^*}]}{E_{n^*}^{\rm phys}}\,.
\end{equation}
As we will now discuss, this is analogous the case for LGTs, with the dimesionless bare parameters playing the role of $\hat \omega$ and the lattice spacing playing the role of the dimensionful scale $a(\hat \omega)$.

To demonstrate that one does not actually choose the lattice spacing in a LGT simulation, it is sufficient to consider a U(1) LGT with a single quark flavor; this analysis generalizes to SU(N) LGTs with additional quark flavors.
Here we only give the level of detail required to demonstrate the need to scale set, further details can be found in, \emph{e.g.}, Ref.~\cite{DiMeglio:2023nsa}.
In $d\geq 2$ dimensions, the Hamiltonian is a sum of 4 terms
\begin{equation}
    H = H_E + H_B + H_M + H_K\,,
\end{equation}
where $H_E$ and $H_B$ are the electric and magnetic Hamiltonians governing the dynamics of the gauge fields, $H_M$ is the fermionic mass Hamiltonian, and $H_K$ is the fermionic kinetic (hopping) Hamiltonian.
Because no symmetry demands it, the coefficients multiplying each of these 4 terms are in general different, which implies we will have 4 bare parameters in the theory; we discuss the need for a 4th bare parameter multiplying the fermionic kinetic term in more detail later in this appendix.
The individual Hamiltonian terms are
\begin{align}
    H_E &= \frac{\widetilde{g}_t(a)^2}{2 a^{d-2}} \sum_{\vec{x}}\sum_{j=1}^d E(\vec{x},j)^2\,,
    \\
    H_B &= -\frac{1}{2 a^{4-d} \widetilde{g}_s(a)^2} \sum_{p} (P_p+ P_p^\dagger)\,,
    \\
    H_M &= m(a) \sum_{\vec{x}} (-1)^{\vec{x}} \psi^\dagger(\vec{x}) \psi(\vec{x})\,,
    \\
    H_K &= \frac{\kappa(a)}{a}\sum_{\vec{x}}\sum_{j=1}^d \eta(\vec{x})\left(\psi^\dagger(\vec{x}) U(\vec{x},j) \psi(\vec{x}+\hat{e}_j) + {\rm h.c.} \right)\,,
\end{align}
where $a$ is the lattice spacing, $\hat E(\vec{x},j)$ and is an electric field operator at site $\vec{x}$ pointing in the direction of the $\hat e_j$ axis, $\hat U(\vec{x},j)$ is a gauge link operator at site $\vec{x}$ pointing in the direction of the $\hat e_j$ axis, $\hat \psi(\vec{x}) (\hat \psi^\dagger (\vec{x}))$ is the fermonic annihilation (creation) operator, and $\eta(\vec{x})$ is the staggered sign factor.
The magnetic Hamiltonian is the sum over all independent plaquettes $p$ on the lattice; a plaquette originating at site $\vec{x}$ oriented in the $j,k$ plane is given by
\begin{equation}
    P_p = U(\vec{x}, j) U(\vec{x}+\hat{e}_j, k) U^\dagger(\vec{x}+\hat{e}_k, j) U^\dagger(\vec{x}, k)\,.
\end{equation}

Analogous to the simple harmonic oscillator case, we want to write a dimensionless Hamiltonian $\hat H$ in terms of dimensionless quantities.
We write $[A] = M^\gamma$ to denote that the operator/parameter $A$ has mass dimension $\gamma$.
By definition, we have $[H] = M$, $[a] = M^{-1}$, and $[U(\vec{x},j)] = M^0$, which implies the fermionic operators and electric field operators are dimensionless, \emph{i.e.}, $[\psi]=M^0, [E(\vec{x},j)]=M^0$, and the bare parameters have dimensions
\begin{equation}
    [\widetilde{g}_t(a)]=M^\frac{3-d}{2}, \quad [\widetilde{g}_s(a)]=M^\frac{3-d}{2}, \quad [m(a)] = M^1, \quad [\kappa(a)]=M^0\,.
\end{equation}
The parameters $\widetilde{g}_t(a), \widetilde{g}_s(a)$ and $m(a)$ can be expressed in terms of the dimensionless bare parameters 
\begin{equation}
    g_t(a)^2 = a^{3-d} \widetilde g_t(a)^2\,, \qquad  g_s(a)^2 = a^{3-d} \widetilde g_s(a)^2\,, \qquad \overline{m}(a) = a m(a)\,. 
\end{equation}
Rather than $g_t$ and $g_s$, we reparameterize (as done in the main text) using $g^2 = g_s g_t$ and $c = g_t/g_s$.
In this way, the Hamiltonian can be written as
\begin{equation}
    H_{\rm KS} = \frac{c}{a} \hat H_{\rm KS} = \frac{c}{a} \left(\hat H_E + \hat H_B + \hat H_M + \hat H_K \right)
\end{equation}
where the individual rescaled terms are given by
\begin{align}
    \hat H_E &= \frac{g(a)^2}{2} \sum_{\vec{x}}\sum_{j=1}^d \hat E(\vec{x},j)^2
    \\
    \hat H_B &= -\frac{1}{2 g(a)^2} \sum_{p} (\hat P_p+\hat P_p^\dagger)
    \\
    \hat H_M &= \hat m(a) \sum_{\vec{x}} (-1)^{\vec{x}} \hat \psi^\dagger(\vec{x}) \hat \psi(\vec{x})
    \\
    \hat H_K &= \hat \kappa(a)\sum_{\vec{x}}\sum_{j=1}^d \eta(\vec{x})\left(\hat \psi^\dagger(\vec{x}) \hat U(\vec{x},j) \hat \psi(\vec{x}+\hat{e}_j) + {\rm h.c.} \right)\,.
\end{align}
and the rescaled bare parameters are $\hat m(a) = \frac{a}{c}m(a)$ and $\hat{\kappa}(a)=\kappa(a)/c$.
Note that the bare speed of light appears in the denominator of the expressions in  Ref.~\cite{Carena:2021ltu} because the incorrect definition of $g_s/g_t$ was used.
Written in this way, it is clear that $a_t = a/c$ is merely an overall scale factor of the Hamiltonian, and the dimensionless Hamiltonian we actually simulate is $\hat H_{\rm KS}$.
From this we see that one can only calculate dimensionless energies $\hat E_n = a_t E_n$.
Because $a_t = a_t(g, c, \hat m, \hat \kappa)$ is an unknown function of the coupling and quark mass, one must determine it through a scale setting procedure, analogous to the previous simple harmonic oscillator example.
It is important to note that this logic applies not only to energies, but to any quantity calculated on the lattice.
For example, if the operator $O$ has units $[O] = M^b$, then the operator one actually implements in a simulation is $\hat O = a^{b} O$, and the physical result is calculated via $O = [\hat O]/a^b$, where the lattice spacing $a$ is assumed to have been already determined via some scale setting procedure.

We now discuss the appearance of the bare parameter $\kappa(a)$ in the fermionic kinetic term, which is analogous to the need for different gauge couplings $g_t$ and $g_s$ multiplying $H_E$ and $H_B$.
The intuition for including $\kappa(a)$ can be drawn from studying lattice gauge theory actions on anisotropic Euclidean lattices, \emph{i.e.}, the spatial lattice spacing $a$ is different from the temporal lattice spacing $a_0$.
For pure gauge theory, this anisotropy requires using different gauge couplings multiplying purely spatial plaquettes and plaqeuttes oriented in the temporal direction; one can equivalently work with a common gauge coupling $g$ and a gauge anisotropy factor $\gamma_g$. 
For fixed $g$, one can tune $\gamma_g$ to determine a renormalization trajectory with fixed physical anisotropy $\xi = a/a_0$.
By carefully keeping track of these couplings, one can use the transfer matrix formalism to show the need for two independent gauge couplings in the Euclidean theory implies one also needs two different couplings, $g_t$ and $g_s$, in the pure gauge Kogut-Susskind Hamiltonian. 

In a similar way, as discussed in, \emph{e.g.}, Refs.~\cite{Umeda:2002vr, Nomura:2003in, Aarts:2014nba}, anisotropic simulations including fermions introduces a bare fermionic anisotropy factor $\gamma_f$, controlling the relative size of the fermionic hopping terms in the spatial and Euclidean-time directions. 
The fermionic anisotropy $\gamma_f$ is generally different from the gauge anisotropy $\gamma_g$, and both parameters need to be tuned such that the renormalized anisotropy $\xi = a/a_0$ ``seen'' by fermionic and gauge degrees of freedom are equal; this can be done by, \emph{e.g.}, studying the dispersion relation of gluonic states and fermionic states.
Analogous to the pure gauge case, by keeping track of the anisotropy factor in the transfer matrix derivation, one finds that the prefactor in front of the fermionic kinetic term $\kappa(a)$ in the Kogut-Susskind Hamiltonian is different from unity and needs to be included in the renormalization procedure.

\section{Comparing numerical tightness of PF and QSP error bounds}
\label{app:comparing_bounds}

In this section, we compare the numerical tightness of the PF and QSP error bounds used to prove \cref{thm:pf_E,thm:qsp_E} respectively for the specific choice of a single site anharmonic oscillator with the following Hamiltonian:
\begin{align}
    H = \frac{\phi^2}{2} + \frac{32\phi^4}{34} + \frac{\pi^2}{2} := H_\phi + H_\pi\,,
\end{align}
where $\phi$ denotes the field operator and $\pi$ denotes the conjugate momentum. Given a choice for number of qubits $n_q$ and hence the size of the Hilbert space ($2^{n_q}\times2^{n_q}$), we discretize $H$ using the scheme presented for instance in Refs.~\cite{Klco:2018zqz,Bauer:2021gek,Hariprakash:2023tla}. Let $N_{\mathrm{sim}}$ be used to denote either of the algorithmic parameters $\npf$ and $N_{\mathrm{QSP}}$. Thus, for a desired exponentiation error $\|\Delta_{\mathrm{sim}}\|$ (see \cref{eq:delta_sim}), and simulation time $t$ we can use \cref{eq:trotter_error_bound,eq:qsp_error_bound} to analytically compute sufficient choices $N_{\mathrm{sim}}$ for both PF and QSP such that the desired exponentiation error is achieved. We then empirically determine the lowest possible value of $N_{\mathrm{sim}}$, denoted by $N_{\mathrm{sim}}^{(\mathrm{empirical})}$ that results in the target exponentiation error $\|\Delta_{\mathrm{sim}}\|$. \cref{fig:ratio_vs_nq} compares the ratio $N_{\mathrm{sim}}^{(\mathrm{empirical})}/N_{\mathrm{sim}}$ for PF and QSP as a function of the number of qubits $n_q$ used to represent $H$ for the fixed values $\|\Delta_{\mathrm{sim}}\| = 10^{-4}$ and $t=1$. The key observation we make is that while both ratios deviate from the ideal value of $1$ as we increase $n_q$, the QSP error bound is closer to the numerically determined result than the PF bound. This numerical experiment is thus a simple example illustrating that the QSP error bounds are generally tighter than the PF error bounds used in this work.

\begin{figure}
    \centering
    \includegraphics[width=0.5\linewidth]{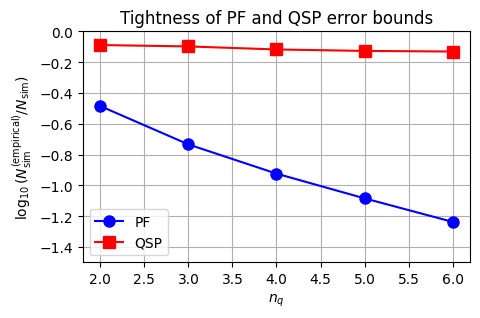}
    \caption{Comparison of the ratios of the analytically determined algorithmic parameters $N_{\mathrm{sim}}$ with its numerically optimized counterpart $N_{\mathrm{sim}}^{(\mathrm{empirical})}$ for PF and QSP at a fixed target exponentiation error $\|\Delta_{\mathrm{sim}}\| = 10^{-4}$ and simulation time $t=1$.}
    \label{fig:ratio_vs_nq}
\end{figure}

\end{document}